\documentclass[a4paper,USenglish,cleveref,numberwithinsect,autoref,thm-restate]{lipics-v2021}
\usepackage[utf8]{inputenc}
\usepackage{color}
\usepackage{csquotes}
\usepackage{xspace}
\usepackage{amsmath}
\usepackage{amssymb}
\usepackage{mathtools}
\usepackage{bbm}
\usepackage[linesnumbered, noend]{algorithm2e}
\usepackage[draft]{fixme}
\usepackage{tikz}
\usetikzlibrary{arrows,arrows.meta,decorations.pathmorphing}
\usepackage{bbding}
\usepackage{xcolor,pifont}
\usepackage{multicol}
\nolinenumbers
\allowdisplaybreaks

\fxsetup{mode=multiuser,theme=color, layout=inline}
\FXRegisterAuthor{eg}{eg}{\color{red} {\bf Egor}}
\FXRegisterAuthor{tk}{atk}{\color{blue} {\bf Tomasz}}
\FXRegisterAuthor{ib}{ib}{\color{green} {\bf Itai}}

\newtheorem{problem}[theorem]{Problem}
\newtheorem{fact}[theorem]{Fact}
\crefname{fact}{Fact}{Facts}

\crefname{claim}{Claim}{Claims}
\newcommand{\abs}[1]{\ensuremath{\lvert #1 \rvert}}

\newcommand{\floor}[1]{\lfloor{#1}\rfloor}
\newcommand{\ceil}[1]{\lceil{#1}\rceil}
\newcommand{\set}[1]{\{#1\}}
\newcommand{\Oh}{\ensuremath{\mathcal{O}}\xspace}
\newcommand{\Ohtilde}{\ensuremath{\smash{\rlap{\raisebox{-0.2ex}{$\widetilde{\phantom{\Oh}}$}}\Oh}}\xspace}
\newcommand{\tOh}{\Ohtilde}

\newcommand{\RR}{\mathbb{R}\xspace}
\newcommand{\Rz}{\mathbb{R}_{\ge 0}}
\newcommand{\ZZ}{\mathbb{Z}\xspace}
\renewcommand{\emptyset}{\varnothing}
\newcommand{\emptystring}{\varepsilon}

\DeclareMathOperator{\poly}{poly}
\DeclareMathOperator{\polylog}{polylog}

\newcommand{\dd}{.\,.}
\newcommand{\fragmentco}[2]{[#1\dd #2)}
\newcommand{\fragmentoc}[2]{(#1\dd #2]}
\newcommand{\fragmentoo}[2]{(#1\dd #2)}
\newcommand{\fragmentcc}[2]{[#1\dd #2]}
\newcommand{\position}[1]{[#1]}
\newcommand{\row}[3]{#1_{#2},\ldots,#1_{#3}}
\newcommand{\setof}[2]{\set{#1\colon\,#2}}

\newcommand{\w}[2]{w}
\newcommand{\maybe}{}

\newcommand{\says}[3]{}

\newcommand{\Esigma}{\overline{\Sigma}}
\newcommand{\wdef}{\ensuremath{w \colon \Esigma^2 \to [0, W]}\xspace}
\newcommand{\dist}{\mathsf{dist}}
\newcommand{\ed}{\mathsf{ed}}
\newcommand{\eed}[2]{\mathsf{ed}(#1,#2)}
\newcommand{\wed}{\ed^w}
\newcommand{\wwed}[3]{\mathsf{ed}^{#1}(#2,#3)}
\newcommand{\hd}{\mathsf{hd}}
\newcommand{\hhd}[2]{\mathsf{hd}(#1,#2)}
\newcommand{\cost}{\mathsf{cost}}
\newcommand{\sed}{\mathsf{self}\text{-}\ed}
\newcommand{\sedk}{\sed^k}
\newcommand{\AG}{\mathsf{AG}}
\newcommand{\AGw}{\AG^w}

\newcommand{\oAGw}{\overline{\AG}^w}

\newcommand{\BMw}{\mathsf{BM}^w}
\newcommand{\meq}[1]{\stackrel{#1}{=}}

\newcommand{\cA}{\mathcal{A}}
\newcommand{\cB}{\mathcal{B}}

\newcommand{\oG}{\overline{G}}
\newcommand{\hX}{\hat{X}}
\newcommand{\hY}{\hat{Y}}
\newcommand{\hk}{\hat{k}}

\newcommand{\Sem}{\mathcal{S}}

\newcommand{\absorbingmatrix}{\mathsf{z}}
\newcommand{\matmon}{\mathcal{M}}
\newcommand{\xotimes}{\operatorname{\hat{\otimes}}}
\newcommand{\localds}{\mathcal{D}}
\newcommand{\localdscor}{\overline{\localds}}
\newcommand{\fancysymbol}{\dagger}

\newcommand{\para}[1]{\subparagraph*{#1}}

\newcommand{\calX}{\mathcal{X}}

\def\ipmOpName{{\tt IPM}\xspace}
\def\accOpName{{\tt Access}\xspace}
\def\extractOpName{{\tt Extract}\xspace}
\def\lenOpName{{\tt Length}\xspace}

\def\lceOpName{{\tt LCP}\xspace}

\newcommand{\modelname}{\texttt{PILLAR}\xspace}
\newcommand{\pillar}{\modelname}
\def\lceOp#1#2{{\tt LCP}(#1, #2)}
\def\lcbOp#1#2{{\tt LCP}^R(#1, #2)}

\def\accOp#1#2{#1\position{#2}}

\definecolor{darkgreen}{RGB}{0,160,0}
\definecolor{darkred}{RGB}{220,20,60}
\definecolor{darkblue}{RGB}{0,0,160}

\def\twoheadleadsto{\tikz[baseline=(a.base)]{\draw[%
    decorate,decoration={zigzag,segment length=4, amplitude=.9},%
    ] (0,0) -- (.25, 0);%
    \draw[%
    -{Classical TikZ Rightarrow}.{Classical TikZ Rightarrow},%
    ] (.25, 0) -- (.4, 0);%
    \node (a) at (.4/2,-.03) {\phantom{\(\leadsto\)}};%
}}

\newcommand{\onto}{\twoheadleadsto}
\def\aonto#1{\onto}

\newlength{\leftstackrelawd}
\newlength{\leftstackrelbwd}
\def\leftstackrel#1#2{\settowidth{\leftstackrelawd}%
{${{}^{#1}}$}\settowidth{\leftstackrelbwd}{$#2$}%
\addtolength{\leftstackrelawd}{-\leftstackrelbwd}%
\leavevmode\ifthenelse{\lengthtest{\leftstackrelawd>0pt}}%
{\kern-.5\leftstackrelawd}{}\mathrel{\mathop{#2}\limits^{#1}}}

\hideLIPIcs

\title{Bounded Weighted Edit Distance}
\subtitle{Dynamic Algorithms and Matching Lower Bounds}
\titlerunning{Bounded Weighted Edit Distance: Dynamic Algorithms and Matching Lower Bounds}

\author{Itai Boneh}{Reichman University and University of Haifa, Israel}{itai.bone@biu.ac.il}{https://orcid.org/0009-0007-8895-4069}{supported by Israel Science Foundation grant 810/21.}

\author{Egor Gorbachev}{Saarland University and Max Planck Institute for Informatics, Saarland Informatics Campus, Germany}{egorbachev@cs.uni-saarland.de}{https://orcid.org/0009-0005-5977-7986}{This work is part of the project TIPEA that has received funding from the European Research Council (ERC) under the European Unions Horizon 2020 research and innovation programme (grant agreement No.\ 850979).}

\author{Tomasz Kociumaka}{Max Planck Institute for Informatics, Saarland Informatics
Campus, Saarbrücken, Germany}{tomasz.kociumaka@mpi-inf.mpg.de}{https://orcid.org/0000-0002-2477-1702}{}

\authorrunning{I. Boneh, E. Gorbachev, and T. Kociumaka} 

\Copyright{Itai Boneh, Egor Gorbachev, and Tomasz Kociumaka} 



\acknowledgements{}

\begin{document}

\maketitle

\begin{abstract}
    The edit distance $\ed(X,Y)$ of two strings $X,Y\in \Sigma^*$ is the minimum number of character edits (insertions, deletions, and substitutions) needed to transform $X$ into $Y$. 
Its weighted counterpart~$\wed(X,Y)$ minimizes the total cost of edits, where the costs of individual edits, depending on the edit type and the characters involved, are specified using a function $w$, normalized so that each edit costs at least one.
The textbook dynamic-programming procedure, given strings $X,Y\in \Sigma^{\le n}$ and oracle access to $w$, computes $\wed(X,Y)$ in $\Oh(n^2)$ time.
Nevertheless, one can achieve better running times if the computed distance, denoted $k$, is small:
$\Oh(n+k^2)$ for unit weights [Landau and Vishkin; JCSS'88] and $\Ohtilde(n+\sqrt{nk^3})$\footnote{Henceforth, the $\Ohtilde(\cdot)$ notation hides factors poly-logarithmic in $n$.} for arbitrary weights [Cassis, Kociumaka, Wellnitz; FOCS'23].

In this paper, we study the dynamic version of the weighted edit distance problem, where the goal is to maintain $\wed(X,Y)$ for strings $X,Y\in \Sigma^{\le n}$ that change over time, with each update specified as an edit in $X$ or~$Y$.
Very recently, Gorbachev and Kociumaka [STOC'25] showed that the \emph{unweighted} distance $\ed(X,Y)$ can be maintained in $\Ohtilde(k)$ time per update after $\Ohtilde(n+k^2)$-time preprocessing; here, $k$ denotes the \emph{current} value of $\ed(X,Y)$.
Their algorithm generalizes to small integer weights, but the underlying approach is incompatible with large weights.

Our main result is a dynamic algorithm that maintains $\wed(X,Y)$ in $\Ohtilde(k^{3-\gamma})$ time per update after $\Ohtilde(nk^\gamma)$-time preprocessing.
Here, $\gamma\in [0,1]$ is a real trade-off parameter and $k\ge 1$ is an integer threshold \emph{fixed} at preprocessing time, with $\infty$ returned whenever $\wed(X,Y)>k$.
We complement our algorithm with conditional lower bounds showing fine-grained optimality of our trade-off for~$\gamma \in [0.5,1)$ and justifying our choice to fix $k$.

We also generalize our solution to a much more robust setting while preserving the fine-grained optimal trade-off.
Our full algorithm maintains $X\in \Sigma^{\le n}$ subject not only to character edits but also substring deletions and copy-pastes, each supported in $\Ohtilde(k^2)$ time.
Instead of dynamically maintaining $Y$, it answers queries that, given any string $Y$ specified through a sequence of $\Oh(k)$ arbitrary edits transforming $X$ into $Y$, in $\Ohtilde(k^{3-\gamma})$ time compute $\wed(X,Y)$ or report that $\wed(X,Y)>k$.

\end{abstract}

\section{Introduction}\label{sec:introduction}
Among the most fundamental string processing problems is the task of deciding whether two strings are similar to each other. 
A classic measure of string (dis)similarity is the \emph{edit distance}, also known as the \emph{Levenshtein distance}~\cite{Levenshtein66}.
The (unit-cost) edit distance $\ed(X,Y)$ of two strings $X$ and $Y$ is defined as the minimum number of character edits (insertions, deletions, and substitutions) needed to transform $X$ into~$Y$.
In typical use cases, edits model naturally occurring modifications, such as typographical errors in natural-language texts or mutations in biological sequences.
Some of these changes occur more frequently than others, which motivated introducing \emph{weighted edit distance} already in the 1970s~\cite{WF74}.
In this setting, each edit has a cost that may depend on its type and the characters involved (but nothing else), and the goal is to minimize the total cost of edits rather than their quantity.
The costs of individual edits can be specified through a weight function~$w \colon \Esigma^2 \to \mathbb{R}_{\ge 0}$, where~$\Esigma=\Sigma\cup\{\emptystring\}$ denotes the alphabet extended with a symbol $\emptystring$ representing the lack of a character.
For~$a,b\in \Sigma$, the cost of inserting $b$ is $w(\emptystring,b)$, the cost of deleting $a$ is $w(a,\emptystring)$, and the cost of substituting $a$ for $b$ is $w(a,b)$.
(Note that the cost of an edit is independent of the position of the edited character in the string.)
Consistently with previous works, we assume that $w$ is normalized: $w(a,b)\ge 1$ and $w(a,a)=0$ hold for every distinct $a,b\in \Esigma$.
The unweighted edit distance constitutes the special case when $w(a,b)=1$ holds for $a\ne b$.

The textbook dynamic-programming algorithm~\cite{Vin68,NW70,Sel74,WF74} takes $\Oh(n^2)$ time to compute the edit distance of two strings of length at most $n$, and some of its original formulations incorporate weights~\cite{Sel74,WF74,Sel80}.
Unfortunately, there is little hope for much faster solutions: 
for any constant $\delta>0$, an $\Oh(n^{2-\delta})$-time algorithm would violate the Orthogonal Vectors Hypothesis~\cite{ABW15,BK15,AHWW16,BI18}, and hence the Strong Exponential Time Hypothesis~\cite{IP01,IPZ01}. 
The lower bound holds already for unit weights and many other fixed weight functions~\cite{BK15}.

\subparagraph*{Bounded Edit Distance.}
On the positive side, the quadratic running time can be improved when the computed distance is small.
For unweighted edit distance, the algorithm by Landau and Vishkin~\cite{LV88}, building upon the ideas of \cite{Ukk85,Mye86}, computes $k\coloneqq \ed(X,Y)$ in $\Oh(n+k^2)$ time for any two strings $X,Y\in \Sigma^{\le n}$.
This running time is fine-grained optimal: for any~$\delta>0$, a hypothetical $\Oh(n+k^{2-\delta})$-time algorithm, even restricted to instances satisfying~$k=\Theta(n^{\kappa})$ for some constant $\kappa \in (0.5, 1]$, would violate the Orthogonal Vectors Hypothesis.

As far as the bounded \emph{weighted} edit distance problem is concerned, a simple optimization by Ukkonen \cite{Ukk85} improves the quadratic time complexity of the classic dynamic-programming algorithm to $\Oh(nk)$, where $k\coloneqq \wed(X,Y)$.
Recently, Das, Gilbert, Hajiaghayi, Kociumaka, and Saha~\cite{DGHKS23} developed an $\Oh(n+k^5)$-time solution and, soon afterwards, Cassis, Kociumaka, and Wellnitz~\cite{CKW23} presented an $\Ohtilde(n+\sqrt{nk^3})$-time algorithm. 
Due to the necessity to read the entire input, the latter running time is optimal for $k\le \sqrt[3]{n}$.
More surprisingly, there is a tight conditional lower bound for $\sqrt{n} \le k \le n$~\cite{CKW23}, valid already for edit costs in the real interval~$[1,2]$.
For any constants $\kappa \in [0.5, 1]$ and $\delta>0$, an $\Oh(\sqrt{nk^{3-\delta}})$-time algorithm restricted to instances satisfying $k=\Theta(n^{\kappa})$ would violate the All-Pairs Shortest Paths Hypothesis~\cite{VWW18}.
The optimal complexity remains unknown for $\sqrt[3]{n} < k < \sqrt{n}$, where the upper bound is $\Ohtilde(\sqrt{nk^3})$, yet the lower-bound construction allows for an $\Ohtilde(n+k^{2.5})$-time algorithm.

\subparagraph*{Dynamic Edit Distance.}
Over the last decades, the edit distance problem has been studied in numerous variants, including the dynamic ones, where the input strings change over time.
The most popular dynamic (weighed) edit distance formulation is the following one~\cite{CKM20,GK24}.

\begin{problem}[Dynamic Weighted Edit Distance]\label{prob:base}
    Given an integer $n$ and oracle access to a normalized weight function $w\colon \Esigma^2 \to \mathbb{R}_{\ge 0}$, maintain strings $X,Y\in \Sigma^{\le n}$ subject to updates (character edits in $X$ and $Y$) and report $\wed(X,Y)$ after each update.\footnote{For clarity, the introduction focuses on the version of the problem where the value $\wed(X, Y)$ is reported after every update. 
    In \cref{thm:full-simpler,thm:full-complete-algorithm}, we address a more general setting where queries may occur less frequently, allowing us to distinguish between update time and query time.}
\end{problem}

The time complexity of \cref{prob:base} is already well understood if it is measured solely in terms of the string length $n$.
In case of unit weights, the algorithm of Charalampopoulos, Kociumaka, and Mozes~\cite{CKM20} (building upon the techniques by Tiskin~\cite{Tis08}) is characterized by~$\Ohtilde(n)$ update time and $\Oh(n^2)$ preprocessing time; any polynomial-factor improvement would violate the lower bound for the static edit distance problem.
As far as arbitrary weights are concerned, the approach presented in \cite[Section 4]{CKM20} achieves $\Ohtilde(n\sqrt{n})$ update time after $\Ohtilde(n^2)$-time preprocessing, and there is a matching fine-grained lower bound~\cite[Section 6]{CKW23} conditioned on the All-Pairs Shortest Paths Hypothesis.

In this work, we aim to understand \cref{prob:base} for small distances:
\begin{quote}
    \centering{\textit{How fast can one dynamically maintain $\wed(X,Y)$ when this value is small?}}
\end{quote}

Very recently, Gorbachev and Kociumaka~\cite{GK24} settled the parameterized complexity of \cref{prob:base} in the \emph{unweighted} case.
Their solution achieves $\Ohtilde(n+k^2)$ preprocessing time and~$\Ohtilde(k)$ update time; improving either complexity by a polynomial factor, even for instances satisfying $k=\Theta(n^\kappa)$ for some $\kappa \in (0,1]$, would violate the Orthogonal Vectors Hypothesis.

Interestingly, the approach of~\cite{GK24} builds upon the static procedure for bounded \emph{weighted} edit distance~\cite{CKW23} and, as a result, the dynamic algorithm seamlessly supports \emph{small integer} weights,
achieving $\Oh(W^2k)$ update time for edit costs in $[1\dd W]$.
Unfortunately, \cite{GK24} does not provide any improvements for arbitrary weights; in that case, the results of~\cite{CKW23} yield a solution to \cref{prob:base} with $\Ohtilde(n)$ preprocessing time and $\Ohtilde(\min(k^3,\sqrt{nk^3}))$ update time.
This is far from satisfactory: for $k\ge \sqrt[3]{n}$, the algorithm simply recomputes $\wed(X,Y)$ from scratch after every update, and for $k< \sqrt[3]{n}$, it does not store (and reuse) anything beyond a dynamic data structure for checking the equality between fragments of $X$ and $Y$.

\paragraph*{Our Results: Lower Bounds}
The reason why the approach of~\cite{GK24} is incompatible with large weights, e.g., $W=\Omega(k)$, is that their presence allows a single update to drastically change $\wed(X,Y)$.
Our first result strengthens the lower bound of~\cite{CKW23} and shows that the $\Ohtilde(n+\sqrt{nk^3})$-time static algorithm is fine-grained optimal already for instances that can be transformed from a pair of equal strings using four edits.
The ideas behind the following theorem are presented in \cref{sec:lower-bounds}. 

\renewcommand{\maybe}{\lipicsEnd}
\begin{restatable}{theorem}{thmstaticlb}\label{thm:static-lb}
    Let $\kappa \in [0.5,1]$ and $\delta > 0$ be real parameters. 
    Assuming the APSP Hypothesis, there is no algorithm that, given two strings $X,Y\in \Sigma^{\le n}$ satisfying $\eed{X}{Y} \le 4$, a threshold $k \le n^\kappa$, 
    and oracle access to a normalized weight function $w\colon\Esigma^2\to \mathbb{R}_{\ge 0}$, 
    in time~$\Oh(n^{0.5+1.5\kappa - \delta})$ decides whether $\wwed{w}{X}{Y} \le k$.
\end{restatable}
\renewcommand{\maybe}{}

As a simple corollary, we conclude that significantly improving upon the na\"ive update time of $\Ohtilde(\sqrt{nk^3})$ for $\sqrt{n} \le k \le n$ requires large preprocessing time.

\begin{corollary}\label{cor:simple}
    Suppose that \cref{prob:base} admits a solution with preprocessing time $T_P(n, \wed(X,Y))$ and update time $T_U(n,\wed(X,Y))$ for some non-decreasing functions $T_P$ and $T_U$.
    If $T_P(n, 0)=\Oh(n^{0.5+1.5\kappa-\delta})$ and $T_U(n,n^\kappa)=\Oh(n^{0.5+1.5\kappa-\delta})$ hold for some real parameters $\kappa \in [0.5,1]$ and $\delta > 0$, then the APSP Hypothesis fails.
\end{corollary}
\begin{proof}[Proof idea.]
For an instance $(X,Y)$ of the static problem of \cref{thm:static-lb}, we initialize the dynamic algorithm with $(X,X)$ and transform one copy of $X$ into $Y$ using four updates.
\end{proof}

As discussed in \cref{app:dynamic-lower-bounds:variable-k}, instead of initializing the dynamic algorithm with $(X,X)$, we can initialize it with a pair of empty strings and gradually transform this instance to~$(X,X)$ using $\Oh(n)$ updates while keeping the weighted edit distance between~$0$ and~$2$ at all times.
Hence, as long as the preprocessing of $(\emptystring,\emptystring)$ takes $n^{\Oh(1)}$ time and is not allowed to access the entire weight function (e.g., weights are revealed online when an update introduces a fresh character for the first time), we can replace $T_P(n, 0)=\Oh(n^{0.5+1.5\kappa-\delta})$ with $T_U(n,2) = \Oh(n^{1.5\kappa - 0.5 - \delta})$ in the statement of \cref{cor:simple}.

Overall, we conclude that, in order to support updates with relatively large distances faster than naively (i.e., by computing $\wed(X,Y)$ from scratch using a static algorithm), one needs to pay a lot for both preprocessing and updates while $\wed(X,Y)$ is still very small.
In particular, we should decide in advance how large distances need to be supported efficiently.
This justifies a simplified variant of \cref{prob:base}, where the parameter $k$ is fixed at preprocessing time.
Here, instead of $\wed(X,Y)$, the algorithm reports the following quantity:
\[\wed_{\le k}(X,Y) = \begin{cases}
    \wed(X,Y) &\text{if }\wed(X,Y)\le k,\\
    \infty & \text{otherwise.}
\end{cases}\]

\begin{problem}[Dynamic Weighted Edit Distance with Fixed Threshold]\label{prob:fixed}
Given integers $1\le k \le n$ and oracle access to a normalized weight function $w\colon \Esigma^2 \to \mathbb{R}_{\ge 0}$, maintain strings $X,Y\in \Sigma^{\le n}$ subject to updates (character edits in $X$ and $Y$) and report $\wed_{\le k}(X,Y)$ upon each update.
\end{problem}

\Cref{cor:simple} can be rephrased in terms of \cref{prob:fixed}, but its major limitation is that it only applies to $k \ge \sqrt{n}$.
To overcome this issue, we compose multiple hard instances from \cref{thm:static-lb}.
This results in the following theorem, discussed in \cref{sec:lower-bounds:dynamic,app:dynamic-lower-bounds:lb_main}.

\begin{theorem}[Simplified version of \cref{thm:lb_main}]\label{thm:lb_main_simple}
    Suppose that \cref{prob:fixed} admits a solution with preprocessing time $T_P(n,k)$ and update time $T_U(n,k)$ for some non-decreasing functions $T_P$ and $T_U$.
    If $T_P(n,n^{\kappa})=\Ohtilde(n^{1+\kappa\cdot \gamma})$ and $T_U(n,n^{\kappa})=\Oh(n^{\kappa\cdot (3-\gamma)-\delta})$ hold for some parameters $\gamma \in [0.5, 1)$, $\kappa \in (0, 1/(3-2\gamma)]$, and $\delta>0$, then the APSP Hypothesis fails.
\end{theorem}

\Cref{thm:lb_main_simple} indicates that (conditioned on the APSP Hypothesis and up to subpolynomial-factor improvements), for a dynamic algorithm with preprocessing time $\Ohtilde(nk^{\gamma})$ for $\gamma\in [0.5,1)$, the best possible update time is $\Ohtilde(k^{3-\gamma})$.
The lower bound does not extend to $\kappa > 1/(3-2\gamma)$; precisely then, the $\Ohtilde(\sqrt{nk^3})$-time static algorithm improves upon the $\Ohtilde(k^{3-\gamma})$ update time.

\paragraph*{Our Results: Algorithms}
Our main positive result is an algorithm that achieves the fine-grained-optimal trade-off.

\begin{theorem}\label{thm:simple_tradeoff}
There exists an algorithm that, initialized with an extra real parameter $\gamma \in [0,1]$, solves \cref{prob:fixed} with preprocessing time $\Ohtilde(nk^{\gamma})$ and update time $\Ohtilde(k^{3-\gamma})$.
\end{theorem}

As a warm-up, in \cref{sec:simple-algorithm}, we prove \cref{thm:simple_tradeoff} for $\gamma=1$.
Similarly to \cite{GK24}, the main challenge is that updates can make $X$ and $Y$ slide with respect to each other.
If we interleave insertions at the end of $Y$ and deletions at the beginning of $Y$, then, after $\Omega(k)$ updates, for every character $X\position{i}$, the fragment $Y\fragmentcc{i-k}{i+k}$ containing the plausible matches of $X\position{i}$ changes completely.
The approach of \cite{GK24} was to restart the algorithm every $\Theta(k)$ updates.

Our preprocessing time is too large for this, so we resort to a trick originating from dynamic edit distance approximation~\cite{KMS23}: we maintain a data structure for aligning $X$ with itself.
Instead of storing $Y$, we show how to compute $\wed_{\le k}(X,Y)$ for any string $Y$ specified using $\Oh(k)$ edits (of arbitrary costs) transforming $X$ into $Y$.
To find such edits in the original setting of \cref{prob:fixed}, one can use a dynamic \emph{unweighted} edit distance algorithm.
Already the folklore approach, based on~\cite{LV88,MSU97}, with $\Ohtilde(k^2)$ time per update is sufficiently fast.

In \cref{sec:full-algorithm,app:full-algorithm}, we discuss how to generalize our dynamic algorithm to arbitrary $\gamma \in [0,1]$. 
This is our most difficult result -- it requires combining the techniques of~\cite{CKW23} and their subsequent adaptations in~\cite{GK24} with several new ideas, including a novel \emph{shifted} variant of self-edit distance.
We derive \cref{thm:simple_tradeoff} as a straightforward corollary of the following result, which demonstrates that the same fine-grained optimal trade-off can be achieved in a much more robust setting compared to the one presented in \cref{prob:fixed}.

\begin{theorem}[Simplified version of \cref{thm:full-complete-algorithm}]\label{thm:full-simpler}
    There is a dynamic algorithm that, initialized with a real parameter $\gamma \in [0, 1]$, integers $1\le k \le n$, and $\Oh(1)$-time oracle access to a normalized weight function $w \colon \Esigma^2 \to \RR_{\ge 0}$, maintains a string $X\in \Sigma^{\le n}$ and, after~$\Ohtilde(nk^\gamma)$-time preprocessing, supports the following operations (updates and queries):
    \begin{itemize}
        \item Apply a character edit, substring deletion, or copy-paste to $X$ in $\Ohtilde(k^2)$ time.
        \item Given $u$ edits transforming $X$ into a string $Y$, compute $\wed_{\le k}(X, Y)$ in $\Ohtilde(u + k^{3 - \gamma})$ time.
    \end{itemize}
\end{theorem}

We remark that our query algorithm for $\gamma=1$ extends a subroutine of~\cite{CKW24}, where the string $X$ is static.
In~\cite{CKW24}, this result is used for approximate pattern matching with respect to weighted edit distance.
If the pattern is far from periodic, which is arguably true in most natural instances, the trade-off of \cref{thm:full-simpler} yields a faster pattern matching~algorithm.

\paragraph*{Open Problems}
Recall that \cref{thm:lb_main_simple} proves fine-grained optimality of \cref{thm:simple_tradeoff} only for $\gamma \in [0.5,1)$.
The omission of $\gamma \in [0,0.5)$ stems from a gap between algorithms and lower bounds for the underlying static problem.
Hence, we reiterate the open question of~\cite{CKW23} asking for the complexity of 
computing $\wed(X,Y)$ when $\sqrt[3]{n} < \wed(X,Y) < \sqrt{n}$.
Since our lower bounds arise from \cref{thm:static-lb}, it might be instructive to first try designing faster algorithms for the case of $\ed(X,Y)\le 4$.
Before this succeeds, one cannot hope for faster dynamic algorithms.

Wider gaps in our understanding of \cref{prob:fixed} remain in the regime of large preprocessing time.
It is open to determine the optimal update time for $\gamma = 1$ and to decide whether $\gamma > 1$ can be beneficial.
As far as $\gamma = 1$ is concerned, we believe that \cite[Section 4]{CKM20} can be used to obtain $\Ohtilde(k\sqrt{n})$ update time, which improves upon \cref{thm:simple_tradeoff} for $k\ge \sqrt{n}$.
Moreover, by monotonicity, the lower bound of \cite[Section 6]{CKW23} excludes update times of the form $\Oh(k^{1.5 - \delta})$.

Finally, we remark that all our lower bounds exploit the fact that $\wed(X,Y)$ can change a lot as a result of a single update. 
It would be interesting to see what can be achieved if the weights are small (but fractional) or when the update time is parameterized by the change in $\wed(X,Y)$. 
In these cases, we hope for non-trivial applications of the approach of~\cite{GK24}.

\section{Preliminaries\protect\footnote{In this section, we partially follow the narration of \cite{CKW23,GK24}.}}\label{sec:preliminaries}

A string $X \in \Sigma^n$ is a sequence of $|X| \coloneqq n$ characters over an alphabet~$\Sigma$.
We denote the set of all strings over $\Sigma$ by $\Sigma^*$, we use $\emptystring$ to represent the empty string,
and we write $\Sigma^+ \coloneqq \Sigma^*\setminus\{\emptystring\}$.
For a \emph{position} $i \in \fragmentco{0}{n}$, we
say that $X\position{i}$ is the $i$-th character of $X$.
We say that a string $X$ occurs as a \emph{substring} of a string $Y$, if there exist indices
$i, j \in \fragmentcc{0}{|Y|}$ satisfying $i \le j$ such that $X = Y[i]\cdots Y[j-1]$.
The occurrence of $X$ in $Y$ specified by~$i$ and~$j$ is a \emph{fragment} of $Y$ denoted by $Y\fragmentco{i}{j}$, $Y\fragmentcc{i}{j-1}$, $Y\fragmentoc{i-1}{j-1}$,
or $Y\fragmentoo{i-1}{j}$.

\paragraph*{Weighted Edit Distances and Alignment Graphs}
Given an alphabet $\Sigma$, we set $\Esigma \coloneqq \Sigma \cup \{\emptystring\}$.
We call $w$ a \emph{normalized weight function} if $w \colon \Esigma \times \Esigma \to \RR_{\geq 1}\cup \{0\}$ and, for $a, b \in \Esigma$, we have $w(a, b) = 0$ if and only if $a = b$.
Note that~\(w\) does not need to satisfy the triangle inequality nor does \(w\) need to be
symmetric.

We interpret weighted edit distances of strings as distances in alignment graphs.

\begin{definition}[Alignment Graph,~{\cite{CKW23}}]\label{def:alignment-graph}
For strings $X,Y\in \Sigma^*$ and a normalized weight function $w\colon \Esigma^2 \to \Rz$,
we define the \emph{alignment graph} $\AGw(X,Y)$ to be a weighted directed graph with vertices $\fragmentcc{0}{|X|}\times \fragmentcc{0}{|Y|}$ and the following edges:
\begin{itemize}
    \item a horizontal edge $(x,y)\to (x+1,y)$ of cost $w(X\position{x},\emptystring)$ representing a deletion of~$X\position{x}$, for every $(x,y)\in \fragmentco{0}{|X|}\times \fragmentcc{0}{|Y|}$,
    \item a vertical edge $(x,y)\to (x,y+1)$ of cost $w(\emptystring,Y\position{y})$ representing an insertion of~$Y\position{y}$, for every $(x,y)\in \fragmentcc{0}{|X|}\times \fragmentco{0}{|Y|}$, and
    \item a diagonal edge $(x,y)\to (x+1,y+1)$ of cost $w(X\position{x},Y\position{y})$ representing a substitution of~$X\position{x}$ into $Y\position{y}$ or, if $X\position{x}=Y\position{y}$, a match of $X\position{x}$ with $Y\position{y}$, for every $(x,y)\in \fragmentco{0}{|X|}\times \fragmentco{0}{|Y|}$.
\end{itemize}
\end{definition}

We visualize the graph $\AGw(X, Y)$ as a grid graph with $|X|+1$ columns and $|Y|+1$ rows, where $(0,0)$ and $(|X|,|Y|)$ are the top-left and bottom-right vertices, respectively.

\emph{Alignments} between fragments of $X$ and $Y$ can be interpreted as paths in $\AGw(X,Y)$.

\begin{definition}[Alignment]
    Consider strings $X,Y\in \Sigma^*$ and a normalized weight function $w\colon \Esigma^2 \to \Rz$.
    An \emph{alignment} $\cA$ of $X\fragmentco{x}{x'}$ onto $Y\fragmentco{y}{y'}$, denoted $\cA : X\fragmentco{x}{x'} \onto Y\fragmentco{y}{y'}$, is a path from $(x,y)$ to $(x',y')$ in $\AGw(X,Y)$, often interpreted as a sequence of vertices.
    The \emph{cost} of $\cA$, denoted by $\wed_{\cA}(X\fragmentco{x}{x'},Y\fragmentco{y}{y'})$, is the total cost of edges belonging to $\cA$. 
\end{definition}

The \emph{weighted edit distance} of strings $X, Y \in \Sigma^*$
with respect to a weight function~$w$ is~$\wed(X, Y) \coloneqq \min \wed_\cA(X, Y)$ where the minimum is taken over all alignments $\cA$ of $X$ onto~$Y$.
An alignment $\cA$ of $X$ onto $Y$ is \emph{$w$-optimal} if $\wed_\cA(X, Y) = \wed(X, Y)$.

We say that an alignment $\cA$ of $X$ onto $Y$ aligns fragments $X\fragmentco{x}{x'}$ and $Y\fragmentco{y}{y'}$
if~$(x,y),(x',y')\in \cA$.
In this case, we also denote the cost of the \emph{induced alignment} (the subpath of $\cA$ from $(x,y)$ to $(x',y')$) by $\wed_\cA(X\fragmentco{x}{x'}, Y\fragmentco{y}{y'})$.

A normalized weight function satisfying $w(a,b) = 1$ for distinct $a, b \in \Esigma$ gives the \emph{unweighted} edit distance (Levenshtein distance~\cite{Levenshtein66}).
In this case, we drop the superscript $w$.

\begin{definition}[Augmented Alignment Graph,~{\cite{GK24}}]
For strings $X,Y\in \Sigma^*$ and a weight function $\wdef$, the \emph{augmented alignment graph} $\oAGw(X, Y)$ is obtained from $\AGw(X,Y)$ by adding, for every edge of $\AGw(X, Y)$, a back edge of weight $W + 1$.\footnote{These edges ensure that $\oAGw(X, Y)$ is strongly connected, and thus the distance matrices used throughout this work are Monge matrices (see~\cref{def:monge}) rather than so-called partial Monge matrices.}
\end{definition}

\noindent
The following fact shows several properties of $\oAGw(X, Y)$.
Importantly, $\wed(X\fragmentco{x}{x'},Y\fragmentco{y}{y'})$ is the distance from $(x,y)$ to $(x',y')$ in $\oAGw(X,Y)$.

\begin{fact}[{\cite[Lemma 5.2]{GK24}}] \label{lm:oAGw-properties}
    Consider strings $X, Y \in \Sigma^{*}$ and a weight function $\wdef$.
    Every two vertices $(x, y)$ and $(x', y')$ of the graph $\oAGw(X,Y)$ satisfy the following properties:
\begin{description}
    \item[Monotonicity.] Every shortest path from $(x, y)$ to $(x', y')$ in $\overline{\AG}^w(X, Y)$ is (weakly) monotone in both coordinates.
    \item[Distance preservation.] If $x \le x'$ and $y \le y'$, then
        \[\dist_{\oAGw(X, Y)}((x, y), (x', y')) = \dist_{\AGw(X, Y)}((x, y), (x', y')).\]
    \item[Path irrelevance.] If ($x \le x'$ and $y \ge y'$) or ($x \ge x'$ and $y \le y'$), then every path from $(x, y)$ to $(x', y')$ in $\oAGw(X, Y)$ monotone in both coordinates is a shortest path between these two vertices.
\end{description}
\end{fact}

Consistently with many edit-distance algorithms, we will often compute not only the distance from $(0,0)$ to $(|X|,|Y|)$ in $\oAGw(X, Y)$ but the entire \emph{boundary distance matrix} $\BMw(X,Y)$ that stores the distances from every \emph{input} vertex on the top-left boundary to every \emph{output} vertex on the bottom-right boundary of the graph $\oAGw(X,Y)$.

\begin{definition}[Input and Output Vertices of an Alignment Graph, {\cite[Definition 5.5]{GK24}}]
    Let $X, Y \in \Sigma^*$ be two strings and $\wdef$ be a weight function.
    Furthermore, for fragments $X\fragmentco{x}{x'}$ and $Y\fragmentco{y}{y'}$, let $R = \fragmentcc{x}{x'} \times \fragmentcc{y}{y'}$ be the corresponding \emph{rectangle} in~$\oAGw(X, Y)$.
    We define the \emph{input vertices} of $R$ as the sequence of vertices of $\oAGw(X, Y)$ on the left and top boundary of $R$ in the clockwise order.
    Analogously, we define the \emph{output vertices} of $R$ as a sequence of vertices of $\oAGw(X, Y)$ on the bottom and right boundary of $R$ in the counterclockwise order.
    Furthermore, the input and output vertices of $X \times Y$ are the corresponding boundary vertices of the whole graph $\oAGw(X, Y)$.
\end{definition}

\begin{definition}[Boundary Distance Matrix, {\cite[Definition 5.6]{GK24}}]
    Let $X, Y \in \Sigma^*$ be two strings and $\wdef$ be a weight function.
    The \emph{\textbf{b}oundary distance \textbf{m}atrix} $\BMw(X, Y)$ of $X$ and $Y$ with respect to $w$ is a matrix $M$ of size $(|X|+|Y|+1)\times (|X|+|Y|+1)$, where $M_{i, j}$ is the distance from the $i$-th input vertex to the $j$-th output vertex of $X \times Y$ in $\oAGw(X, Y)$.
\end{definition}

The following fact captures the simple dynamic programming algorithm for bounded weighted edit distance.

\begin{fact}[{\cite[Fact 3.3]{GK24}}] \label{lm:baseline-wed}
    Given strings $X, Y \in \Sigma^*$, an integer $k \geq 1$, and $\Oh(1)$-time oracle access to
    a normalized weight function
    $w \colon \Esigma^2 \to \RR_{\geq 0}$, the value $\wed_{\leq k}(X, Y)$ can be computed in
    $\Oh(\min(|X| + 1, |Y| + 1) \cdot \min(k, |X| + |Y| + 1))$ time.
    Furthermore, if $\wed(X, Y) \le k$, the algorithm returns a $w$-optimal alignment of $X$ onto $Y$.
\end{fact}

The proof of \cref{lm:baseline-wed} relies on the fact that a path of length $k$ can visit only a limited area of the alignment graph $\AGw(X, Y)$. This property can be formalized as follows:

\begin{fact}[{\cite[Lemma 5.3]{GK24}}] \label{lm:paths_dont_deviate_too_much}
    Let $X, Y \in \Sigma^*$ be two strings, $k$ be an integer, and $w \colon \Esigma^2 \to \RR_{\ge 0}$ be a normalized weight function.
    Consider a path $P$ of cost at most $k$ connecting $(x, y)$ to $(x', y')$ in $\oAGw(X,Y)$.
    All vertices $(x^*, y^*) \in P$ satisfy $|(x - y) - (x^* - y^*)| \le k$ and $|(x' - y') - (x^* - y^*)| \le k$.
\end{fact}

The \emph{breakpoint representation} of an alignment $\cA=(x_t,y_t)_{t=0}^m$ of $X$ onto $Y$
is the subsequence of $\cA$ consisting of pairs $(x_t,y_t)$ such that $t\in \{0,m\}$ or
$\cA$ does not match $X\position{x_t}$ with $Y\position{y_t}$.
Note that the size of the breakpoint representation is at most $2+\ed_\cA(X,Y)$ and that it
can be used to retrieve the entire alignment and its cost:
for any two consecutive elements $(x',y'),(x,y)$ of the breakpoint representation, it
suffices to add $(x-\delta,y-\delta)$ for~$\delta \in \fragmentoo{0}{\max(x-x',y-y')}$.
We will also talk of \emph{breakpoint representations} of paths in~$\oAGw(X, Y)$ that are weakly monotone in both coordinates.
If such a path starts in $(x, y)$ and ends in $(x', y')$ with $x \le x'$ and $y \le y'$, then the breakpoint representation of such a path is identical to the breakpoint representation of the corresponding alignment.
Otherwise, if $x > x'$ or $y > y'$, we call the breakpoint representation of such a path a sequence of all its vertices; recall that all edges in such a path have strictly positive weights.

\paragraph*{Planar Graphs, the Monge Property, and Min-Plus Matrix Multiplication}\label{sec:alg:sec:alg-periodic:sec:planar-toolbox}

As alignment graphs are planar, we can use the following tool of Klein.

\begin{fact}[Klein~\cite{Klein2005}]\label{thm:klein}
    Given a directed planar graph $G$ of size $n$ with non-negative edge weights, we can construct in $\Oh(n \log n)$ time a data structure that, given two vertices $u, v \in V(G)$ at least one of which lies on the outer face of $G$, computes the distance $\dist_G(u, v)$ in $\Oh(\log n)$ time.
    Moreover, the data structure can report the shortest $u\leadsto v$ path $P$
    in $\Oh(|P|\log \log\Delta(G))$ time, where $\Delta(G)$ is the maximum degree of $G$.
\end{fact}

The following lemma specifies a straightforward variation of Klein's algorithm tailored for path reconstruction in alignment graphs. 

\renewcommand{\maybe}{\lipicsEnd}
\begin{restatable}{lemma}{lmkleinupgraded}\label{lm:klein-upgraded}
    Given strings $X, Y \in \Sigma^+$ and $\Oh(1)$-time access to a normalized weight function $\wdef$, we can construct in $\Oh(|X| \cdot |Y| \log^2 (|X| + |Y|))$ time a data structure that, given two vertices $u, v \in \oAGw(X, Y)$ at least one of which lies on the outer face of $\oAGw(X, Y)$, computes the distance $\dist_{\oAGw(X, Y)}(u, v)$ in $\Oh(\log (|X| + |Y|))$ time.
    Moreover, the breakpoint representation of a shortest $u \leadsto v$ path can be constructed in $\Oh((1 + \dist_{\oAGw(X, Y)}(u, v)) \log^2 (|X| + |Y|))$ time.
\end{restatable}
\renewcommand{\maybe}{}

\begin{proof}
    The data structure we build is recursive.
    We first build the data structure of Klein's algorithm (\cref{thm:klein}) for $\oAGw(X, Y)$, and then if $|X| > 1$, recurse onto $(X\fragmentco{0}{\floor{|X| / 2}}, Y)$ and $(X\fragmentco{\floor{|X| / 2}}{|X|}, Y)$.
    Preprocessing costs $\Oh(|X| \cdot |Y| \log^2 (|X| + |Y|))$ time.
    For distance queries, the top-level instance of Klein's algorithm is sufficient.
    It remains to answer path reconstruction queries.
    We reconstruct the shortest $u \leadsto v$ path recursively.
    First, in $\Oh(\log (|X| + |Y|))$ time we query $\dist_{\oAGw(X, Y)}(u, v)$.
    If the distance is $0$, we return the breakpoint representation of the trivial shortest $u \leadsto v$ path.
    Otherwise, we make further recursive calls.
    If $u$ and $v$ both lie in $X\fragmentco{0}{\floor{|X| / 2}} \times Y$ or both lie in $X\fragmentco{\floor{|X| / 2}}{|X|} \times Y$, by the monotonicity property of \cref{lm:oAGw-properties} we can recurse onto the corresponding small recursive data structure instance.
    Otherwise, every $u \leadsto v$ path contains a vertex $(x, y)$ with $x = \floor{|X| / 2}$.
    Furthermore, due to \cref{lm:paths_dont_deviate_too_much}, there are $\Oh(\dist_{\oAGw(X, Y)}(u, v))$ possible values for $y$.
    In $\Oh(\dist_{\oAGw(X, Y)}(u, v) \log (|X| + |Y|))$ time we query $\dist_{\oAGw(X, Y)}(u, (x, y))$ and $\dist_{\oAGw(X, Y)}((x, y), v)$ using the smaller recursive data structure instances and find some $y$ satisfying $\dist_{\oAGw(X, Y)}(u, v) = \dist_{\oAGw(X, Y)}(u, (x, y)) + \dist_{\oAGw(X, Y)}((x, y), v)$.
    We then recursively find the breakpoint representations of the optimal $u \leadsto (x, y)$ path and the optimal $(x, y) \leadsto v$ path in the smaller recursive data structure instances.
    There are $\Oh(\log (|X| + 1))$ recursive levels, on each one of them our total work is $\Oh((1 + \dist_{\oAGw(X, Y)}(u, v)) \log (|X| + |Y|))$; thus, the total time complexity is $\Oh((1 + \dist_{\oAGw(X, Y)}(u, v)) \log^2 (|X| + |Y|))$.
\end{proof}

As discussed in \cite[Section 2.3]{FR06} and~\cite[Fact 3.7 and Section 5.1]{GK24}, the planarity of alignment graphs implies that the $\BMw(X, Y)$ matrices satisfy the following \emph{Monge property}:

\begin{definition}[Monge Matrix] \label{def:monge}
    A matrix $A$ of size $p \times q$ is a \emph{Monge} matrix if, for all~$i \in [0 \dd p - 1)$ and~$j \in [0 \dd q - 1)$, we have $A_{i, j} + A_{i + 1, j + 1} \le A_{i, j + 1} + A_{i + 1, j}$.
\end{definition}

In the context of distance matrices, concatenating paths corresponds to the min-plus product of matrices; this operation preserves Monge property and can be evaluated~fast.

\begin{fact}[Min-Plus Matrix Product, SMAWK Algorithm~\cite{SMAWK87}, {\cite[Theorem 2]{Tiskin10}}]\label{thm:smawk}
    Let $A$ and~$B$ be matrices of size $p \times q$ and $q \times r$, respectively.
    Their \emph{min-plus product} is a matrix $A \otimes B\coloneqq C$ of size $p \times r$ such that $C_{i, k} = \min_j A_{i, j} + B_{j, k}$ for all $i \in [0\dd p)$ and $k \in [0\dd r)$.

    If $A$ and $B$ are Monge matrices, then $C$ is also a Monge matrix. 
    Moreover, $C$ can be constructed in $\Oh(pr + \min(pq, qr))$ time
    assuming $\Oh(1)$-time random access to $A$ and $B$.
\end{fact}

\section{\boldmath \texorpdfstring{$\Ohtilde(k^2)$}{Õ(k²)}-Time Updates after  \texorpdfstring{$\Ohtilde(nk)$}{Õ(nk)}-Time Preprocessing}\label{sec:simple-algorithm}

In this section, we present a dynamic algorithm that maintains $\wed_{\le k}(X, Z)$ for two dynamic strings $X, Z \in \Sigma^{\le n}$ with $\tOh(nk)$ preprocessing time and $\tOh(k^2)$ update time.
As all vertical and horizontal edges cost at least $1$, the relevant part of the graph $\oAGw(X, Z)$ has size~$\Oh(nk)$ (i.e., a band of width $\Theta(k)$ around the main diagonal), and thus we can afford to preprocess this part of $\oAGw(X, Z)$.
Nevertheless, the curse of dynamic algorithms for bounded weighted edit distance is that, after $\Theta(k)$ updates (e.g., $\Theta(k)$ character deletions from the end of~$Z$ and $\Theta(k)$ character insertions at the beginning of~$Z$), the preprocessed part of the original graph~$\oAGw(X, Z)$ can be completely disjoint from the relevant part of the updated graph~$\oAGw(X, Z)$, which renders the original preprocessing useless.
To counteract this challenge, we use the idea from~\cite{KMS23}: we dynamically maintain information about $\oAGw(X, X)$ instead of $\oAGw(X,Z)$.
Every update to $X$ now affects both dimensions of $\oAGw(X, X)$, so the relevant part of the graph does not shift anymore, and hence our preprocessing does not expire.
Tasked with computing $\wed_{\le k}(X,Z)$ upon a query, we rely on the fact that, unless~$\wed_{\le k}(X,Z)=\infty$, the graph $\oAGw(X,Z)$ is similar to $\oAGw(X, X)$.

Consistently with the previous works~\cite{CKW23,GK24}, we cover the relevant part of $\oAGw(X, X)$ with $m = \Oh(n / k)$ rectangular subgraphs $G_i$ of size $\Theta(k) \times \Theta(k)$ each; see \cref{fig:boxes-stripe-D}.
We preprocess each subgraph in $\tOh(k^2)$ time.
A single update to $X$ alters a constant number of subgraphs $G_i$, so we can repeat the preprocessing for such subgraphs in $\tOh(k^2)$ time.
The string $Z$ does not affect any subgraph $G_i$; we only store it in a dynamic strings data structure supporting efficient substring equality queries for the concatenation $X\cdot Z$.

Let $V_i$ for $i \in \fragmentco{1}{m}$ be the set of vertices in the intersection of $G_i$ and $G_{i+1}$; see \cref{fig:boxes-stripe-D}.
Furthermore, let $V_0 \coloneqq \set{(0, 0)}$ and $V_m \coloneqq \set{(|X|, |X|)}$.
Let $G$ be the union of all graphs $G_i$.
Let $D_{i, j}$ for $i, j \in \fragmentcc{0}{m}$ with $i < j$ denote the matrix of pairwise distances from~$V_i$ to~$V_j$ in~$G$.
Note that $D_{a, c} = D_{a, b} \otimes D_{b, c}$ holds for all $a, b, c \in \fragmentcc{0}{m}$ with $a < b < c$.

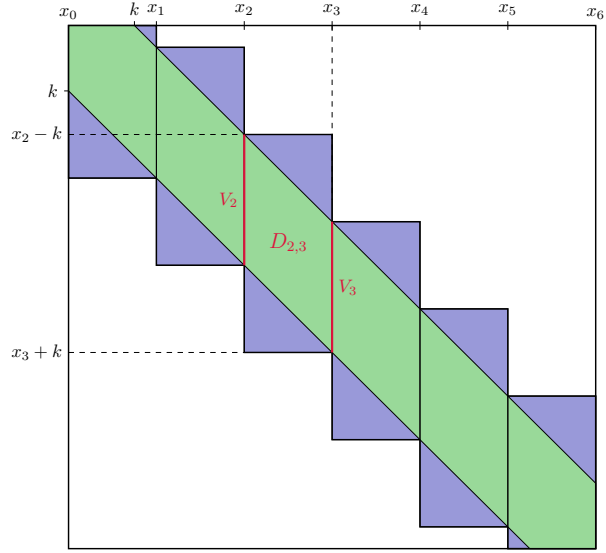
\begin{figure}
    \begin{center}
        \begin{tikzpicture}[y=-.85cm,x=0.85cm]
            \useasboundingbox (-.5, -.5) rectangle (8,8);
            \scope[transform canvas={scale=.68}]
                \def\shft{0}
\def\W{12}
\def\H{12}
\def\myi{2}
\def\cnt{6}
\def\gp{2}
\def\kp{1.5}
\def\charwidth{1 / 15}

\foreach \x in {0,5*\gp}
    \draw[fill=darkblue!40, thick] (\x, {max(\x - \kp, 0)}) rectangle (\x + \gp, {min(\x + \gp + \kp, \H)});

\foreach \x in {\gp,2*\gp,3*\gp} {
    \draw[fill=darkblue!40, thick] (\x, {max(\x - \kp, 0)}) rectangle (\x + \gp, {min(\x + \gp + \kp, \H)});
}

\foreach \x in {4*\gp} {
    \draw[fill=darkblue!40, thick] (\x, {max(\x - \kp, 0)}) rectangle (\x + \gp, {min(\x + \gp + \kp, \H)});
}

\draw (\kp,0.1) -- (\kp,-0.1) node[above] {$k$};
\draw (0.1,\kp) -- (-0.1,\kp) node[left] {$k$};
\draw[fill=darkgreen!40] (0, 0) -- (\kp, 0) -- (\H, \W - \kp) -- (\H, \W) -- (\H - \kp, \W) -- (0, \kp) -- (0, 0);

\foreach \x in {0,\gp,...,{\the\numexpr \W - \gp}}
    \draw (\x, {max(\x - \kp, 0)}) rectangle (\x + \gp, {min(\x + \kp + \gp, \H)});

\draw[thick] (0, 0) rectangle (\H, \W);

\draw (0,0.1) -- (0,0) node[above] {$x_0$};
\draw (\W,0.1) -- (\W,0) node[above] {$x_{\cnt}$};
\foreach \x in {1,2,...,{\the\numexpr \cnt - 1}}
    \draw (\x * \gp,0.1) -- (\x * \gp, -0.1) node[above] {$x_{\x}$};

\draw[dashed] (\myi * \gp + \gp, 0) -- (\myi * \gp + \gp, \myi * \gp);
\draw[dashed] (0,\myi * \gp - \kp) node[left] {$x_{\myi} - k$} -- (\myi * \gp,\myi * \gp - \kp);
\draw[dashed] (0,\myi * \gp + \gp + \kp) node[left] {$x_{\the\numexpr \myi + 1} + k$} -- (\myi * \gp,\myi * \gp + \kp+ \gp);

\node[left, darkred] at (\myi * \gp, \myi * \gp) {$V_{\myi}$};
\node[right, darkred] at (\myi * \gp + \gp, \myi * \gp + \gp) {$V_{\the\numexpr \myi + 1}$};

\draw[very thick, darkred] (\myi * \gp, \myi * \gp - \kp) -- (\myi * \gp, \myi * \gp + \kp);
\draw[very thick, darkred] (\myi * \gp + \gp, \myi * \gp+\gp-\kp) -- (\myi * \gp + \gp, \myi * \gp + \gp + \kp);

\node[darkred] at (\myi * \gp + \gp / 2, \myi * \gp + \gp / 2) {\Large $D_{\myi, {\the\numexpr \myi + 1}}$};
            \endscope
        \end{tikzpicture}
    \end{center}

    \caption{The decomposition of the alignment graph of $X$ onto $X$ into \textcolor{darkblue!60}{subgraphs} $G_i$ of size $\Theta(k) \times \Theta(k)$ that cover the whole \textcolor{darkgreen!60}{stripe} of width $\Theta(k)$ around the main diagonal. Each matrix~$D_{i, i + 1}$ represents the pairwise distances from $V_i$ to $V_{i+1}$.}
    \label{fig:boxes-stripe-D}
\end{figure}

Upon a query for $\wed(X, Z)$, we note that only $\Oh(\ed(X, Z))$ subgraphs $G_i$ in $\oAGw(X, X)$ differ from the corresponding subgraphs in $\oAGw(X, Z)$.
(If $\ed(X, Z) > k$, then $\wed_{\le k}(X, Z) = \infty$, so we can assume that $\ed(X, Z) \le k$.)
We call such subgraphs \emph{affected}.
Let $G'_i$, $V'_i$, and $D'_{i, j}$ be the analogs of $G_i$, $V_i$, and $D_{i, j}$ in $\oAGw(X, Z)$.
Note that if $\wed(X, Z) \le k$, then $\wed(X, Z)$ is the only entry of the matrix $D'_{0, m}$.
Furthermore, we have $D'_{0, m} = D'_{0, 1} \otimes \cdots \otimes D'_{m-1, m}$.
We compute this product from left to right.
For unaffected subgraphs~$G_i$, we have $D'_{i-1, i} = D_{i-1, i}$, and we can use precomputed information about $\oAGw(X, X)$.
To multiply by $D'_{i-1, i}$ for affected subgraphs $G_i$, we use the following lemma, the proof of which is deferred to \cref{app:simple-algorithm:box-with-one-string-changing}.

\begin{restatable}[{Generalization of~\cite[Claim 4.2]{CKW24}}]{lemma}{lmboxwithonestringchanging}\label{lm:box-with-one-string-changing}
    Let $X, Y \in \Sigma^+$ be nonempty strings and \wdef be a weight function.
    Given $\Oh(1)$-time oracle access to $w$, the strings $X$ and $Y$ can be preprocessed in $\Oh(|X| \cdot |Y| \log^2 (|X| + |Y|))$ time so that, given a string $Y' \in \Sigma^*$ and a vector $v$ of size $|X|+|Y'|+1$, the row $v^T \otimes \BMw(X, Y')$ can be computed in $\Oh((\ed(Y, Y') + 1) \cdot (|X| + |Y|) \log (|X| + |Y|))$ time.
    Furthermore, given an input vertex $a$ and an output vertex $b$ of $X \times Y'$, the shortest path from $a$ to $b$ in $\oAGw(X, Y')$ can be computed in the same time complexity.
\end{restatable}

There are $\ed(X, Z)$ edits that transform $\oAGw(X, X)$ into $\oAGw(X, Z)$, and each edit affects a constant number of subgraphs $G_i$.
Therefore, applying \cref{lm:box-with-one-string-changing} for all affected subgraphs takes time $\tOh(k^2)$ in total.
Between two consecutive affected subgraphs $G_i$ and $G_j$, we have $D'_{i, j - 1} = D_{i, j - 1}$, and thus it is sufficient to store some range composition data structure over $(D_{i, i + 1})_{i=0}^{m-1}$ to quickly propagate between subsequent affected subgraphs.

As the information we maintain dynamically regards $\oAGw(X, X)$ and does not involve the string~$Z$, we may generalize the problem we are solving.
All we need upon a query is an alignment of $X$ onto $Z$ of \emph{unweighted} cost at most $k$.
Such an alignment can be obtained, for example, from a dynamic bounded unweighted edit distance algorithm~\cite{GK24} in $\tOh(k)$ time per update.
Alternatively, it is sufficient to have $\tOh(1)$-time equality tests between fragments of $X$ and $Z$.
In this case, we can run the classical Landau--Vishkin~\cite{LV88} algorithm from scratch after every update in $\tOh(k^2)$ time to either get an optimal unweighted alignment or learn that $\wed(X, Z) \ge \ed(X, Z) > k$.
Efficient fragment equality tests are possible in a large variety of settings~\cite{CKW20} including the dynamic setting~\cite{MSU97,ABR00,GKKLS18,KK22}, so we will not focus on any specific implementation.
Instead, we assume that the string $Z$ is already given as a sequence of $u \le k$ edits of an unweighted alignment transforming $X$ into $Z$, and our task is to find a $w$-optimal alignment.
We are now ready to formulate the main theorem of this section.

\begin{restatable}{theorem}{lmsimplecompletealgorithm}\label{lm:simple-complete-algorithm}
    There is a dynamic data structure that, given a string $X \in \Sigma^*$, a threshold $k \ge 1$, and $\Oh(1)$-time oracle access to a normalized weight function $w \colon \Esigma^2 \to \RR_{\ge 0}$, can be initialized in $\Oh((|X| + 1) k \log^2 (|X| + 2))$ time and allows for the following operations:
    \begin{itemize}
        \item Apply a character edit to $X$ in $\Oh(k^2 \log^2 (|X| + 2))$~time.
        \item Given $u \le k$ edits transforming $X$ into a string $Z \in \Sigma^*$, compute $\wed_{\le k}(X, Z)$ in $\Oh(k\cdot (u+ \log (|X| + 2))\cdot \log (|X| + 2))$ time.
            Furthermore, if $\wed(X, Z) \le k$, the algorithm returns the breakpoint representation of a $w$-optimal alignment of $X$ onto $Z$.
    \end{itemize}
\end{restatable}

\begin{proof}[Proof Sketch]
    At the initialization phase, we decompose $\oAGw(X, X)$ into subgraphs $G_i$ and, for each subgraph $G_i$, initialize the algorithm of \cref{lm:box-with-one-string-changing} and compute $D_{i-1, i}$ using \cref{lm:klein-upgraded}.
    We also build a dynamic range composition data structure over $(D_{i, i + 1})_{i=0}^{m-1}$, implemented as a self-balancing binary tree.
    Given a range $\fragmentco{a}{b}$, in $\Oh(\log n)$ time, this data structure returns (pointers to) $\ell = \Oh(\log n)$ matrices $\row{S}{1}{\ell}$ such that $D_{a, b} = S_1 \otimes \cdots \otimes S_{\ell}$.

    When an update to $X$ comes, it affects a constant number of subgraphs $G_i$, which we locally adjust and recompute the initialization for in $\tOh(k^2)$ time.

    Upon a query, we are tasked with computing $D'_{0, 1} \otimes \cdots \otimes D'_{m-1, m}$.
    We locate the $\Oh(u)$ affected subgraphs $G_i$ and process them one-by-one.
    We use \cref{lm:box-with-one-string-changing} to multiply the currently computed row $v^T = D'_{0, 1} \otimes \cdots \otimes D'_{i-2, i-1}$ by $D'_{i-1, i}$.
    The algorithm then uses the range composition data structure over $(D_{i, i + 1})_{i=0}^{m-1}$ to obtain $\ell = \Oh(\log n)$ matrices $\row{S}{1}{\ell}$ that correspond to the segment of matrices between the current affected subgraph $G_i$ and the next affected subgraph $G_j$.
    We use SMAWK algorithm (\cref{thm:smawk}) to multiply $v^T$ by matrices $\row{S}{1}{\ell}$.
    After processing all the matrices, we obtain $D'_{0, 1} \otimes \cdots \otimes D'_{m-1, m}$, the only entry of which is equal to $\wed(X, Z)$ if $\wed_{\le k}(X, Z) \neq \infty$.
    The applications of \cref{lm:box-with-one-string-changing} work in $\Oh(k u \log n)$ time in total.
    Each of the $\Oh(u)$ blocks of unaffected subgraphs is treated by invoking SMAWK algorithm (\cref{thm:smawk}) $\Oh(\log n)$ times for a total of $\Oh(k u \log n)$ time.

    See \cref{app:simple-algorithm:simple-complete-algorithm} for a formal proof including the alignment reconstruction procedure.
\end{proof}

\noindent
We complete this section with a simple corollary improving the query time if \mbox{$\wed(X, Z) \ll k$}.

\begin{corollary}\label{lm:simple-complete-algorithm-upgraded}
    There is a dynamic data structure that, given a string $X \in \Sigma^*$, a threshold $k \ge 1$, and $\Oh(1)$-time oracle access to a normalized weight function $w \colon \Esigma^2 \to \RR_{\ge 0}$ can be initialized in $\Oh((|X| + 1) k \log^2 (|X| + 2))$ time and allows for the following operations:
    \begin{itemize}
        \item Apply a character edit to $X$ in $\Oh(k^2 \log^2 (|X| + 2))$ time.
        \item Given $u \le k$ edits transforming $X$ into a string $Z \in \Sigma^*$, compute $\wed_{\le k}(X, Z)$ in $\Oh(1 + (u + d)\cdot (u +\log (|X| + 2))\cdot \log (|X| + 2))$ time, where $d\coloneqq \min(\wed(X, Z), k)$.
            Furthermore, if $\wed(X, Z) \le k$, the algorithm returns the breakpoint representation of a $w$-optimal alignment of $X$ onto $Z$.
    \end{itemize}
\end{corollary}

\begin{proof}
    We maintain the data structures of \cref{lm:simple-complete-algorithm} for thresholds $1, 2, 4, \ldots, 2^{\ceil{\log k}}$.
    Upon a query, we query the data structures for subsequent thresholds starting from $2^{\ceil{\log u}}$ until we get a positive result or reach $2^{\ceil{\log k}}$.
\end{proof}

As discussed earlier, in a variety of settings, upon a query, in $\tOh(k^2)$ time we can compute an optimal unweighted alignment of $X$ onto $Z$ or learn that $\wed(X, Z) \ge \ed(X, Z) > k$.
Therefore, we may assume that whenever \cref{lm:simple-complete-algorithm-upgraded} is queried, we have $u = \ed(X, Z) \le k$.
In this case, the query time complexity can be rewritten as $\Oh(1 + \min(\wed(X, Z), k) \cdot (\ed(X, Z)+ \log(|X| + 2))\cdot \log (|X| + 2))$.

Due to the robustness of the setting of \cref{lm:simple-complete-algorithm-upgraded}, there are many problems beyond dynamic weighted edit distance that can be solved using \cref{lm:simple-complete-algorithm-upgraded}.
For example, given strings $X, Y_0, \ldots, Y_{m-1} \in \Sigma^{\le n}$, we can find $\wed_{\le k}(X, Y_i)$ for all $i \in \fragmentco{0}{m}$ in $\tOh(nm + nk + mk^2)$ time, compared to $\tOh(m \cdot (n + \sqrt{nk^3}))$ time arising from $m$ applications of the optimal static algorithm~\cite{CKW23}.

\section{Trade-Off Algorithm: Technical Overview}\label{sec:full-algorithm}

In \cref{app:full-algorithm}, we generalize the result of \cref{lm:simple-complete-algorithm} in two ways.
First, we extend the set of updates the string $X$ supports from character edits to \emph{block edits}, which include substring deletions and copy-pastes.
More importantly, rather than having $\tOh(nk)$ preprocessing time and $\tOh(k^2)$ query time, we allow for a complete trade-off matching the lower bound of \cref{thm:lb_main_simple}, i.e., $\tOh(n k^{\gamma})$ preprocessing and $\tOh(k^{3 - \gamma})$ update time for any $\gamma \in [0, 1]$.

A major challenge for $\gamma < 1$ is that we cannot preprocess the whole width-$\Theta(k)$ band around the main diagonal of $\oAGw(X, X)$.
In the proof of \cref{lm:simple-complete-algorithm}, we cover the band with $\Theta(k) \times \Theta(k)$-sized subgraphs $G_i$ and preprocess each of them in $\tOh(k^2)$ time for a total of $\tOh(nk)$.
To achieve $\tOh(n k^{\gamma})$-time preprocessing, we are forced to instead use $\Theta(k^{2 - \gamma}) \times \Theta(k^{2 - \gamma})$-sized subgraphs $G_i$ and still spend just $\tOh(k^2)$ preprocessing time per subgraph.

To remedy the shortage of preprocessing time, we use the ideas behind the recent static algorithms \cite{CKW23,GK24}.
At the expense of a logarithmic-factor slowdown, they allow us to answer queries only for ``repetitive'' fragments of $X$ rather than the whole $X$.
More precisely, we are tasked to answer queries involving fragments $\hX$ of $X$ satisfying $\sed(\hX) \le k$, where the \emph{self-edit distance} $\sed(\hX)$ is the minimum unweighted cost of an alignment of $\hX$ onto itself that does not match any character of $\hX$ with itself.
We can use this to our advantage: the subgraphs $G_i$ corresponding to more ``repetitive'' fragments $X_i$ allow for more thorough preprocessing in $\tOh(k^2)$ time since we can reuse some computations.
On the other hand, upon a query, we can afford to spend more time on subgraphs $G_i$ with less ``repetitive'' fragments~$X_i$ as their number is limited due to $\sed(\hX) \le k$.

This approach requires the repetitiveness measure to be super-additive: $\sum_i \sed(X_i) \le \sed(\hX)$.
Unfortunately, self-edit distance is sub-additive: $\sum_i \sed(X_i) \ge \sed(\hX)$.
To circumvent this issue, we introduce a better-suited notion of repetitiveness we call \emph{$k$-shifted self-edit distance} $\sedk(\hX)$: a relaxed version of $\sed(\hX)$ allowed to delete up to $k$ first characters of $\hX$ and insert up to $k$ last characters of $\hX$ for free.
In contrast to self-edit distance, the $k$-shifted variant $\sedk$ satisfies \mbox{$\sum_i \sedk(X_i) \le \sed(\hX)$ if $\sed(\hX) \le k$.}

By using the value of $\sedk(X_i)$ to determine the level of preprocessing performed on $G_i$, upon a query, similarly to \cref{lm:simple-complete-algorithm}, we have $\Oh(k)$ ``interesting'' subgraphs $G_i$ (subgraphs that are either affected by the input edits or satisfy $\sedk(X_i) > 0$) that we process one-by-one in time $\tOh((u_i + \sedk(X_i)) \cdot k^{2 - \gamma})$ each using a generalization of \cref{lm:box-with-one-string-changing}.
The super-additivity property of $\sedk$ guarantees that this sums up to a total of $\tOh(k^{3 - \gamma})$.
The remaining subgraphs $G_i$ are not affected by the input edits and satisfy $\sedk(X_i) = 0$, which means that $X_i$ has a period of at most $k$.
For such subgraphs $G_i$, we can afford the complete preprocessing and, upon a query, processing such subgraphs in chunks in $\tOh(k^2)$ time in total precisely as in the proof of \cref{lm:simple-complete-algorithm}.
Furthermore, similarly to \cref{lm:simple-complete-algorithm}, all the information the data structure stores is either local or cumulative, and thus block edits can be supported straightforwardly.
This rather simple high-level approach runs into several low-level technical complications that are discussed in detail in \cref{app:full-algorithm}.

\section{Lower Bounds: Technical Overview}\label{sec:lower-bounds}
As discussed in \cref{sec:introduction}, our conditional lower bounds stem from the following result:
\thmstaticlb*

The reduction provided in~\cite{CKW23} is insufficient because the instances it produces satisfy $\ed(X,Y)=\Theta(k)$ rather than $\ed(X,Y)\le 4$.
Still, we reuse hardness of the following problem:

\begin{problem}[Batched Weighted Edit Distance~\cite{CKW23}]\label{prob:batched}
 Given a \emph{batch} of strings $X_1,\ldots,X_{m}\in \Sigma^x$, a string $Y\in \Sigma^y$, a threshold $k \in \Rz$, and oracle access to a weight function $w:\Esigma^2\to \Rz$, decide if $\min_{i=1}^m\wed(X_i,Y) \le k$.
\end{problem}

The hard instances of \cref{prob:batched}, which cannot be solved in $\Oh(x^{2-\delta}\sqrt{m})$ time for any $\delta>0$ assuming the APSP Hypothesis, satisfy many technical properties (see \cref{thm:BatchedWEDLBReduction}) including $h\coloneqq \max_{i=1}^{m-1} \hd(X_i,X_{i+1})=\Oh(x/m)$, where $\hd(X_i, X_{i+1})$ is the Hamming distance, i.e., the number of positions in which $X_i$ differs from $X_{i + 1}$.

The approach of~\cite[Section 7]{CKW23} is to construct strings $\bar{X}$ and $\bar{Y}$ that take the following form if we ignore some extra gadgets: $\bar{X} = X_1 Y X_2 \cdots X_{m-1} Y X_m$ and $\bar{Y} = X_0^{\bot} Y X_1^\bot \cdots X_{m-1}^\bot Y X_m^\bot$, where $X_i^\bot\in \Sigma^x$ is chosen so that $\wed(X_i^\bot,X_{i-1}) = \wed(X_i^\bot,X_{i})=h$.
The ignored gadgets let us focus on alignments $\cA_i : \bar{X} \onto \bar{Y}$ that satisfy the following properties for $i\in \fragmentcc{1}{m}$:
\begin{itemize}
    \item the fragment $X_i$ is aligned with the copy of $Y$ in $\bar{Y}$ located between $X_{i-1}^\bot$ and $X_{i}^\bot$;
    \item the $m-1$ copies of $Y$ in $\bar{X}$ are matched with the remaining $m-1$ copies of $Y$ in $\bar{Y}$;
    \item all the characters within the fragments $X_{i-1}^{\bot}$ and $X_{i}^\bot$ of $\bar{Y}$ are inserted;
    \item for $j\ne i$, the fragment $X_j$ is aligned with $X_{j-1}^\bot$ if $j<i$ and with $X_{j}^\bot$ if $j>i$.
\end{itemize}
This ensures that the cost of $\cA_i$ is equal to a baseline cost (independent of $i$) plus $\wed(X_i,Y)$.
\newcommand{\rdol}{\textcolor{red}{\$}}

Our reduction uses three types of `$X$ gadgets': $X_i$, $X^\bot_i$, and $X^\top_i$, with the latter two playing similar roles.
We also introduce a new symbol $\rdol$ that is extremely expensive to delete, insert, or substitute with any other symbol.
Modulo some extra gadgets, our strings are of the following form:
\begin{align*}
\tilde X &= X^\top_0 Y X^\bot_0 Y \rdol \bigodot_{i=1}^{m-1} \Big ( X_i Y X^\top_i Y X^\bot_i Y \Big ) X_m Y X^\top_m Y X^\bot_m \rdol, \\
\tilde Y &= \textcolor{red}{\$} X^\top_0 Y X^\bot_0 Y \bigodot_{i=1}^{m-1} \Big ( X_i Y X^\top_i Y X^\bot_i Y \Big ) X_m Y X^\top_m \textcolor{red}{\$} Y X^\bot_m.
\end{align*}
Notice that removing the two $\rdol$ symbols from $\tilde X$ or from $\tilde Y$ results in the same string, so we have $\ed(\tilde{X},\tilde{Y})\le4$.
Due to the expensive cost for editing a $\rdol$, any optimal alignment must match the two $\rdol$ symbols in $\tilde X$ with the two $\rdol$ symbols in $\tilde Y$, deleting a prefix of $\tilde X$ and a suffix of $\tilde Y$ for some predictable cost.
Thus, we can focus our analysis on $\wwed{w}{\hat X}{\hat Y}$, where 
\[
    \hX = X_1 Y X_1^\top Y X_1^\bot Y \cdots X_m Y X^\top_m Y X^\bot_m \quad\text{and}\quad
    \hY = X^\top_0 Y X^\bot_0 Y X_1 Y \cdots X_{m-1}^\bot Y X_m Y X^\top_m.
\]
Observe that both $\hX$ and $\hY$ consists of `$X$ gadgets' interleaved with $Y$, and that $\hY$ contains one more copy of $Y$ and one more `$X$ gadget'. 
Analogously to~\cite{CKW23}, the ignored extra gadgets let us focus on the alignments $\cA_i : \hX \onto \hY$ satisfying the following properties for $i\in \fragmentcc{1}{m}$:
\begin{itemize}
    \item the fragment $X_i$ in $\hX$ is aligned with the copy of $Y$ in $\hY$ located between $X_{i-1}^\bot$ and $X_{i-1}^\top$;
    \item the $3m-1$ copies of $Y$ in $\hX$ are matched with the remaining $3m-1$ copies of $Y$ in $\hY$;
    \item all characters within the fragments $X_{i-1}^\bot$ and $X_{i-1}^\top$ of $\hY$ are inserted;
    \item the remaining `$X$ gadgets' in $\hX$ are aligned with the remaining `$X$ gadgets' in $\hY$.
\end{itemize}
To perfectly mimic~\cite{CKW23}, we should ensure that all the `$X$ gadgets' that we can possibly align are at weighted edit distance $h$.
We only managed to construct $X_i^\bot$ and $X_i^\top$ so that:
\begin{itemize}
    \item $\wed(X_i,X_{i-1}^\top) = \wed(X_i,X_{i-1}^\bot) = \wed(X_i^\top,X_i)=\wed(X_i^\bot,X_i)=2h$, and
    \item $\wed(X_i^\top, X_{i-1}^\bot)=\wed(X_{i}^\bot, X_{i}^\top)=4h$.
\end{itemize}
Although we have two different distances between the relevant pairs of `$X$ gadgets', it is easy to see the number of pairs of either type is constant across all the alignments $\cA_i$.
This is sufficient to guarantee that the cost of $\cA_i$ is equal to some baseline cost plus $\wed(X_i,Y)$. 
Moreover, as in~\cite{CKW23}, the costs of alignments $\cA_i$ is $\Oh(x+mh)=\Oh(x)$ and, if we denote $n\coloneqq |\tilde X|=|\tilde Y|$, then the lower bound derived from the lower bound for \cref{prob:batched} excludes running times of the form $\Oh(x^{2-\delta}\sqrt{m})=\Oh(x^{2-\delta}\sqrt{n/x})=\Oh(\sqrt{nx^{3-2\delta}})$.

\paragraph*{Lower Bounds for Dynamic Weighed Edit Distance}\label{sec:lower-bounds:dynamic}
From \cref{thm:static-lb} we derive two dynamic lower bounds, each justifying a different limitation of our algorithm given in \cref{thm:simple_tradeoff}.
Our first lower bound, formalized as \cref{lm:dynamic-lb-variable-k}, concerns our choice to fix a threshold $k$ at the preprocessing phase.
If, instead, we aimed for time complexities depending on the current value of $\wed(X, Y)$, then, for sufficiently large values of $\wed(X, Y)$, we could not achieve update times improving upon the static algorithm.

Our second dynamic lower bound justifies the specific trade-off between the preprocessing and update time in \cref{thm:simple_tradeoff}.
In simple words, we prove that, with preprocessing time $\Ohtilde(nk^{\gamma})$ for $\gamma \in [0.5,1)$, the best possible update time is $\Oh(k^{3 - \gamma-o(1)})$.
For that, we note that \cref{thm:static-lb} is based on a reduction from the Negative Triangle problem, which is \emph{self-reducible}: solving $m$ instances of bounded weighted edit distance from \cref{thm:static-lb} is hence, in general, $m$ times harder than solving a single instance. 
Given such $m$ instances $(X_0, Y_0), \ldots, (X_{m-1}, Y_{m-1})$, we initialize a dynamic algorithm with a pair of equal strings $\hX = \hY = \bigodot_{i=0}^{m-1} (X_i \cdot \fancysymbol)$, where $\fancysymbol$ is an auxiliary symbol that is very expensive to edit.
For $i \in \fragmentco{0}{m}$, in $\ed(X_i, Y_i) = \Oh(1)$ updates, we can transform the fragment $X_i$ of $\hY$ into $Y_i$ and retrieve $\wed_{\le k}(\hX, \hY) = \wed_{\le k}(X_i, Y_i)$.
By applying and reverting these updates for every $i \in \fragmentco{0}{m}$, we can thus find $\wed_{\le k}(X_i, Y_i)$ for all $i$.
If we pick $m \coloneqq n / k^{3 - 2\gamma-2\delta}$ for an arbitrary small constant $\delta > 0$, then the static lower bound requires $m \cdot (\sqrt{(n / m) k^3})^{1-o(1)} \ge (n k^{\gamma+\delta})^{1-o(1)}$ total time.
Our preprocessing time is asymptotically smaller, so, among $\Oh(m)$ updates, some must take $\Omega((\sqrt{(n / m) k^3})^{1-o(1)})=\Omega(k^{3-\gamma-\delta-o(1)})$ time.
See \cref{lm:dynamic-lb-fixed-k} for the formal proof.


\bibliography{refs}

\appendix

\section{Deferred Proofs from Section \ref{sec:simple-algorithm}}\label{app:simple-algorithm}

In this section we formally prove \cref{lm:box-with-one-string-changing,lm:simple-complete-algorithm}.

\subsection{Proof of Lemma~\ref{lm:box-with-one-string-changing}}\label{app:simple-algorithm:box-with-one-string-changing}

\lmboxwithonestringchanging*

\begin{proof}
\textbf{Preprocessing.}
    We iterate over $c \in \fragmentcc{0}{|Y|}$ and apply Klein's algorithm (\cref{thm:klein}) for $\oAGw(X, Y\fragmentco{\ell_c}{c})$ and $\oAGw(X, Y\fragmentco{c}{r_c})$, where $\ell_c = \max(0, c - 2^{e_c})$ and $r_c = \min(|Y|, c + 2^{e_c})$ are defined based on $2^{e_c}$ -- the largest power of $2$ that divides $c$ (with $e_0 = \ceil{\log |Y|})$.
    This costs $\Oh(|X| \cdot 2^{e_c} \log (|X|\cdot 2^e))$ time, for a total of $\Oh(|X| \cdot |Y| \log^2 (|X| + |Y|))$ across all $c \in \fragmentcc{0}{|Y|}$.

    We call a fragment $Y\fragmentco{i}{j}$ \emph{simple} if $j \le r_i$ or $i \ge \ell_j$.
    Note that, for every simple fragment $Y\fragmentco{i}{j}$, Klein's algorithm (\cref{thm:klein}) gives $\Oh(\log (|X| + |Y|))$-time oracle access to $\BMw(X, Y\fragmentco{i}{j})$.
    Moreover, an arbitrary fragment $Y\fragmentco{i}{j}$ can be decomposed into two simple fragments $Y\fragmentco{i}{c}$ and $Y\fragmentco{c}{j}$ by picking $c \in \fragmentcc{i}{j}$ that maximizes $e_c$.

\smallskip

\textbf{Distance Propagation Query.}
    We reinterpret the problem of computing $v^T \otimes \BMw(X, Y')$ as follows.
    Let $G$ be the graph $\oAGw(X, Y')$, to which a source node $r$ is added with outgoing edges to all input vertices of $X \times Y'$ with weights specified by $v$.
    The row $v^T \otimes \BMw(X, Y')$ consists of the distances from $r$ to the output vertices of $X \times Y'$ in $G$.

    If $|Y'| > 2 |Y|$, we have $\ed(Y, Y') = \Theta(|Y'|)$.
    In this case, we solve the problem via a standard dynamic programming in time $\Oh(|Y'| \cdot |X|)$ using the monotonicity property of \cref{lm:oAGw-properties}.

    Otherwise, assume that $|Y'| \le 2 |Y|$.
    We first determine the distances from $r$ to all the input vertices of $X \times Y'$ in $\Oh(|X| + |Y'|)$ time.
    Then, we compute $\ed(Y, Y')$ and an optimal unweighted alignment $\cA:Y\onto Y'$ in time $\Oh((\ed(Y, Y') + 1) \cdot |Y|)$ using \cref{lm:baseline-wed}.
    Given the alignment $\cA$, in $\Oh(|Y|)$ time, we decompose $Y'$ into at most $\ed(Y, Y')$ single characters and at most $\ed(Y, Y')+1$ fragments of $Y'$ exactly matching fragments of $Y$.
    Every fragment of $Y$ can be decomposed into at most two simple fragments, so we can compute a decomposition $Y'=\bigodot_{i=0}^{m-1} Y'_i$ into $m \le 3\ed(Y, Y') + 2$ phrases, where every phrase of length at least two matches a simple fragment of $Y$. 
     
    We process the decomposition of $Y'$ in $m$ steps.
    Our invariant is that, after $i \in \fragmentcc{0}{m}$ iterations, we have the distances from $r$ to the output vertices of $X \times (Y'_0 \cdots Y'_{i-1})$ in $G$.
    For $i = 0$, these are the top input vertices of $X \times Y'$, to which we already computed the distances; for $i = m$, these are exactly the output vertices of $X \times Y'$.
    Say, we already processed the first $i \in \fragmentco{0}{m}$ phrases of the decomposition of $Y'$, and now we process $Y'_i = Y'\fragmentco{y'_i}{y'_{i+1}}$.
    We already computed the distance from $r$ to the vertices $(|X|, y)$ for $y \in \fragmentcc{0}{y'_i}$.
    It remains to compute the distances from $r$ to the output vertices of $R \coloneqq \fragmentcc{0}{|X|} \times \fragmentcc{y'_i}{y'_{i+1}}$.
    If $Y'_i$ is a single character, we compute the distances in $\Oh(|X|)$ time using standard dynamic programming and the monotonicity property of \cref{lm:oAGw-properties}.
    Otherwise, $Y'_i$ matches a simple fragment of $Y$, and Klein's algorithm (\cref{thm:klein}) provides $\Oh(\log (|X| + |Y|))$-time oracle access to $\BMw(X, Y'_i)$. 
    Let $u$ be the already computed vector of distances from $r$ to the input vertices of $R$.
    We compute the distances from $r$ to the output vertices of $R$ as the product $u^T \otimes \BMw(X, Y'_i)$ in $\Oh((|X| + |Y'_i|) \log (|X| + |Y|))$ time using SMAWK algorithm (\cref{thm:smawk}).

    We process each of the $m = \Oh(\ed(Y, Y') + 1)$ phrases $Y'_i$ in $\Oh((|X| + |Y|) \log (|X| + |Y|))$ time, for a total of $\Oh((\ed(Y, Y') + 1) \cdot (|X| + |Y|) \log (|X| + |Y|))$.

\smallskip

\textbf{Path Query.}
    If $|Y'| > 2 |Y|$, we compute the path trivially in $\Oh(|Y'| \cdot |X|)$ time.
    Otherwise, we run the same algorithm as for distance propagation queries, this time setting the entries of $v$ to be equal to the distances in $\oAGw(X, Y')$ from $a$ to the remaining input vertices.
    Next, we compute the distances from $a$ to all output vertices of $X \times Y'$, then backtrack the path from $a$ to $b$ into individual subpaths inside $X \times Y'_i$, and finally use Klein's algorithm (\cref{thm:klein}) to find the paths inside these subgraphs.
\end{proof}

\subsection{Proof of Theorem~\ref{lm:simple-complete-algorithm}}\label{app:simple-algorithm:simple-complete-algorithm}

Before proving \cref{lm:simple-complete-algorithm}, we first formally introduce the dynamic range composition data structure that allows us to quickly propagate distances over a range of unaffected subgraphs~$G_i$.

\begin{definition}[Monge Matrix Multiplication Semigroup]\label{def:matrix-monoid}
    For a positive integer $k$, let $\matmon_{\le k}$ be the set containing a special symbol~$\absorbingmatrix$ and all real-valued Monge matrices with dimensions at most~$k$.
    We define the $\xotimes$ operator on the set $\matmon_{\le k}$.
    For two Monge matrices $A, B \in \matmon_{\le k}$, if they can be min-plus multiplied, we set $A \xotimes B \coloneqq A \otimes B$.
    Otherwise, we set $A \xotimes B \coloneqq \absorbingmatrix$.
    Furthermore, we set $A \xotimes \absorbingmatrix = \absorbingmatrix \xotimes A = \absorbingmatrix \xotimes \absorbingmatrix \coloneqq \absorbingmatrix$ for all Monge matrices $A \in \matmon_{\le k}$.
    We call $(\matmon_{\le k}, \xotimes)$ the \emph{$k$-restricted Monge matrix multiplication semigroup}.
\end{definition}

The elements of $(\matmon_{\le k}, \xotimes)$ can be multiplied in $\Oh(k^2)$ time using SMAWK algorithm (\cref{thm:smawk}).

\begin{lemma}[Folklore]\label{fct:dynamic-range-composition-ds}
    Let $(\Sem, \times)$ be a semigroup in which the product of every two elements can be computed in time $T$.
    There is a fully persistent \emph{Dynamic Range Composition} data structure for a list $L \in \Sem^+$ that supports the following operations:
    \begin{description}
        \item[Build:] Given a list $L \in \Sem^+$, build the dynamic range composition data structure for $L$ in $\Oh(T\cdot |L|)$ time.
        \item[Split:] Given a dynamic range composition data structure for a list $L \in \Sem^{\ge 2}$ and an index $i \in \fragmentco{1}{|L|}$, build the dynamic range composition data structures for $L\fragmentco{0}{i}$ and $L\fragmentco{i}{|L|}$ in $\Oh(T \log |L|)$ time.
        \item[Concatenate:] Given dynamic range composition data structures for lists $L, L' \in \Sem^+$, build the dynamic range composition data structure for $L \cdot L'$ in $\Oh(T \log (|L| + |L'|))$ time.
        \item[Query:] Given a dynamic range composition data structure for a list $L \in \Sem^+$ and indices $\ell, r \in \fragmentcc{0}{|L|}$ satisfying $\ell < r$, in time $\Oh(1 + \log |L|)$ return (references to) a sequence of $m = \Oh(1 + \log |L|)$ elements $\row{S}{0}{m-1} \in \Sem$ such that $L_{\ell} \times \cdots \times L_{r-1} = S_0 \times \cdots \times S_{m-1}$.
    \end{description}
\end{lemma}

\begin{proof}[Proof Sketch]
    The data structure maintains a purely functional self-balancing binary tree over the elements of the list $L$.
    We use indices $i \in \fragmentco{0}{|L|}$ as (implicit) keys for the binary tree.
    Each vertex $i$ stores some element $L_i$ and the product $L_{\ell} \times \cdots \times L_{r - 1}$, where $L\fragmentco{\ell}{r}$ is the range of keys in the subtree of vertex $i$.
    For the build operation, we build the self-balancing binary tree and compute the subtree products of all vertices in the bottom-up manner.
    It takes $\Oh(|L| + T \cdot |L|)$ time.

    Split and concatenate operations are implemented using the self-balancing binary tree split and join operations.
    Whenever a new vertex is created, we immediately compute the value in it as the product of the values in its children and the local element $L_i$.
    This increases the time complexity of each self-balancing binary tree operation by a factor of $T$.
    For a query, we return the minimum cover of the queried segment in logarithmic time similarly to how queries to segment trees are handled.
\end{proof}

We now prove \cref{lm:simple-complete-algorithm}.

\lmsimplecompletealgorithm*

\begin{proof}
    As we are focused on computing $\wed_{\le k}$, we may cap the weight function $w$ by $k+1$ from above and without loss of generality assume that we have $w \colon \Esigma^2 \to [0, k + 1]$.\footnote{Note that this is a necessary step because the augmented alignment graphs $\oAGw$ are only defined for capped weight functions.}
    Let $n \coloneqq |X|$.
    We first assume that $n > k$ holds throughout the lifetime of the data structure.
    Call a sequence $\row{X}{0}{m-1}$ where $X_i = X\fragmentco{x_i}{x_{i+1}}$ a \emph{$k$-partitioning} of $X$ if $x_0 = 0$, $x_m = n$, and $x_{i+1}-x_i \in \fragmentco{k}{2k}$ for all $i \in \fragmentco{0}{m}$.
    Fix a $k$-partitioning $\row{X}{0}{m-1}$ of $X$.
    Furthermore, define a sequence of fragments $(Y_i)_{i=0}^{m-1}$ such that $Y_i = X\fragmentco{y_i}{y'_{i+1}}$ for $y_i = \max(x_i - 2k, 0)$ and $y'_{i+1} = \min(x_{i+1} + 2k, \abs{X})$ for all $i \in \fragmentco{0}{m}$.
    Let $G_i$ be the subgraph of $\oAGw(X, X)$ induced by $\fragmentcc{x_i}{x_{i+1}}\times \fragmentcc{y_i}{y'_{i+1}}$.
    Denote by $G$ the union of all subgraphs $G_i$.
    For each $i \in \fragmentcc{0}{m}$, denote $V_i \coloneqq \setof{(x, y) \in V(G)}{x = x_i, y \in \fragmentcc{y_i}{y'_i}}$, where $y'_0 = 0$ and $y_m = \abs{X}$.
    Moreover, for $i, j \in \fragmentcc{0}{m}$ with $i \le j$, let $D_{i, j}$ denote the matrix of pairwise distances from $V_i$ to $V_j$ in $G$, where the rows of $D_{i, j}$ represent the vertices of $V_i$ in the decreasing order of the second coordinate, and the columns of $D_{i, j}$ represent the vertices of $V_j$ in the decreasing order of the second coordinate.

    Our dynamic data structure maintains a $k$-partitioning $\row{X}{0}{m-1}$ of $X$, the data structures of \cref{lm:klein-upgraded,lm:box-with-one-string-changing} for every pair $X_i, Y_i$, the dynamic range composition data structure (\cref{fct:dynamic-range-composition-ds}) over matrices $D_{i, i + 1}$ with $(\matmon_{\le 4k+1}, \xotimes)$ as the semigroup, and a self-balancing binary tree over $(x_i)_{i=0}^{m}$.

    \subparagraph*{Initialization.} We start the initialization phase by building an arbitrary $k$-partitioning $\row{X}{0}{m-1}$ of $X$.
    For every $i \in \fragmentco{0}{m}$, we initialize the algorithms of \cref{lm:klein-upgraded,lm:box-with-one-string-changing} for $X_i$ and $Y_i$.
    We then use \cref{lm:klein-upgraded} to build all the matrices $D_{i, i + 1}$ and initialize the dynamic range composition data structure over $(D_{i, i + 1})_{i=0}^{m-1}$.
    Finally, we build a self-balancing binary tree over $(x_i)_{i=0}^{m}$.
    In total, the initialization takes $\Oh(m \cdot k^2 \log^2 k + m \cdot k^2 + m) = \Oh(n k \log^2 n)$ time.

    \subparagraph*{Update.} An edit in $X$ corresponds to an edit in one of the fragments $X_i$ from the $k$-partitioning.
    For all $j \in \fragmentco{0}{m} \setminus \fragmentcc{i-2}{i+2}$, the subgraphs $G_j$ do not change.
    For $X\fragmentco{x_{\max(i - 2, 0)}}{x_{\min(i + 3, m)}}$, we perform the initialization from scratch and change the $k$-partitioning of $X$ accordingly in $\Oh(k^2 \log^2 n)$ time.
    This leads to a constant number of updates (splits and concatenations) to the dynamic range composition data structure taking $\Oh(k^2 \log n)$ time each.
    Finally, we minimally update the self-balancing binary tree over $(x_i)_{i=0}^{m}$ to reflect the new $k$-partitioning of $X$.
    This requires a constant number of joins, splits, and range updates.
    The total time complexity of the update is $\Oh(k^2 \log^2 n)$.

    \subparagraph*{Query.}
    Call the \emph{input edits} the $u$ edits transforming $X$ into $Z$ given to the algorithm.
    If $u = 0$, we have $X = Z$, and hence can return a trivial alignment of weight $0$.
    Otherwise, we analyze how the fragments $Y_i$ and the subgraphs $G_i$ would change if in $\oAGw(X, X)$ we would fix the first string and transform the second string into $Z$ using the input edits.
    Let $G'$ be the new graph $G$, let $V'_i$ be the new sets $V_i$, and let $D'_{i, j}$ be the new matrices $D_{i, j}$.
    As $G$ contains all vertices $(x, y)$ of $\oAGw(X, X)$ with $|x - y| \le 2k$, the graph $G'$ contains all vertices $(x, y)$ of $\oAGw(X, Z)$ with $|x - y| \le k$, and thus $\wed(X, Z) = \dist_{G'}((0, 0), (|X|, |Z|))$ holds if $\wed(X, Z) \le k$ due to \cref{lm:paths_dont_deviate_too_much}.
    Hence, it is sufficient to compute the shortest path from $(0, 0)$ to $(|X|, |Z|)$ in $G'$.
    Every input edit affects a constant number of subgraphs $G_i$, and thus in total $\Oh(u)$ subgraphs change.
    Let $F$ be the set of indices $i \in \fragmentco{0}{m}$ of all affected subgraphs $G_i$.
    Note that $\dist_{G'}((0, 0), (|X|, |Z|))$ is equal to the only entry of $D'_{0, m}$.
    Furthermore, for all $a, b, c \in \fragmentcc{0}{m}$ with $a \le b \le c$, we have $D'_{a, c} = D'_{a, b} \otimes D'_{b, c}$.
    We compute $D'_{0, m}$ as $\bigotimes_{i=0}^{m-1} D'_{i, i + 1}$.
    Let $v_i = \bigotimes_{j=0}^{i-1} D'_{j, j + 1}$ for all $i \in \fragmentcc{0}{m}$.
    Initially we have $v_0$ which is a row with a single value $0$.
    We want to compute $v_m$.
    Note that, for every $i \notin F$, we have $D'_{i, i + 1} = D_{i, i + 1}$.
    Therefore, for all pairs $a, b$ of consecutive elements in $F$, given $v_{a+1}$, we can compute $v_b$ in time $\Oh(k \log n)$ by querying the dynamic range composition data structure over $(D_{i, i + 1})_{i \in \fragmentco{0}{m}}$ and applying SMAWK algorithm (\cref{thm:smawk}).
    Furthermore, for every $a \in F$, given $v_a$ we can compute $v_{a+1}$ in time $\Oh((u_a + 1) \cdot k \log n)$ time using \cref{lm:box-with-one-string-changing}, where $u_a$ is the number of input edits that affect $Y_a$.
    As we have $\sum_{a \in F} u_a = \Oh(u)$ and $|F| = \Oh(u)$, the total time to compute $v_m$ is $\Oh(u k \log n)$.

    We computed $\wed_{\le k}(X, Z)$.
    If $\wed(X, Z) \le k$, we now compute a $w$-optimal alignment of $X$ onto~$Z$.
    For that, we backtrack the computation of $v_m$.
    For all $a \in F$, we can compute the vertices in $V_a$ and $V_{a+1}$ through which the optimal path goes.
    Given such nodes, we can compute the breakpoint representation of a shortest path between them using \cref{lm:klein-upgraded}.
    Furthermore, for all pairs $a, b$ of consecutive elements of $F$, given the vertices in $V_{a+1}$ and $V_b$, we can backtrack down the balanced tree of \cref{fct:dynamic-range-composition-ds} (similarly to the proof of \cref{lm:klein-upgraded}) and use \cref{lm:klein-upgraded} inside individual subgraphs $G_i$ to recover the breakpoint representation of a shortest path between the vertices in $V_{a+1}$ and $V_b$.
    As the total weight of the backtracked path is at most $k$, the whole path reconstruction procedure takes $\Oh(k \log^2 n)$ time.

    It remains to drop the assumption that $n > k$ holds at all time.
    If at some point $n\le k$, we drop the dynamic algorithm described above and, after every update, construct the data structure of \cref{lm:box-with-one-string-changing} for $Y = X$ in $\Oh(k^2 \log^2 (n + 2))$ time.
    Now, each query costs $\Oh(u k \log (n + 2))$ time.
    As soon as $n > k$, we initialize the regular dynamic algorithm in $\Oh(k^2 \log^2 (n + 2))$ time and continue.
\end{proof}

\section{Trade-Off Algorithm}\label{app:full-algorithm}

\subsection{\boldmath Preliminaries: \pillar Model, Self-Edit Distance, and $k$-Shifted Self-Edit Distance} \label{app:full-algorithm-preliminaries}

We start by defining some useful concepts for our trade-off algorithm.

\paragraph*{The PILLAR Model}
As discussed in \cref{sec:simple-algorithm}, it is helpful to have a data structure that provides $\tOh(1)$-time equality tests between fragments of $X$ and $Z$ to find an optimal unweighted alignment of $X$ onto $Z$ upon every query.
In this section, we will use such operations not only as a preparation step for a query but also inside the query computation.
For this purpose, we will use the \pillar model introduced by Charalampopoulos, Kociumaka, and Wellnitz~\cite{CKW20}.
The \modelname{} model provides an abstract interface to a set of primitive operations on
strings which can be efficiently implemented in different settings. Thus, an algorithm
developed using the \modelname{} interface does not only yield algorithms in the standard
setting, but also directly yields algorithms in diverse other settings,
for instance, fully compressed, dynamic, etc.

In the~\modelname{} model, we are given a family $\calX$ of strings to preprocess. The
elementary objects are fragments $X\fragmentco{\ell}{r}$ of strings $X \in \calX$.
Initially, the model gives access to each $X \in \calX$, interpreted as $X\fragmentco{0}{|X|}$.
Other fragments can be retrieved via an \extractOpName{} operation:
\begin{itemize}
    \item $\extractOpName(S, \ell, r)$: Given a fragment $S$ and positions $0 \leq \ell \leq r \leq |S|$, extract
    the fragment $S\fragmentco{\ell}{r}$, which is defined as $X\fragmentco{\ell' + \ell}{\ell' + r}$ if $S = X\fragmentco{\ell'}{r'}$ for
    $X \in \calX$.
\end{itemize}
Moreover, the \modelname model provides the following primitive queries for fragments $S$ and $T$~\cite{CKW20}:
\begin{itemize}
    \item $\lceOp{S}{T}$: Compute the length of~the longest common prefix of~$S$ and $T$.
    \item $\lcbOp{S}{T}$: Compute the length of~the longest common suffix of~$S$ and $T$.
    \item $\accOpName(S,i)$: Assuming $i\in \fragmentco{0}{|S|}$, retrieve the character $\accOp{S}{i}$.
    \item $\lenOpName(S)$: Retrieve the length $|S|$ of~the string $S$.
\end{itemize}
Observe that, in the original definition~\cite{CKW20}, the \modelname{} model
also includes an \(\ipmOpName\) operation that finds all exact occurrences of one fragment within another. 
We do not need the \(\ipmOpName\) operation in this
work.\footnote{We still use the name \modelname, and not {\tt PLLAR}, though.}
\pillar operations can be implemented to work in constant time in the static setting \cite[Section 7.1]{CKW20} and in near-logarithmic time in the dynamic setting \cite[Section 8]{KK22}, \cite[Section 7.2]{CKW22}.

We use the following version of the classic Landau--Vishkin algorithm \cite{LV88} in the \modelname{} model.

\begin{fact}[{\cite[Lemma 6.1]{CKW20}}]\label{lm:k2-ed}
    There is a \modelname{} algorithm that, given strings $X, Y \in \Sigma^*$, in time $\Oh(k^2)$ computes $k = \ed(X, Y)$ along with the breakpoint representation of an optimal unweighted alignment of $X$ onto $Y$.
\end{fact}

\paragraph*{Self-Edit Distance}
Recent developments \cite{CKW23,GJKT24,GK24} on edit distance allow reducing arbitrary instances of bounded edit distance to instances where the input strings are compressible.
The exact compressibility measure used is the so-called self-edit distance.

\begin{definition}[Self-Edit Distance {\cite{CKW23}}] \label{def:self-ed}
    Say that an alignment $\cA : X \onto X$ is a \emph{self-alignment}
    if $\cA$ does not align any character $X\position{x}$ to itself.
    We define the \emph{self-edit distance} of $X$ as $\sed(X) \coloneqq \min_\cA \ed_{\cA}(X, X)$,
    where the minimization ranges over all self-alignments $\cA : X \onto X$. In words,
    $\sed(X)$ is the minimum (unweighted) cost of a self-alignment.
\end{definition}

We can interpret a self-alignment as a
$(0, 0) \leadsto (|X|, |X|)$ path in the alignment graph $\AG(X, X)$
that does not contain any edges of the main diagonal.

Similarly to \cref{lm:k2-ed}, there is a \modelname{} algorithm that computes $\sed(X)$ if this value is bounded by $k$, in time $\Oh(k^2)$.

\begin{fact}[{\cite[Lemma 4.5]{CKW23}}] \label{lm:selfed}
    There is an $\Oh(k^2)$-time \modelname{} algorithm that, given a string $X\in \Sigma^*$
    and an integer $k\in \ZZ_{> 0}$, determines whether $\sed(X)\le k$ and, if so,
    retrieves the breakpoint representation of an optimal self-alignment $\cA \colon X \onto X$.
\end{fact}

\paragraph*{\boldmath $k$-Shifted Self-Edit Distance}

As discussed in \cref{sec:full-algorithm}, our algorithms on small self-edit distance strings break the relevant part of the alignment graph $\oAGw(X, X)$ into subgraphs and use the fact that, if $\sed(X) \le k$, all but $\Oh(k)$ subgraphs are ``highly compressible''.
Note that this would require our compressibility measure to be super-additive ($\sed(X_1) + \sed(X_2) \le \sed(X_1 \cdot X_2)$) while, in reality, self-edit distance is sub-additive ($\sed(X_1) + \sed(X_2) \ge \sed(X_1 \cdot X_2)$ \cite[Lemma 4.2]{CKW23}).
To circumvent this issue, we define a better-suited notion of compressibility we call $k$-shifted self-edit distance.

\begin{definition}\label{def:sedk}
    Let $X \in \Sigma^*$ be a string and $G$ be the alignment graph $\AG(X, X)$ of $X$ onto itself with all the edges on the main diagonal removed.
    For an integer $k \ge 1$, we define the \emph{$k$-shifted self-edit distance} $\sedk(X)$ of $X$ as
    \[\min\setof{\dist_G((x, 0), (|X|, y))}{x \in \fragmentcc{0}{\min(|X|, k)}, y \in \fragmentcc{\max(0, |X| - k)}{|X|}}.\]
\end{definition}

In contrast with self-edit distance, $k$-shifted self-edit distance is super-additive under mild constraints.

\begin{lemma}\label{lm:superadditivity-of-sedk}
    Let $k \ge 1$ be an integer, $X \in \Sigma^*$ be a string with $\sed(X) \le k$, and $\row{X}{0}{\ell-1}$ be a partitioning of $X$.
    We have $\sum_{i=0}^{\ell-1} \sedk(X_i) \le \sed(X)$.
\end{lemma}

\begin{proof}
    Consider the optimal self-alignment $\cA$ of $X$ onto itself.
    Due to symmetry, we may assume that it always stays on or above the main diagonal of $\AG(X, X)$.
    Furthermore, according to \cref{lm:paths_dont_deviate_too_much}, we have $y-x\in \fragmentcc{0}{k}$ for all $(x,y) \in \cA$.
    Therefore, the restriction of $\cA$ to $X_i \times X_i$ costs at least $\sedk(X_i)$ for every $i$.
    As such restrictions are disjoint for different $i$, we have $\sum_{i=0}^{\ell-1} \sedk(X_i) \le \sed(X)$.
\end{proof}

Similarly to \cref{lm:k2-ed,lm:selfed}, there is a Landau--Vishkin-like algorithm for computing $\sedk$.

\begin{lemma}\label{lm:compute-sedk}
    There is an $\Oh(k^2)$-time \modelname{} algorithm that, given a string $X\in \Sigma^*$
    and an integer $k\in \ZZ_{> 0}$, determines whether $\sedk(X)\le k$ and, if so,
    retrieves the breakpoint representation of the underlying optimal path.
\end{lemma}

\begin{proof}
    Our algorithm is a modification of the classic Landau--Vishkin algorithm \cite{LV88} similar to the proof of \cref{lm:selfed}.
    Let $G$ be the alignment graph $\AG(X, X)$ where all the edges on the main diagonal are removed and all the edges of the form $(i, 0) \to (i+1, 0)$ for $i \in \fragmentco{0}{\min(|X|, k)}$ have weight $0$.
    Note that $\sedk(X) = \min\setof{\dist_G((0, 0), (|X|, j))}{j \in \fragmentcc{\max(0, |X| - k)}{|X|}}$ holds.
    If $|X| < k$, we compute all the distances in time $\Oh(|X| \cdot |X|) = \Oh(k^2)$ using dynamic programming.
    Otherwise, assume that $|X| \ge k$.
    Every path from $(0, 0)$ to $(|X|, j)$ in $G$ of weight at most $k$ stays within $k$ from the main diagonal.
    Due to symmetry, we may also assume that such a path does not go through any vertices below the main diagonal.
    For every distance value $j \in \fragmentcc{0}{k}$ and diagonal $i \in \fragmentcc{0}{2k}$, we compute an entry $D[i, j]$ which stores
    \[ D[i, j] \coloneqq \max\setof{x}{\text{there exists a path from $(0, 0)$ to $(x, x+i)$ in $G$ of cost at most $j$}}. \]
    We compute all the values $D[i, j]$ in the increasing order of $j$.
    For $j = 0$, only diagonals with $i \in \fragmentcc{0}{k}$ have vertices with distance $0$ from $(0, 0)$.
    To compute $D[\cdot, 0]$ we use the \lceOpName operation of \pillar to compute the furthest point with distance $0$.
    For $j \ge 1$ and $i \in \fragmentcc{0}{2k}$, we compute $D[i, j]$ from $D[i-1, j-1]$, $D[i, j-1]$, and $D[i+1, j-1]$ by making a step onto the $i$-th diagonal from the corresponding vertices and continuing along the diagonal for as long as possible using the \lceOpName operation.
    After computing the whole table $D$, we can find $\min(\dist_G((0, 0), (|X|, j)), k + 1)$ for all $j \in \fragmentcc{|X| - k}{|X|}$ using the values $D[|X|-j, \cdot]$.
    We then use $D$ to backtrack the optimal path in $\Oh(k)$ time.
    Correctness of the algorithm follows similarly to the correctness of the Landau--Vishkin algorithm.
\end{proof}

One of the main contributions of \cite{CKW23} is that strings of small self-edit distance can be decomposed into fragments of approximately equal length where almost all fragments already appear at most $\Theta(k)$ symbols before.
In the following two lemmas, we generalize these results to $k$-shifted self-edit distance.

\begin{lemma}[Variation of {\cite[Lemma 4.6 and Claim 4.11]{CKW23}}]\label{lem:decomp2}
There is a \modelname{} algorithm that, given strings $X,Y \in \Sigma^{*}$ such that $X$ is a substring of $Y$ and a positive integer $k$ such that $\sedk(X)\le k \le |X|$, in $\Oh(k^2)$ time builds a decomposition $X=\bigodot_{i=0}^{m-1} X_i$ and a sequence of fragments $(Y_i)_{i=0}^{m-1}$ of $Y$ such that
\begin{itemize}
    \item For $i\in \fragmentco{0}{m}$, each phrase $X_i=X\fragmentco{x_{i}}{x_{i+1}}$ is of length $x_{i+1}-x_i\in \fragmentco{k}{2k}$, and each fragment $Y_i=Y\fragmentco{y_i}{y'_{i+1}}$ satisfies $y_i=\max(x_i-k,0)$ and $y'_{i+1}=\min(x_{i+1}+3k,|Y|)$.
    \item There is a set $F \subseteq \fragmentco{0}{m}$ of size $|F|= \Oh((\sedk(X) + 1) \cdot ((|Y| - |X|) / k + 1))$ such that $X[x_i \dd x_{i+1})=X\fragmentco{x_{i-1}}{x_{i}}$ and $Y[y_i \dd y'_{i+1})=Y\fragmentco{y_{i-1}}{y'_{i}}$ holds for each $i\in \fragmentco{0}{m}\setminus F$ (in particular, $0\in F$).
\end{itemize}
The algorithm returns the set $F$ and, for all $i\in F$, the endpoints of $X_i=X\fragmentco{x_{i}}{x_{i+1}}$.\footnote{This determines the whole decomposition because $(x_i)_{i=\ell}^r$ is an arithmetic progression for every $\fragmentoc{\ell}{r}\subseteq \fragmentoc{0}{m}\setminus F$.}
\end{lemma}

\begin{proof}
    We apply \cref{lm:compute-sedk} to compute $\sedk(X)$ in time $\Oh(k^2)$.
    Consider the optimal path underlying the value $\sedk(X)$ and add a horizontal segment at the beginning of it starting at $(0, 0)$ and a vertical segment at the end of it ending at $(|X|, |X|)$.
    We then use this self-alignment of $X$ onto itself for \cite[Lemma 4.6]{CKW23}.
    Even though the weight of such a self-alignment is $\sedk(X) + \Oh(k)$, one can see that, following the algorithm of \cite[Lemma 4.6]{CKW23}, we get a decomposition of $X$ into phrases and a set $F$ of size $\Oh(\sedk(X) + 1)$ such that $X[x_i \dd x_{i+1})=X\fragmentco{x_{i-1}}{x_{i}}$ holds for each $i\in \fragmentco{0}{m}\setminus F$.
    We then use \cite[Claim 4.11]{CKW23} to extend the set $F$ so that $Y[y_i \dd y'_{i+1})=Y\fragmentco{y_{i-1}}{y'_{i}}$ holds for each $i\in \fragmentco{0}{m}\setminus F$.
    For that, it is sufficient to consider the set $F' \coloneqq \setof{i + j}{i \in F, j \in \fragmentcc{-\ell}{\ell}}$ for a sufficiently large $\ell = \Theta((|Y| - |X|) / k + 1)$, such that, for $\fragmentoo{a}{b} \in \fragmentco{0}{m} \setminus F'$, the period of $X\fragmentco{x_a}{x_b}$ is equal to the period of $Y\fragmentco{y_a}{y'_{b+1}}$.
    See \cite[Lemma 7.1]{GK24} for details.
\end{proof}

\begin{lemma}[Variation of {\cite[Lemma 4.12]{CKW23}}]\label{lem:decomp3}
There is a \modelname{} algorithm that, given a string $X \in \Sigma^{*}$ and integers $k, \ell \ge 1$ satisfying $\ell \le k / 2$ and $\sedk(X) \le k \le |X|$, in $\Oh(k^2)$ time builds a decomposition $X=\bigodot_{i=0}^{m-1} X_i$ such that:
\begin{itemize}
    \item For $i\in \fragmentco{0}{m}$, each phrase $X_i=X\fragmentco{x_{i}}{x_{i+1}}$ is of length $x_{i+1}-x_i\in \fragmentco{\ell}{2\ell}$.
    \item There is a set $F \subseteq \fragmentco{0}{m}$ of at most $\Oh(k / \ell + \sedk(X))$ fresh phrases such that every phrase $i\in \fragmentco{0}{m}\setminus F$ has a source phrase $i' \in \fragmentco{0}{i}$ with $X\fragmentco{x_i}{x_{i+1}}=X\fragmentco{x_{i'}}{x_{i'+1}}$ and $x_i - x_{i'} \le k$.
\end{itemize}
The algorithm returns the set $F$ and, for all $i\in F$, the endpoints of $X_i=X\fragmentco{x_{i}}{x_{i+1}}$.
\end{lemma}

\begin{proof}
    This variation of \cite[Lemma 4.12]{CKW23} follows similarly to \cref{lem:decomp2}.
    We first apply \cref{lm:compute-sedk} to compute $\sedk(X)$ in time $\Oh(k^2)$.
    We then consider the optimal path underlying the value $\sedk(X)$ and add a horizontal segment at the beginning of it starting at $(0, 0)$ and a vertical segment at the end of it ending at $(|X|, |X|)$.
    Next, we use this self-alignment of $X$ onto itself for the algorithm of \cite[Lemma 4.6]{CKW23}.
    The algorithm starts with a horizontal segment of length $\Oh(k)$ that is decomposed into $\Oh(k / \ell)$ (potentially fresh) phrases and then proceeds with $\Oh(|X| / k)$ phrases of which at most $\Oh(\sedk(X) + 1)$ are fresh.
\end{proof}

\paragraph*{Miscellaneous}

As we are only interested in computing $\wed_{\le k}$, any intermediate values we compute that exceed $k$ are irrelevant.
The following definition from \cite{GK24} captures this idea.

\begin{definition}[{\cite[Definition 4.11]{GK24}}]\label{def:k-equiv}
    Let $a, b, k \in \RR_{\ge 0}$.
    We say that $a$ and $b$ are \emph{$k$-equivalent} and write $a \meq{k} b$ if $a=b$ or $\min(a,b) > k$.
    Furthermore, let $A, B \in \RR_{\ge 0}^{p \times q}$ be matrices.
    We call these two matrices \emph{$k$-equivalent} and write $A \meq{k} B$ if $A_{i, j} \meq{k} B_{i, j}$ holds for all $i \in \fragmentco{0}{p}$ and $j \in \fragmentco{0}{q}$.
\end{definition}

\subsection{The Case of Small Self-Edit Distance}

As discussed in \cref{sec:full-algorithm}, we break the relevant part of the alignment graph $\oAGw(X, X)$ into subgraphs of size $\Theta(k^{2 - \gamma}) \times \Theta(k^{2 - \gamma})$.
If such a subgraph induced by some fragments $\hX \times \hY$ of $X$ satisfies $\sedk(\hX) = 0$, \cref{lem:decomp2} implies that we can split this subgraph into $\Theta(k) \times \Theta(k)$-sized pieces with only a constant number of distinct pieces, and thus we have time for the complete preprocessing as in \cref{lm:simple-complete-algorithm}.
Otherwise, if $\sedk(\hX) > 0$, we do not have time for complete preprocessing from the simple algorithm, and we break the subgraph into $\Theta(\ell) \times \Theta(\ell)$-sized pieces for an appropriately chosen $\ell \le k$ using \cref{lem:decomp3}.
This allows us to answer queries in time $\tOh(\sedk(\hY) \cdot k^{2 - \gamma})$, which adds up to $\tOh(k^{3 - \gamma})$ time in total due to the super-additivity property of $\sedk$.
The following lemma describes the data structure we store for the $\Theta(k^{2 - \gamma}) \times \Theta(k^{2 - \gamma})$-sized subgraphs of $\oAGw(X, X)$.

\begin{lemma}\label{lm:block-data-structure}
    Let $k \ge 1$ be an integer and $X, Y \in \Sigma^{\ge 2k}$ be strings such that $X$ is a substring of~$Y$.
    Suppose we are given $\Oh(1)$-time oracle access to a normalized weight function \wdef.
    There is a data structure that, after $\Oh(|Y| + k^2 \log^2 |Y|)$-time preprocessing of $X$, $Y$, and $k$, answers the following queries:
    \begin{description}
        \item[Distance Matrix Query:] Return an element $M \in \matmon_{\le k + 1}$ in $\Oh(k^2)$ time. If $\sedk(X) \ge 1$, we have $M = \absorbingmatrix$, and if $\sedk(X) = 0$, we have $M_{i, j} \meq{k} \dist_{\oAGw(X, Y)}((0, k-i), (|X|, |Y|-j))$ for all $i, j \in \fragmentcc{0}{k}$.
        \item[Distance Propagation Query:] Given $u \le k$ edits transforming $Y$ into a string $Z \in \Sigma^*$ and a vector $v \in \RR_{\ge 0}^{k+1}$, in $\Oh(((u+1) k + (\sedk(Y) + 1) |Y|) \log |Y|)$ time compute a row that is $k$-equivalent to $v^T \otimes M$, where $M_{i, j} = \dist_{\oAGw(X, Z)}((0, k-i), (|X|, |Z|-j))$ for all $i, j \in \fragmentcc{0}{k}$.
        \item[Path Query:] Given $u \le k$ edits transforming $Y$ into a string $Z \in \Sigma^*$ and integers $i, j \in \fragmentcc{0}{k}$ such that $\dist_{\oAGw(X, Z)}((0, i), (|X|, |Z| - j)) \le k$, in $\Oh(((u+\log |Y|) k + (\sedk(Y) + 1) |Y|) \log |Y|)$ time compute the breakpoint representation of an optimal path from $(0, i)$ to $(|X|, |Z| - j)$ in $\oAGw(X, Z)$.
    \end{description}
\end{lemma}

\begin{proof}
    We first preprocess $X\cdot Y$ in $\Oh(|Y|)$ time to allow constant-time \pillar operations \cite[Section 7.1]{CKW20}.
    In $\Oh(k^2)$ time, we then compute $\min(\sedk(X), k+1)$ and $\min(\sedk(Y), k+1)$ using \cref{lm:compute-sedk}.
    We then compute $\ed_{\le 4k}(X, Y)$ in $\Oh(k^2)$ time by running the \pillar version of the Landau--Vishkin algorithm (\cref{lm:k2-ed}).
    First assume that $\ed(X, Y) > 4k$.
    By the triangle inequality we have $\ed(X, Z) \ge \ed(X, Y) - \ed(Y, Z) > 4k - u \ge 3k$ for every query.
    Therefore,
    \[\dist_{\oAGw(X, Z)}((0, k - i), (|X|, |Z| - j)) \ge \ed(X, Z) - (k - i) - j > k\]
    holds for all $i, j \in \fragmentcc{0}{k}$.
    Similarly, $\dist_{\oAGw(X, Y)}((0, k - i), (|X|, |Y| - j)) > k$ holds for all $i, j \in \fragmentcc{0}{k}$.
    Therefore, in this case, we can answer all queries trivially.
    From now on, we may assume that $\ed(X, Y) \le 4k$.
    In particular, $|Y| \le |X| + 4k$.

    If $|X| < 18k$, we use Klein's algorithm (that is, \cref{thm:klein}) to compute the values $\dist_{\oAGw(X, Y)}((0, k-i), (|X|, |Y|-j))$ for all $i, j \in \fragmentcc{0}{k}$ and initialize the algorithm of \cref{lm:box-with-one-string-changing} in $\Oh(k^2 \log^2 |Y|)$ time.
    For distance matrix query, we already store all the required distances.
    For distance propagation and path queries, we apply \cref{lm:box-with-one-string-changing} in $\Oh((u+1) k \log |Y|)$ time.
    From now on, we assume that $|X| \ge 18k$.

    \medskip

    We now consider the case $\sedk(X) = 0$.
    In this case, we proceed similarly to the \pillar implementation of \cite[Lemma 4.9]{CKW23}.
    We run the algorithm of \cref{lem:decomp2} for $X$, $Y$, and $6k$, arriving at a decomposition $X=\bigodot_{i=0}^{m - 1}X_i$
    and a sequence of fragments $(Y_i)_{i=0}^{m-1}$, such that $X_i=X\fragmentco{x_i}{x_{i+1}}$ and $Y_i=Y\fragmentco{y_i}{y'_{i+1}}$ for $y_i=\max(x_i-6k,0)$ and $y'_{i+1}=\min(x_{i+1}+18k,|Y|)$ for all $i\in \fragmentco{0}{m}$.
    For consistency, define $y_m = \max(x_m - 6k, 0) = |X| - 6k$ and $y'_0 = \min(x_0 + 18k, |Y|) = 18k$.
    The decomposition is represented using a set $F$ of size $\Oh((\sedk(X) + 1) \cdot ((|Y| - |X|) / k + 1)) = \Oh(1)$ such that $X_i=X_{i-1}$ and $Y_i=Y_{i-1}$ holds for each $i\in \fragmentco{0}{m}\setminus F$ and the endpoints of $X_i$ for $i\in F$. 
    To easily handle corner cases, we assume that $\fragmentco{0}{m}\setminus F\subseteq \fragmentcc{2}{m-5}$; if the original set $F$ does not satisfy this condition, we add the missing $\Oh(1)$ elements.
    
    For every $i \in F$, we run the preprocessing of \cref{lm:klein-upgraded,lm:box-with-one-string-changing} for $X_i$ and $Y_i$ in $\Oh(k^2 \log^2 |Y|)$ time in total.

    For each $i\in \fragmentco{0}{m}$, consider a subgraph $G_i$ of $\oAGw(X, Y)$ induced by $\fragmentcc{x_i}{x_{i+1}}\times \fragmentcc{y_i}{y'_{i+1}}$. Denote by $G$ the union of all subgraphs $G_i$.
    For each $i \in \fragmentcc{0}{m}$, denote $V_i \coloneqq \setof{(x_i, y) \in V(G)}{y \in \fragmentcc{y_i}{y'_i}}$.
    Note that $V_0=\setof{(0,i)}{i \in \fragmentcc{0}{18k}}$, $V_m=\setof{(|X|,j)}{j \in \fragmentcc{|X| - 6k}{|Y|}}$, and $V_i = V(G_i) \cap V(G_{i - 1})$ for $i \in \fragmentoo{0}{m}$.
    Moreover, for $i, j \in \fragmentcc{0}{m}$ with $i \le j$, let $D_{i, j}$ denote the matrix of pairwise distances from $V_i$ to $V_j$ in $G$, where rows of $D_{i, j}$ represent the vertices of $V_i$ in the decreasing order of the second coordinate, and columns of $D_{i, j}$ represent the vertices of $V_j$ in the decreasing order of the second coordinate.
    By the monotonicity property of \cref{lm:oAGw-properties}, the shortest paths between vertices of $G_i$ stay within $G_i$. 
    Hence, for each $i \in F$, we can retrieve $D_{i, i + 1}$ in $\Oh(k^2 \log |Y|)$ time using \cref{lm:klein-upgraded}.
    
    As $V_i$ is a separator between the graphs $G_0, \ldots, G_{i-1}$ and the graphs $G_i, \ldots, G_{m-1}$, we have $D_{\ell, r} = \bigotimes_{i \in \fragmentco{\ell}{r}} D_{i, i + 1}$ for all $\ell, r \in \fragmentcc{0}{m}$ with $\ell < r$.
    We build a dynamic range composition data structure of \cref{fct:dynamic-range-composition-ds} over $(D_{i, i+1})_{i = 0}^{m-1}$ with $(\matmon_{\le 24k+1}, \xotimes)$ as the semigroup in $\Oh(k^2 \log^2 |Y|)$ time using binary-exponentiation-like concatenate sequences for segments of equal matrices.
    In particular, we compute $D_{0, m}$, which stores all the values $\dist_{G}((0, k-i), (|X|, |Y|-j))$ for $i, j \in \fragmentcc{0}{k}$.
    We use them to answer distance matrix queries.
    Note that $G$ contains all vertices $(x, y)\in \fragmentcc{0}{|X|}\times \fragmentcc{0}{|Y|}$ with $|x - y| \le 6k$, and thus the returned distances are $k$-equivalent to the corresponding distances in $\oAGw(X, Y)$.
    
    For a distance propagation query, we observe how the fragments $Y_i$ and the corresponding subgraphs $G_i$ would change under the $u$ given edits transforming $Y$ into $Z$.
    Call these $u$ edits the \emph{input edits}.
    Let the new subgraphs be $G'_i$, their union be $G'$, the new sets $V_i$ be $V'_i$, and the new matrices $D_{a, b}$ be $D'_{a, b}$.
    Note that the distances $\dist_{G'}((0, k-i), (|X|, |Z|-j))$ for $i, j \in \fragmentcc{0}{k}$ form a submatrix of $D'_{0, m}$.
    Furthermore, $G'$ contains all vertices $(x, y) \in \fragmentcc{0}{|X|} \times \fragmentcc{0}{|Z|}$ with $|x - y| \le 5k$, and thus these distances in $G'$ are $k$-equivalent to the corresponding distances in $\oAGw(X, Z)$.
    We pad the vector $v$ with $\infty$ to construct a vector $v'$ of size $|V'_0|$.
    The values we want to compute are a part of $v'^T \otimes D'_{0, m}$.

    Each of the input edits affects a constant number of subgraphs $G_i$.
    Let $i_1 \le i_2 \le \cdots \le i_t$ for $t = \Oh(u)$ be the indices of all the affected subgraphs.
    We have
    \[D'_{0, m} = D_{0, i_1} \otimes D'_{i_1, i_1+1} \otimes D_{i_1+1, i_2} \otimes \cdots \otimes D'_{i_t, i_t+1} \otimes D_{i_t+1,m}.\]
    We compute $v'^T \otimes D'_{0, m}$ by iteratively multiplying $v'^T$ by the matrices on the right-hand side of the above equality.
    Each $D$ matrix in this product can be replaced by $\Oh(\log |Y|)$ already computed matrices using a query to the dynamic range composition data structure (\cref{fct:dynamic-range-composition-ds}).
    We use SMAWK algorithm (\cref{thm:smawk}) to multiply a row by such matrices in $\Oh(k)$ time for $\Oh((u + 1) k \log |Y|)$ time in total.
    We use \cref{lm:box-with-one-string-changing} to multiply a row by $D'_{i_j, i_j+1}$.
    Let $u_1, \ldots, u_t$ with $\sum_j u_j = \Oh(u)$ be the number of input edits applied to each fragment $Y_{i_j}$.
    The application of \cref{lm:box-with-one-string-changing} for $D'_{i_j, i_j + 1}$ costs $\Oh((u_j+1) k \log |Y|)$ time for a total of $\Oh((u+1) k \log |Y|)$.

    To answer path queries, we run the distance propagation query with all entries of $v$ equal to the distances in $\oAGw(X, Z)$ from $(0, i)$ to the corresponding input vertices.
    In particular, we compute $\dist_{\oAGw(X, Z)}((0, i), (|X|, |Z| - j))$.
    We then backtrack the computation of this value to retrieve the underlying $w$-optimal path.
    We backtrack the answer down the computation tree similarly to the proof of \cref{lm:klein-upgraded} to instances of \cref{lm:klein-upgraded,lm:box-with-one-string-changing}, both of which allow for path reconstruction.
    The reconstruction of the path takes $\Oh((u + 1) k \log |Y| + k \log^2 |Y|)$ time.
    This concludes our description of the case of $\sedk(X) = 0$.

    \medskip

    We now consider the case of $\sedk(X) \ge 1$.
    In this case, we return $\absorbingmatrix$ for matrix distance queries.
    If $\sedk(Y) > k$, we do not do any further preprocessing.
    Upon a distance propagation query, let $G$ be the restriction of $\oAGw(X, Z)$ onto the vertices $(i, j)$ with $|i - j| \le 2k$.
    Note that the distances between pairs of vertices $(0, k-i)$ and $(|X|, |Z|-j)$ for all $i, j \in \fragmentcc{0}{k}$ in $G$ are $k$-equivalent to the corresponding distances in $\oAGw(X, Z)$ due to \cref{lm:paths_dont_deviate_too_much}.
    We use Klein's algorithm (\cref{thm:klein}) to compute the matrix $M'$ of such distances in $G$ in $\Oh(|Y|k \log |Y|)$ time.
    We then min-plus multiply $v^T$ and $M'$ in time $\Oh(k^2) \le \Oh(|Y|k)$.
    The total query time is $\Oh(|Y|k \log |Y|) = \Oh(|Y| \sedk(Y) \log |Y|)$.
    To answer path queries, we use trivial dynamic programming in $\Oh(|Y| k)$ time.

    It remains to consider the case when $\sedk(X) \ge 1$ and $\sedk(Y) \le k$.
    As $X$ is a substring of~$Y$, it follows that $\sedk(X) \le \sedk(Y)$ due to \cref{lm:superadditivity-of-sedk}.
    In this case, we proceed similarly to the algorithm of \cite[Lemma 4.9]{CKW23} in the standard setting.
    We set $\ell \coloneqq \floor{k / \sedk(Y)}$ and run the algorithm of \cref{lem:decomp3} with parameters $4k$ and $\ell$ for both $X$ and $Y$, arriving at decompositions
    $X=\bigodot_{i=0}^{m_X-1} X_i$ and $Y=\bigodot_{j=0}^{m_Y-1} Y_j$ with $\Oh(k / \ell + \sedk(X))$ and $\Oh(k / \ell + \sedk(Y))$ fresh phrases correspondingly, both of which are $\Oh(k / \ell)$ due to the definition of $\ell$.
    Next, we define a box decomposition of the alignment graph $\oAGw(X, Y)$ as the indexed family of boxes $(B_{i, j})_{(i, j) \in \fragmentco{0}{m_X} \times \fragmentco{0}{m_Y}}$, where box $B_{i, j}$ is the subgraph of $\oAGw(X, Y)$ induced by $\fragmentcc{x_i}{x_{i+1}} \times \fragmentcc{y_i}{y_{i+1}}$.
    We say that a box $B_{i, j}$ is \emph{relevant} if it contains at least one vertex $(x, y)$ with $|x - y| \le 4k$.
    Call box $B_{i, j}$ fresh if $X\fragmentco{x_i}{x_{i+1}}$ or $Y\fragmentco{y_i}{y_{i+1}}$ is a fresh phrase.
    We say that two boxes are isomorphic if and only if their corresponding fragments in $X$ and $Y$ match.

    \begin{claim}
        Every relevant box is isomorphic to some fresh relevant box.
    \end{claim}

    \begin{claimproof}
        Consider some relevant box $B_{i, j}$.
        We prove the claim by induction on $i$ and $j$.
        Without loss of generality assume that $x_i \ge y_j$.
        If $B_{i, j}$ is fresh, it is isomorphic to itself.
        Otherwise, due to \cref{lem:decomp3}, there is an index $i' \in \fragmentco{0}{i}$ with $X\fragmentco{x_i}{x_{i+1}}=X\fragmentco{x_{i'}}{x_{i'+1}}$ and $x_i - x_{i'} \le 4k$.
        Hence, $B_{i', j}$ is isomorphic to $B_{i, j}$.
        As $B_{i, j}$ is relevant, it contains some vertex $(x, y)$ with $|x - y| \le 4k$.
        Because $x_i \ge y_j$, we can further assume that $x \ge y$.
        Thus, vertex $(x - (x_i - x_{i'}), y)$ belongs to $B_{i', j}$ and satisfies $|(x - (x_i - x_{i'})) - y| \le 4k$.
        Therefore, $B_{i', j}$ is a relevant box.
        By the induction hypothesis, it is isomorphic to some fresh relevant box and so is $B_{i, j}$.
    \end{claimproof}
    
    Note that, for each $i \in \fragmentco{0}{m_X}$, there are $\Oh(k / \ell)$ relevant boxes of the form $B_{i, \cdot}$, and for each $j \in \fragmentco{0}{m_Y}$ there are $\Oh(k / \ell)$ relevant boxes of the form $B_{\cdot, j}$.
    In particular, there are $\Oh(|Y| k / \ell^2)$ relevant boxes in total and $\Oh(k^2 / \ell^2)$ fresh relevant boxes.
    For every fresh relevant box $B_{i, j}$, we run the preprocessing of \cref{lm:box-with-one-string-changing} in $\Oh(\ell^2 \log^2 |Y|)$ time and save the results into a dictionary (implemented as a self-balancing binary tree) with $(X_i, Y_j)$ as the key.
    It takes $\Oh(k^2 / \ell^2 \cdot \ell^2 \log^2 |Y|) = \Oh(k^2 \log^2 |Y|)$ time.
    This concludes the preprocessing phase.
    
    For a distance propagation query, we analyze how the fragments $Y_j$ and the corresponding boxes $B_{i, j}$ would change under the $u$ input edits transforming $Y$ into $Z$.
    Let the new boxes be $B'_{i, j}$ and their union be $G'$.
    Consider the matrix $M'$ consisting of distances $\dist_{G'}((0, k-i), (|X|, |Z|-j))$ for $i, j \in \fragmentcc{0}{k}$.
    Note that $G'$ contains all vertices $(x, y) \in \fragmentcc{0}{|X|} \times \fragmentcc{0}{|Z|}$ with $|x - y| \le 3k$, and thus these distances in $G'$ are $k$-equivalent to the corresponding distances in $\oAGw(X, Z)$.
    We want to compute $v^T \otimes M'$.
    Add an auxiliary vertex $r$ to $G'$ and add edges from it to all vertices of the form $(0, k - i)$ for $i \in \fragmentcc{0}{k}$ with weights corresponding to the values of $v$.
    Let the new graph be $\oG'$.
    The values of $v^T \otimes M'$ are equal to the distances from $r$ to the vertices $(|X|, |Z| - j)$ for $j \in \fragmentcc{0}{k}$ in $\oG'$.
    To determine these distances, we process relevant boxes $B'_{i, j}$ in the lexicographic order.
    For each such box, we already have the distances from $r$ to every input vertex of $B'_{i, j}$ (either computed trivially or from boxes $B'_{i - 1, j}$ and $B'_{i, j - 1}$).
    We use the dictionary to find the data structure of \cref{lm:box-with-one-string-changing} corresponding to $B'_{i, j}$ and query it to compute the distances from $r$ to the output vertices of $B'_{i, j}$.
    It takes $\Oh((\ed(Y_i, Z_i) + 1) \cdot \ell \log |Y|)$ time for a total of $\Oh((u \cdot (k / \ell) + |Y| k / \ell^2) \cdot \ell \log |Y|) = \Oh(u k \log |Y| + |Y| \cdot \sedk(Y) \log |Y|)$.

    To answer path queries, we backtrack the computation of distance propagation queries similarly to the case $\sedk(X) = 0$.
\end{proof}

We are now ready to describe the small self-edit distance dynamic algorithm that maintains instances of \cref{lm:block-data-structure} over $\Theta(k^{2 - \gamma}) \times \Theta(k^{2 - \gamma})$-sized subgraphs of $\oAGw(X, X)$.

\begin{lemma} \label{lm:dynamic-algo-for-the-conquer-step}
    Let $\gamma \in (0, 1]$ be a real value, $k \ge 1$ be an integer, $X, Z \in \Sigma^+$ be strings, and $w \colon \Esigma^2 \to [0, W]$ be a normalized weight function.
    There is a fully persistent data structure $\localds(\gamma, w, k, X)$ that supports the following operations:
    \begin{description}
        \item[Build:] Given $\gamma$, $k$, $X$, and $\Oh(1)$-time oracle access to $w$, build $\localds(\gamma, w, k, X)$ in $\Oh(|X| k^{\gamma}\cdot \allowbreak \log^2 (|X| + 1))$ time.
        \item[Split:] Given $\localds(\gamma, w, k, X)$ and $p \in \fragmentco{1}{|X|}$, build $\localds(\gamma, w, k, X\fragmentco{0}{p})$ and $\localds(\gamma, w, k, \linebreak X\fragmentco{p}{|X|})$ in $\Oh(k^2 \cdot \log^2 (|X| + 1))$ time.
        \item[Concatenate:] Given $\localds(\gamma, w, k, X)$ and $\localds(\gamma, w, k, Z)$, build $\localds(\gamma, w, k, X \cdot Z)$ in $\Oh(k^2\cdot \linebreak \log^2 (|X| + |Z|))$ time.
        \item[Query:] Given $\localds(\gamma, w, k, X)$ such that $\sed(X) \le k$ and a sequence of $u \le k$ edits transforming $X$ into a string $Z$, compute $\wed_{\le k}(X, Z)$ in $\Oh(k^{3 - \gamma} \log (|X| + 1) + k^2 \log^2 (|X| + 1))$ time.
            Furthermore, if $\wed(X, Z) \le k$, the query returns the breakpoint representation of a $w$-optimal alignment of $X$ onto $Z$.
    \end{description}
\end{lemma}

\begin{proof}
    Let $B \coloneqq 10 \floor{k^{2 - \gamma}}$.
    We first describe the data structure $\localds(\gamma, w, k, X)$ and how to build it.
    If $|X| < B$, we do not store anything, so there is nothing to be built.
    Otherwise, let a sequence $X_0, X_1, \ldots, X_{m-1}$, where $X_i = X\fragmentco{x_i}{x_{i+1}}$ be called a \emph{$B$-partitioning} of $X$ if $x_0 = 0$, $x_m = |X|$, and $x_{i+1}-x_i \in \fragmentco{B}{2B}$ for all $i \in \fragmentco{0}{m}$.
    Fix a $B$-partitioning $X_0, X_1, \ldots, X_{m-1}$ of $X$.
    Furthermore, define a sequence of fragments $(Y_i)_{i=0}^{m-1}$ such that $Y_i = X\fragmentco{y_i}{y'_{i+1}}$ for $y_i = \max(x_i - 2k, 0)$ and $y'_{i+1} = \min(x_{i+1} + 2k, \abs{X})$ for all $i \in \fragmentco{0}{m}$.
    Let $G_i$ be the subgraph of $\oAGw(X, X)$ induced by $\fragmentcc{x_i}{x_{i+1}} \times \fragmentcc{y_i}{y'_{i+1}}$.
    Denote by $G$ the union of all subgraphs $G_i$.
    For each $i \in \fragmentcc{0}{m}$, denote $V_i \coloneqq \setof{(x_i, y) \in V(G)}{y \in \fragmentcc{y_i}{y'_i}}$, where $y'_0 = 0$ and $y_m = |X|$.

    Our data structure (implicitly) maintains a $B$-partitioning $X_0, X_1, \ldots, X_{m-1}$ of $X$.
    Furthermore, it stores a dictionary implemented as a functional balanced binary tree that, for all $i \in \fragmentco{0}{m}$, stores the data structures of \cref{lm:block-data-structure} for $X_i$, $Y_i$, and parameter $5k$.
    Moreover, we store a dynamic range composition data structure of \cref{fct:dynamic-range-composition-ds} with $(\matmon_{\le 5k+1}, \xotimes)$ as the semigroup over the elements $M_i$ of $\matmon_{\le 5k+1}$ returned by the distance matrix queries to \cref{lm:block-data-structure}.
    Finally, we store a functional balanced binary tree over $(x_i)_{i=0}^{m}$.

    At the initialization phase, we build an arbitrary $B$-partitioning of $X$.
    For every $i \in \fragmentco{0}{m}$, we initialize the algorithm of \cref{lm:block-data-structure} for $X_i$ and $Y_i$ and query elements $M_i \in \matmon_{\le 5k+1}$.
    It takes $\Oh(m \cdot (B + k^2 \log^2 B)) = \Oh(|X|k^{\gamma} \log^2 |X|)$ time.
    We then build the dynamic range composition data structure of \cref{fct:dynamic-range-composition-ds} over $(M_i)_{i=0}^{m-1}$ in time $\Oh(m k^2) = \Oh(|X| k^{\gamma} \log |X|)$ and a functional balanced binary tree over $(x_i)_{i=0}^m$ in time $\Oh(|X|)$.
    This concludes the initialization phase.

    \medskip

    For a split operation, we show how to compute $\localds(\gamma, w, k, X\fragmentco{0}{p})$; a symmetric approach can be used to compute $\localds(\gamma, w, k, X\fragmentco{p}{|X|})$.
    If $p < B$, there is nothing to be built.
    Otherwise, in $\Oh(\log |X|)$ time, we identify the index $i \in \fragmentco{0}{m}$ such that $p \in \fragmentco{x_i}{x_{i+1}}$.
    The $B$-partitioning of $X\fragmentco{0}{p}$ will consist of phrases $X_0, \ldots, X_{i-2}$ and either one or two new phrases that partition $X\fragmentco{x_{i-1}}{p}$.
    In $\Oh(\log |X|)$ time, we can get a dictionary storing the data structures of \cref{lm:block-data-structure} for $j \in \fragmentco{0}{i-1}$ and a functional balanced binary tree over $(x_j)_{j=0}^{i - 1}$.
    Analogously, in $\Oh(k^2 \log |X|)$ time we can get the corresponding prefix of the dynamic range composition data structure we store.
    It remains to build the data structure of \cref{lm:block-data-structure} for the last one or two phrases in the $B$-partitioning of $X\fragmentco{0}{p}$ and adjust the data structures we store accordingly.
    It takes $\Oh(k^2 \log^2 |X|)$ time in total.
    The only major difference with constructing $\localds(\gamma, w, k, X\fragmentco{p}{|X|})$ is that the indices $x_i$ will be shifted, and we need to subtract $p$ from all of them in the corresponding balanced binary tree.

    \medskip

    We now describe how to handle a concatenate operation.
    If $|X| < B$ and $|Z| < B$, we build the data structure $\localds(\gamma, w, k, X \cdot Z)$ from scratch in $\Oh(k^2 \log^2 (|X| + |Z|))$ time.
    If $|X| \ge B$ and $|Z| < B$, we remove the last phrase in the $B$-partitioning of $X$ and replace it with either one or two phrases partitioning the end of $X$ and $Z$.
    We update the data structure for $X$ similarly to how we implemented the split operation.
    It takes $\Oh(k^2 \log^2 (|X| + |Z|))$ time.
    If $|X| < B$ and $|Z| \ge B$, we proceed analogously.
    It remains to consider the case $|X|, |Z| \ge B$.
    In this case, we just concatenate the $B$-partitionings of $X$ and $Z$ and join the data structures we store (while shifting the indices of the partitioning of $Z$) in time $\Oh(k^2 \log (|X| + |Z|))$.

    \medskip

    It remains to handle a query.
    If $|X| < B$, we compute the corresponding $w$-optimal path from $(0, 0)$ to $(|X|, |Z|)$ in the subgraph of $\oAGw(X, Z)$ induces by the vertices $(x, y)$ with $|x - y| \le k$ in time $\Oh(k \cdot |X|) \le \Oh(k^{3 - \gamma})$.
    The computed distance is $k$-equivalent to $\wed(X, Z)$ due to \cref{lm:paths_dont_deviate_too_much}.
    Otherwise, if $|X| \ge B$, we analyze how the fragments $Y_i$ and the corresponding subgraphs $G_i$ would change under the input edits transforming the second string in $\oAGw(X, X)$ into $Z$.
    Let the new subgraphs be $G'_i$, their union be $G'$, and the changed sets $V_i$ be $V'_i$.
    We compute the required distance in $G'$.
    Note that $G'$ contains all vertices $(x, y) \in \fragmentcc{0}{|X|} \times \fragmentcc{0}{|Z|}$ with $|x - y| \le k$, and thus the computed distance is $k$-equivalent to $\wed(X, Z)$ due to \cref{lm:paths_dont_deviate_too_much}.
    
    For $i, j \in \fragmentcc{0}{m}$ with $i \le j$, let $D'_{i, j}$ denote the matrix of pairwise distances from $V'_i$ to $V'_j$ in $G'$, where rows of $D'_{i, j}$ represent the vertices of $V'_i$ in the decreasing order of the second coordinate, and columns of $D'_{i, j}$ represent the vertices of $V'_j$ in the decreasing order of the second coordinate.
    As $V'_i$ is the separator between the graphs $G'_0, \ldots, G'_{i-1}$ and the graphs $G'_i, \ldots, G'_{m-1}$, we have $D'_{\ell, r} = \bigotimes_{i \in \fragmentco{\ell}{r}} D'_{i, i + 1}$ for all $\ell, r \in \fragmentcc{0}{m}$ with $\ell < r$.
    The value we want to compute is the single entry of $D'_{0, m}$.
    
    Let $q$ be the number of indices $i \in \fragmentco{0}{m}$ with $\sedk(X_i) \ge 1$.
    As $\sed(X) \le k$ is required for the query, we have $q \le k$ by \cref{lm:superadditivity-of-sedk}.
    Furthermore, we can find all such indices $i$ in time $\Oh((q + 1) \log |X|)$ as the indices that have $M_i = \absorbingmatrix$ stored in the dynamic range composition data structure.

    Each input edit affects a constant number of subgraphs $G_i$.
    Let $i_1 \le i_2 \le \cdots \le i_t$ for $t = \Oh(q + u) \le \Oh(k)$ be the indices of all the affected subgraphs and all subgraphs that have $\sedk(X_i) \ge 1$.

    We have
    \[D'_{0, m} = D'_{0, i_1} \otimes D'_{i_1, i_1+1} \otimes D'_{i_1+1, i_2} \otimes \cdots \otimes D'_{i_t, i_t+1} \otimes D'_{i_t+1,m}.\]
    We compute this product from left to right maintaining the current row $D'_{0, j}$.
    By \cref{lm:oAGw-properties} the shortest paths between vertices of $G'_j$ stay within $G'_j$.
    Therefore, for all $j \in \fragmentoo{i_{\ell} + 1}{i_{\ell+1}}$, we have $D'_{j, j + 1} = M_j$, and thus we can use the query operation of the dynamic range composition data structure to get a list $A_1, \ldots, A_a$ of matrices with $a = \Oh(\log |X|)$ and
    \[A_1 \otimes \cdots \otimes A_a = D'_{i_{\ell} + 1, i_{\ell} + 2} \otimes \cdots \otimes D'_{i_{\ell + 1} - 1, i_{\ell + 1}} = D'_{i_{\ell} + 1, i_{\ell + 1}}.\]
    We then use SMAWK algorithm (\cref{thm:smawk}) to multiply a row by such matrices in $\Oh(k)$ time each for a total of $\Oh(k^2 \log |X|)$.
    To multiply the row by $D'_{i_{\ell}, i_{\ell + 1}}$, we use \cref{lm:block-data-structure}.
    Let $u_1, \ldots, u_t$ with $\sum_{\ell} u_{\ell} = \Oh(u)$ be the number of input edits applied to the fragments $Y_{i_1}, \ldots, Y_{i_t}$.
    The application of \cref{lm:block-data-structure} for $D'_{i_{\ell}, i_{\ell} + 1}$ costs $\Oh(((u_{\ell}+1) k + (\sedk(Y_{\ell}) + 1) B) \log |X|)$ time for a total of $\Oh(k^{3 - \gamma} \log |X|)$.
    Here, note that even-indexed fragments $Y_{2j}$ are disjoint and odd-indexed fragments $Y_{2j+1}$ are disjoint, so we have $\sum_{j} \sedk(Y_{2j}) \le \sed(X) \le k$ and $\sum_j \sedk(Y_{2j+1}) \le \sed(X) \le k$ due to \cref{lm:superadditivity-of-sedk}.

    We computed $\wed_{\le k}(X, Z)$.
    If $\wed(X, Z) \le k$, we now compute the $w$-optimal alignment of $X$ onto~$Z$.
    For that, we backtrack the computation of $D'_{0, m}$.
    We backtrack the answer down the computation tree (while pruning branches with distance zero) to instances of \cref{lm:block-data-structure}, which allow for path reconstruction.
    It analogously takes $\Oh(k^{3 - \gamma} \log (|X| + 1) + k^2 \log^2 (|X| + 1))$ time in total.
\end{proof}

To improve the query time of \cref{lm:dynamic-algo-for-the-conquer-step} when $\wed(X, Z) \ll k$, we maintain a logarithmic number of instances of \cref{lm:dynamic-algo-for-the-conquer-step} for exponentially growing thresholds.

\begin{corollary} \label{cor:dynamic-algo-for-the-conquer-step}
    Let $\gamma \in (0, 1]$ be a real value, $k \ge 1$ be an integer, $X, Z \in \Sigma^+$ be strings, and $w \colon \Esigma^2 \to [0, W]$ be a normalized weight function.
    There is a fully persistent data structure $\localdscor(\gamma, w, k, X)$ that supports the following operations:
    \begin{description}
        \item[Build:] Given $\gamma$, $k$, $X$, and $\Oh(1)$-time oracle access to $w$, build $\localdscor(\gamma, w, k, X)$ in $\Oh(|X| k^{\gamma} \cdot \allowbreak \log^2 (|X| + 1))$ time.
        \item[Split:] Given $\localdscor(\gamma, w, k, X)$ and $p \in \fragmentco{1}{|X|}$, build $\localdscor(\gamma, w, k, X\fragmentco{0}{p})$ and $\localdscor(\gamma, w, k, \linebreak X\fragmentco{p}{|X|})$ in $\Oh(k^2\cdot \log^2 (|X| + 1))$ time.
        \item[Concatenate:] Given $\localdscor(\gamma, w, k, X)$ and $\localdscor(\gamma, w, k, Z)$, build $\localdscor(\gamma, w, k, X \cdot Z)$ in $\Oh(k^2 \cdot \linebreak \log^2 (|X| + |Z|))$ time.
        \item[Query:] Given $\localdscor(\gamma, w, k, X)$, an integer $d \in \fragmentcc{1}{k}$ such that $\sed(X) \le d$, and a sequence of $u \le d$ edits transforming $X$ into a string $Z$, compute $\wed_{\le d}(X, Z)$ in $\Oh(d^{3 - \gamma} \log (|X| + 1) + d^2 \log^2 (|X| + 1))$ time.
            Furthermore, if $\wed(X, Z) \le d$, the query returns the breakpoint representation of a $w$-optimal alignment of $X$ onto $Z$.
    \end{description}
\end{corollary}

\begin{proof}
    We maintain the data structures of \cref{lm:dynamic-algo-for-the-conquer-step} for thresholds $1, 2, 4, \ldots, 2^{\ceil{\log k}}$.
    This does not affect the build, split, and concatenate time complexities.
    To answer a query, we query the $\ceil{\log d}$-th data structure.
\end{proof}

\subsection{The General Case}

It remains to reduce the general case of dynamic bounded weighted edit distance to instances of small self-edit distance that we can handle using \cref{cor:dynamic-algo-for-the-conquer-step}.
For that, we follow a generalized version of the framework introduced in \cite{GK24}.
We first show that $\localdscor(\gamma, w, k, X)$ lets us compute $\wed(X, Z)$ in time $\tOh(k^{3 - \gamma})$ given as a hint some approximate alignment of $w$-cost at most $k$.

\begin{lemma} \label{lm:improve-alignment}
    Let $\gamma \in (0, 1]$ be a real value, $k \ge 1$ and $k' \ge 10k+2$ be integers, $X \in \Sigma^+$ be a string, and $w \colon \Esigma^2 \to [0, W]$ be a normalized weight function.
    There is an algorithm that, given $\localdscor(\gamma, w, k', X)$, \pillar access to $X$, and an alignment $\cA$ of $X$ onto some string $Z \in \Sigma^*$ of $w$-cost at most $k$, computes $\wed(X, Z)$ and the $w$-optimal alignment of $X$ onto $Z$ in $\Oh(k^{3 - \gamma} \log (|X| + 1) + k^2 \log^2 (|X| + 1))$ time and $\Oh(k^2)$ \pillar operations.
\end{lemma}

\begin{proof}
    We apply the algorithm of \cite[Lemma 7.5]{GK24} by replacing the calls to \cite[Lemma 7.3]{GK24} with the calls to $\localdscor(\gamma, w, k', X)$.
    \cite[Lemma 7.5]{GK24} takes an alignment $\cA$ and uses it to compute the $w$-optimal alignment in a divide-and-conquer manner.
    For a conquer step, \cite[Lemma 7.5]{GK24} uses \cite[Lemma 7.3]{GK24} to compute $w$-optimal alignments between $X'$ and its image $Z'$ under $\cA$ for fragments $X'$ of $X$ with $\sed(X') \le 10k+2$ and $\wed_{\cA}(X', Z') \le k$.
    For that, we can use the split and query operations of $\localdscor(\gamma, w, k', X)$.
    The number of \pillar operations is $\Oh(k^2)$ as it is not different from \cite[Lemma 7.5]{GK24}.
    The time complexity $T(|X|, k)$ follows the following recurrence relation:
    \begin{align*}
        T(n, 0) &= 1,&\\
        T(n, d) &\le 2 T(n, \floor{d / 2}) + \Oh(d^{3 - \gamma} \log (n + 1) + d^2 \log^2 (n + 1))&\text{ for $d \ge 1$.}\\
    \end{align*}
    Therefore, the time complexity of the algorithm is $\Oh(k^{3 - \gamma} \log (|X| + 1) + k^2 \log^2 (|X| + 1))$.
\end{proof}

Before proceeding to our main trade-off dynamic algorithm, we formalize the update types that the string $X$ supports and show how to efficiently maintain \pillar over $X$ undergoing such updates.

\begin{definition}[Block Edit]\label{def:block-edit}
    Let $X \in \Sigma^*$ be a string.
    A \emph{block edit} is one of the following transformations applied to $X$:
    \begin{description}
        \item[Character Insertion:] Insert an arbitrary character $c \in \Sigma$ at some position $i \in \fragmentcc{0}{|X|}$ in $X$.
        \item[Character Deletion:] Delete some character $X[i]$ for $i \in \fragmentco{0}{|X|}$ from $X$.
        \item[Character Substitution:] Replace $X[i]$ in $X$ for some $i \in \fragmentco{0}{|X|}$ with some other character $c \in \Sigma$.
        \item[Substring Removal:] Remove fragment $X\fragmentco{\ell}{r}$ from $X$ for some $\ell, r \in \fragmentcc{0}{|X|}$ with $\ell < r$.
        \item[Copy-Paste:] For some $\ell, r, p \in \fragmentcc{0}{|X|}$ with $\ell < r$, insert $X\fragmentco{\ell}{r}$ at position $p$ of $X$.
    \end{description}
\end{definition}

\begin{lemma}\label{lm:dynamic-pillar-under-block-edits}
    Let $X \in \Sigma^{*}$ be a string and $n \ge 2$ be an integer.
    There is an algorithm that dynamically maintains \modelname{} over $X$ while $X$ undergoes block edits assuming that $|X| \le n$ holds throughout the lifetime of the algorithm.
    The algorithm takes $\Oh((|X| + 1) \log^{\Oh(1)}\log n)$ time for the initialization and $\Oh(\log n \log^{\Oh(1)}\log n)$ time for every block edit and \modelname{} operation.
\end{lemma}

\begin{proof}
    We use a deterministic implementation of the \modelname{} over $X$ in the dynamic setting~\cite[Section 8]{KK22},~\cite[Section 7.2]{CKW22} in epochs.
    It takes $\Oh((|X| + 1) \log^{\Oh(1)}\log n)$ time for the initialization and $\Oh(\log (n + m) \log^{\Oh(1)}\log (n + m))$ time for every block edit and \modelname{} operation where $m$ is the number of block edits applied to $X$.
    For an amortized solution, we may rebuild the data structure every $n$ block edits.
    To deamortize this solution, we use a standard idea of maintaining two versions of the data structure and rebuilding one while using the other, and then switching between the two in epochs.
\end{proof}

Our final algorithm starts with an optimal unweighted alignment of $X$ onto $Z$ and iteratively improves it using \cref{lm:improve-alignment} until it arrives at a $w$-optimal alignment.

\begin{lemma}\label{lm:full-complete-algorithm}
    There is a dynamic data structure that, given a real value $\gamma \in (0, 1]$, an integer $n \ge 2$, a string $X \in \Sigma^*$ that throughout the lifetime of the algorithm satisfies $|X| \le n$, an integer $k \ge 1$, and $\Oh(1)$-time oracle access to a normalized weight function $w \colon \Esigma^2 \to \RR_{\ge 0}$, can be initialized in $\Oh((|X| + 1) k^{\gamma} \log^2 n)$ time and allows for the following operations:
    \begin{itemize}
        \item Apply a block edit to $X$ in $\Oh(k^2 \log^2 n)$ time.
        \item Given $u \le k$ edits transforming $X$ into a string $Z \in \Sigma^*$, compute $\wed_{\le k}(X, Z)$ in $\Oh(k^{3 - \gamma} \log^2 n + k^2 \log^3 n)$ time.
            Furthermore, if $\wed(X, Z) \le k$, the query returns the breakpoint representation of a $w$-optimal alignment of $X$ onto $Z$.
    \end{itemize}
\end{lemma}

\begin{proof}
    As we are focused on computing $\wed_{\le k}$, we may cap the weight function by $k+1$ from above and without loss of generality assume that we have $w \colon \Esigma^2 \to [0, k + 1]$.
    We maintain \cref{lm:dynamic-pillar-under-block-edits} over $X \cdot X$ with threshold $2n$.
    It takes $\Oh((|X| + 1) \log^{\Oh(1)} \log n)$ time to initialize and $\Oh(\log n \log^{\Oh(1)}\log n)$ time per block edit.
    Furthermore, we maintain $\localdscor(\gamma, w, 20k+4, X)$.
    According to \cref{cor:dynamic-algo-for-the-conquer-step}, we build it in time $\Oh((|X| + 1) k^{\gamma} \log^2 n)$.
    
    Each block edit can be executed via a constant number of merge and split operations from \cref{cor:dynamic-algo-for-the-conquer-step}.
    Hence, an update takes $\Oh(k^2 \log^2 n)$ time.

    It remains to explain how to handle a query.
    We start by transforming $X \cdot X$ into $X \cdot Z$ inside \cref{lm:dynamic-pillar-under-block-edits} via $u$ input edits in $\Oh(k \log n \log^{\Oh(1)} \log n)$ time in total.
    This allows us to have \pillar access to $X$ and $Z$.
    At the end of the query, we transform the string back to $X \cdot X$.
    Our goal is to find the breakpoint representation of a $w$-optimal alignment $\cA \colon X\onto Z$ or decide that $\wed(X, Z) > k$.
    From the alignment, in time $\Oh(k)$, we can compute $\wed(X, Z)$.
    If $|X| \le k$, we use \cref{lm:baseline-wed} to answer a query in $\Oh((|X| + 1) \cdot k) \le \Oh(k^2)$ time.
    Otherwise, we proceed similarly to the proof of \cite[Theorem 7.6]{GK24}.
    Let $w_0(a, b) = w(a, b)$ and $w_t(a, b) = \ceil{w(a, b) / 2^t}$ for all $a, b \in \Esigma$ and $t \in \fragmentcc{1}{\ceil{\log (k+1)}}$.
    We now iterate from $t = \ceil{\log (k+1)}$ down to $t = 0$ and, for each such $t$, compute a $w_t$-optimal alignment $\cA_t \colon X \onto Z$.
    At the end, we get a $w$-optimal alignment $\cA = \cA_0$.

    For $t = \ceil{\log (k+1)}$, we have $w_t(a, b) = 1$ for all $a, b \in \Esigma$ with $a \neq b$.
    We find an optimal unweighted alignment of $X$ onto $Z$ using \cref{lm:k2-ed} in $\Oh(k^2)$ \pillar operations or decide that $\ed(X, Z) > k$.
    In the later case we have $\wed(X, Z) \ge \ed(X, Z) > k$, and thus we can terminate.

    For every smaller $t$, given a $w_{t + 1}$-optimal alignment $\cA_{t + 1} \colon X \onto Z$, we apply \cref{lm:improve-alignment} to find a $w_t$-optimal alignment $\cA_t \colon X \onto Z$.
    It takes $\Oh(\ell^{3 - \gamma} \log n)$ time and $\Oh(\ell^2)$ \pillar operations, where $\ell \coloneqq \ed_{\cA_{t + 1}}^{w_t}(X, Y)$.
    Note that $w_t(a, b) \le 2 \cdot w_{t+1}(a, b)$ for all $a, b \in \Esigma$, so $\ed_{\cA_{t+1}}^{w_t}(X, Z) \le 2 \ed_{\cA_{t + 1}}^{w_{t + 1}}(X, Y) \le 2k$.
    Hence, a call to \cref{lm:improve-alignment} takes $\Oh(k^{3 - \gamma} \log n + k^2 \log^2 n)$ time and $\Oh(k^2)$ \pillar operations.
    If we have $\wed(X, Z) \ge \ed^{w_t}(X, Z) > k$ at some point, then we terminate.

    In total, the algorithm takes $\Oh(k^{3 - \gamma} \log^2 n + k^2 \log^3 n)$ time and $\Oh(k^2 \log n)$ \pillar operations.
    Each \pillar operation takes $\Oh(\log n \log^{\Oh(1)} \log n)$ time, and thus the total time complexity is $\Oh(k^{3 - \gamma} \log^2 n + k^2 \log^3 n)$.
\end{proof}

Finally, we once again maintain a logarithmic number of instances of \cref{lm:full-complete-algorithm} for exponentially growing thresholds to improve the query time when $\wed(X, Z) \ll k$.

\begin{theorem}\label{thm:full-complete-algorithm}
    There is a dynamic data structure that, given a real value $\gamma \in (0, 1]$, an integer $n \ge 2$, a string $X \in \Sigma^*$ that throughout the lifetime of the algorithm satisfies $|X| \le n$, an integer $k \ge 1$, and $\Oh(1)$-time oracle access to a normalized weight function $w \colon \Esigma^2 \to \RR_{\ge 0}$, can be initialized in $\Oh((|X| + 1) k^{\gamma} \log^2 n)$ time and allows for the following operations:
    \begin{itemize}
        \item Apply a block edit to $X$ in $\Oh(k^2 \log^2 n)$ time.
        \item Given $u$ edits transforming $X$ into a string $Z$, compute $\wed_{\le k}(X, Z)$ in time $\Oh(1 + u \log n \log^{\Oh(1)} \log n + d^{3 - \gamma} \log^2 n + d^2 \log^3 n)$, where $d = \min(k, \wed(X, Z))$.
            Furthermore, if $\wed(X, Z) \le k$, the algorithm returns the breakpoint representation of a $w$-optimal alignment of $X$ onto $Z$.
    \end{itemize}
\end{theorem}

\begin{proof}
    We maintain the data structures of \cref{lm:full-complete-algorithm} for thresholds $1, 2, \ldots, 2^{\ceil{\log k}}$.
    This does not affect the asymptotic construction and update times.
    Upon a query, we use the trick from \cref{lm:full-complete-algorithm} to get \pillar access to both $X$ and $Z$ in $\Oh(u \log n \log^{\Oh(1)} \log n)$ time.
    After that, we apply \cref{lm:k2-ed} to find $\ed_{\le k}(X, Z)$ in time $\Oh(d^2 \log n \log^{\Oh(1)} \log n)$.
    If $\ed(X, Z) > k$, then we terminate.
    Otherwise, we also get an optimal unweighted alignment $\cA$ of $X$ onto $Z$ of cost at most $d$.
    We then query the data structures of \cref{lm:full-complete-algorithm} in the increasing order of thresholds starting at $\ceil{\log \ed(X, Z)}$ and stop as soon as we get a positive result or hit $2^{\ceil{\log k}}$.
    We use $\cA$ as the source of edits transforming $X$ into $Z$.
\end{proof}

This concludes our trade-off dynamic bounded weighted edit distance algorithm.
Note that the need for a universal upper bound $n$ on the length of $X$ stems from the complexity of updates that can be applied to $X$.
In particular, after a single update, the length $|X|$ of $X$ can become any value in $\fragmentcc{0}{2|X|}$.
We are not aware of a counterpart of \cref{lm:dynamic-pillar-under-block-edits} that would maintain \pillar over $X$ undergoing block edits with $\Oh(\polylog (|X| + 2))$ time per operation.
In contrast, if we restrict the possible updates to just simple edits, we can use \cite[Lemma 8.10]{GK24} to get rid of the universal upper bound $n$ and replace all occurrences of $n$ in the time complexities with $|X| + 2$.

\section{Small Edit Distance Static Lower Bound}\label{app:smallEDLB}
This section provides a complete proof of \cref{thm:static-lb}.
For technical reasons, we prove a slightly stronger statement formulated as a reduction from the tripartite negative triangle problem.
\begin{problem}[Tripartite Negative Triangle]
    Given a tripartite complete graph with vertex set $P\cup Q \cup R$ of size $N$ and a weight function $w_G: (P\times Q)\cup (Q\times R)\cup (R\times P) \to \fragmentcc{-N^{\Oh(1)}}{N^{\Oh(1)}}$, decide if there exists a triangle $(p,q,r)\in P\times Q\times R$ such that $w_G(p,q)+w_G(q,r)+w_G(r,p) < 0$.
\end{problem}

\renewcommand{\maybe}{\lipicsEnd}
\begin{restatable}{theorem}{thmstaticlbreduction}\label{lm:smallEDLBReduction}
    Let $\beta \in [0,1]$ be a real parameter. 
    There is an algorithm that, given $n\in \ZZ_+$ and an instance of the Tripartite Negative Triangle problem with part sizes at most $n$, $n$, and $n^{\beta / 2}$, in $\Oh(n^2)$ time
    produces an alphabet $\Sigma$ of size $\Oh(n)$, a normalized weight function $w:\Esigma^2\to \Rz$, strings $X,Y\in \Sigma^{\Oh(n^{1+\beta})}$ such that $\eed{X}{Y} \le 4$, and a threshold $k =\Oh(n)$ such that the answer to the Tripartite Negative Triangle instance is ``YES'' if and only if $\wed(X, Y) \le k$.
\end{restatable}
\renewcommand{\maybe}{}

As we prove at the end of the section, \cref{thm:static-lb} is a consequence of \cref{lm:smallEDLBReduction} combined with the following result:

\begin{fact}[Vassilevska-Williams \& Williams~\cite{VWW18}]\label{fct:negtri}
    Let $\alpha,\beta,\gamma,\delta>0$ be real parameters. Assuming the APSP Hypothesis, there is no algorithm that, given a tripartite graph with part sizes at most $n^\alpha $, $n^\beta$, and $n^\gamma$, solves the Tripartite Negative Triangle problem in $\Oh(n^{\alpha+\beta+\gamma-\delta})$ time.
\end{fact}

As discussed in \cref{sec:lower-bounds}, our starting point is the hardness of the Batched Weighted Edit Distance problem (\cref{prob:batched}).
In order to derive \cref{lm:smallEDLBReduction}, we need to reinterpret the proof of~\cite[Theorem 6.13]{CKW23} as a reduction from Tripartite Negative Triangle.

\begin{theorem}[See {\cite[Theorem 6.13]{CKW23}}]\label{thm:BatchedWEDLBReduction}
    Let $\beta \in [0,1]$ be a real parameter.
    There is an algorithm that, given $n\in \ZZ_+$ and an instance of the Tripartite Negative Triangle problem with part sizes at most $n$, $n$, and $n^{\beta / 2}$, in $\Oh(n^2)$ time constructs an instance of the Batched Weighted Edit Distance problem with $m = \Oh(n^{\beta})$ strings of lengths $x\le y \le \Oh(n)$, an alphabet $\Sigma$ of size $\Oh(n)$, and a normalized weight function $w:\Esigma^2\to \Rz$ such that the answer to the Tripartite Negative Triangle instance is ``YES'' if and only if the answer for the Batched Weighted Edit Distance is ``YES''.
    Additionally, the constructed Batched Weighted Edit Distance instance satisfies the following properties:
    \begin{itemize}
        \item The values of the weight function $w$ are rationals with a
            common $O(\log n)$-bit integer denominator.
        \item We have $w(a,b)=w(b,a)\in [1,2]$ for all distinct $a,b\in
            \Esigma$.
        \item We have $w(a,\emptystring)=1$ if $a\in \Sigma$ occurs in any of the
            strings $X_i$.
        \item We have $w(\emptystring, b)=2$ if $b\in \Sigma$ occurs in $Y$.
        \item The threshold satisfies $k \in [2y-x,2y-x+1)$.
        \item Subsequent strings $X_i$ have small Hamming distances,
            that is, $\max_{i=1}^{m-1} \hhd{X_i}{X_{i+1}}=\Oh(n^{1-\beta})$.
    \end{itemize}
\end{theorem}

Consider an instance $X_1,X_2, \ldots X_m$, $Y$ and threshold $k$
for the Batched Weighted Edit Distance problem with the restrictions
stated in \cref{thm:BatchedWEDLBReduction}.

Let $\Sigma_X$ and $\Sigma_Y$ be the symbols that occur in any $X_i$ and in $Y$, respectively.
The conditions of \cref{thm:BatchedWEDLBReduction} guarantee that $\Sigma_X \cap \Sigma_Y = \emptyset$.
Let $h = \max(1, \max_{i=1}^{m-1}(\hhd{X_i}{X_{i+1}}))$ and 
\[r = (m-1)(8h+8) + 2x+k+6h+7.\]
We define a new alphabet 
$\hat \Sigma = \Sigma_X \cup \Sigma_Y \cup \{u_1,u_2 , \ldots u_r, v_1,v_2, \ldots v_r, \diamond, \bot, \top , \$ \}$ and extend $w$ to $\hat w$ as follows:
\[
    \hat w(a, b)
    =
    \begin{cases}
      0 & \text{if }  a=b \\  
      w(a , b) & \text{if }  a,b \in \Sigma_X \cup \Sigma_Y \cup \{\emptystring\},  \\
      2  & \text{if } a\ne b \text{ and }   \{ a,b\}\subseteq \{\top,\bot,\diamond\}, \\
      \infty & \text{if } a\ne b \text{ and } \$ \in \{a,b \},\\ 
      1  &\text{otherwise}. 
    \end{cases}
\]

We make the following observation regarding $\hat w$.
\begin{observation}\label{clm:whatsymmetrictriangle}
    The function $\hat w$ is symmetric ($\hat w(a, b)=\hat w(b , a)$ holds for every $a,b\in \hat \Sigma \cup \{\emptystring\}$) and satisfies the triangle inequality ($\hat w(a, c) \le \hat w(a, b) + \hat w(b, c)$ holds for every $a,b,c\in \hat \Sigma \cup \{\emptystring\}$). 
\end{observation}
\begin{proof}
    We prove symmetry and triangle inequality separately.
    \para{Symmetry.}
    For every $a,b \in \Sigma_X \cup \Sigma_Y \cup \{ \emptystring\}$, the symmetry of $\hat w$ follows from the symmetry of $w$.
    Otherwise, the value of $\hat w(a, b)$ is decided by the case satisfied by $a$ and $b$ in the definition of $\hat w$.
    Each case in the definition of $\hat w$ depends only on $\{a,b\}$ and is therefore symmetric with respect to $a$ and $b$ (i.e., each case is satisfied for $(a,b)$ if and only if it is satisfied for $(b,a)$).
    
    \para{Triangle inequality.}
    The triangle inequality is trivially satisfied if $a=b$ or $b=c$. 
    If $\$=a\ne b$ or $b \ne c = \$$, then $\hat w(a , c) \le \infty \le \hat w(a, b) + \hat w(b , c)$.
    In all the remaining cases, we have $\hat w(a , c) \le 2 = 1+1 \le \hat w(a, b) + \hat w(b , c)$.
\end{proof}

\Cref{clm:whatsymmetrictriangle} lets us apply the following fact to $\hat w$.
\begin{fact}[see {\cite[Corollary 3.6]{CKW23}}]\label{clm:justmatchpreforsuff}
    If a weight function \wdef satisfies the triangle inequality, then $\wwed{w}{aX}{aY}=\wwed{w}{X}{Y}=\wwed{w}{Xa}{Ya}$ holds for every $a\in \Sigma$ and $X,Y \in \Sigma^*$. 
\end{fact}

For every $i\in [1\dd m)$, we define the set $H_i = \setof{j\in [0\dd x)}{ X_i[j] \neq X_{i+1}[j]}$; for $i\in \ZZ \setminus [1\dd m)$, we set $H_i = \emptyset$.
  Note that the definition of the parameter $h$ implies $|H_i|\le h$ for every $i \in \ZZ$.
  Additionally, for every $i\in \ZZ$, we define $F_i\subseteq [0\dd x)$ to be an arbitrary superset of $H_i \cup H_{i+1}$ of size $2h$ exactly.\footnote{Notice that in order for $F_i$ to be well-defined, we must have $2h \le x$.
  Otherwise, i.e., if $x \le 2h$, we can decide $\wed(X_i,Y)\le k$ in $\Oh(xy)=\Oh(hn)$ time, for a total of $\Oh(mhn)=\Oh(n^{\beta+1-\beta+1})=\Oh(n^2)$ time across all $i\in [1\dd m]$.
  In this case, we can output a trivial YES or NO instance of the bounded weighted edit distance problem.}
  Notice that $F_i$ consists of at most $2h$ indices from $H_i \cup H_{i-1}$ and $2h-|H_i \cup H_{i-1}|$ additional arbitrary indices. 

  For convenience of presentation, let us define $X_0 = X_1$. 
  For every $i\in [0\dd m]$, we define the strings $X_i^\top$ and $X_i^\bot$ as follows:
\[
        X_i^\top[j]
        =
        \begin{cases}
        \top &\text{if } j \in F_{i-1}, \\  
        X_i[j] &\text{otherwise, }
        \end{cases}
        \qquad\text{and}\qquad
            X_i^\bot[j]
            =
            \begin{cases}
            \bot &\text{if } j \in F_{i}, \\  
            X_i[j] &\text{otherwise}. 
            \end{cases}
\]

Let $U=u_1u_2 \cdots u_r$ and $V = v_1v_2\cdots v_r$.
We define 
\begin{align*}
        \hat{X_p} &=U \cdot X_0^\top \cdot V \cdot \diamond Y \diamond \cdot U \cdot X_0^\bot \cdot V \cdot Y,  \\
        \hat X &= \bigodot_{i=1}^m \big ( U \cdot \diamond X_i \diamond \cdot V \cdot Y \cdot U \cdot X_i^\top \cdot V \diamond Y \diamond \cdot U \cdot X_i^\bot \cdot V \cdot Y \big ), \\
        \hat Y &= \bigodot_{i=1}^{m} \big (U \cdot X_{i-1}^\top \cdot V \cdot \diamond Y \diamond \cdot U \cdot X_{i-1}^\bot \cdot V \cdot Y \cdot U \cdot \diamond X_i \diamond \cdot V \cdot Y \big ) \cdot U \cdot X_m^\top \cdot V \cdot \diamond Y \diamond, \\
        \hat{Y_s} &= U \cdot X_m^\bot \cdot V \cdot Y.
\end{align*}

    Finally, we define $\tilde{X}= \hat{X_p} \cdot \$ \cdot \hat X \cdot \$$, $\tilde{Y} = \$ \cdot \hat Y \cdot \$ \cdot \hat{Y_s}$, $\hat k = (m-1)(8h+8)+ 2r + 2x + k + 6h + 6$, and $\tilde k = \hat k + 6r + 3x + 6y + 2 $.
    Notice that $\hat{X_p}\cdot \hat X = \hat Y \cdot \hat{Y_s}$, so we have $\eed{\tilde{X}}{\tilde{Y}}\le 4$ as witnessed by deleting all four $\$$ symbols.
    We highlight the following simple property of $\tilde X$ and $\tilde Y$, which is useful in another context.
    \begin{observation}\label{obs:dolarlessZ}
    There exists a string $Z$ such that $\ed(\tilde X,Z) =\ed(\tilde Y ,Z)=2$ and, for every $i\in \fragmentco{0}{|Z|}$, we have $\hat w(Z[i] , \emptystring) \le 2$.
    \end{observation}
    \begin{proof}
    It is easy to see that $Z=\hat X_p \cdot \hat X = \hat Y \cdot \hat Y_s$ satisfies \cref{obs:dolarlessZ}.
    \end{proof}

    Our main goal is to prove the following.
    \begin{lemma}\label{lem:tildeEqivalence}
        $\wwed{\hat w}{\tilde X}{\tilde Y}\le \tilde k$ if and only if $\wwed{w}{X_i}{Y}\le k$ for some $i\in [1\dd m]$.
    \end{lemma}

    We make the following simple observation.
    \begin{observation}\label{clm:fromtildetohat}
    $\wwed{\hat w }{\tilde X}{\tilde Y} = \wwed{\hat w}{\hat X}{\hat Y} + 6r+3x+6y+2$
    \end{observation}
    \begin{proof}
    Since $\hat w(\$,b) = \infty$ for every $b \neq \$$, any alignment of $\hat X$ onto $\hat Y$ with a finite cost must align the two $\$$ symbols in $\hat X$ with the two $\$$ symbols in $\hat Y$.
    It follows that $\wwed{\hat w}{\tilde X}{\tilde Y}= \wwed{\hat w}{\tilde{X_p}}{\emptystring}+\wwed{\hat w}{\hat X}{\hat Y}+ \wwed{\hat w}{\emptystring}{\hat{Y_s}} = \wwed{\hat w}{\hat X}{\hat Y} + 6r + 3x + 6y + 2$ as required.
    \end{proof}

    \cref{clm:fromtildetohat} reduces the proof of \cref{lem:tildeEqivalence} to the following.
    \begin{lemma}\label{lem:hatEquivalence}
        $\wwed{\hat w}{\hat X}{\hat Y}\le \hat k$ if and only if $\wwed{w}{X_i}{Y}\le k$ for some $i\in [1\dd m]$.
    \end{lemma}
    \begin{proof}
    Let us first analyze the weighted edit distance between pairs of gadgets that occur in our construction.
    This analysis is given below as \cref{clm:topbot4h,clm:topprevbot4h,clm:topbotdiamond2h,clm:topbotequalwed,clm:theoffsetchanger,clm:wrongVUdeletion,clm:wrongVUdeletion,clm:UVdeletion}.

    \begin{claim}\label{clm:topbotequalwed}
        Let $Z_\top \in (\Sigma_X \cup \{ \top \})^x$ and $Z_\bot \in (\Sigma_X \cup \{ \bot \})^x$ such that, for every $j\in [0\dd x)$, we have $Z_\top[j]=Z_\bot[j]$, $Z_\top[j]=\top$, or $Z_{\bot}[j]=\bot$.
        Let $c_\top$ and $c_\bot$ be the number of occurrences of $\top$ in $Z_\top$ and of $\bot$ in $Z_\bot$, respectively.
        It holds that $\wwed{\hat w}{Z_\top}{Z_\bot} = \wwed{\hat w}{Z_\bot}{Z_\top} = c_\top+c_\bot$.  
      \end{claim}
    \begin{claimproof}
        We prove that $\wwed{\hat w}{Z_{\top}}{Z_\bot}= c_\top + c_\bot$; the equality $\wwed{\hat w}{Z_\bot}{Z_\top} = c_\top+c_\bot$ then follows from symmetry of $\hat w$ (\cref{clm:whatsymmetrictriangle}).

        Let $F_\top = \setof{j\in [0\dd x)}{Z_\top[j] = \top}$ and $F_\bot = \setof{ j\in [0\dd x)}{Z_\bot[j] = \bot }$.
        \para{$\wwed{\hat w}{Z_\top}{Z_\bot} \le c_\top + c_\bot$:}
        Consider the alignment $\mathcal{A}$ that aligns $Z_\top[j]$ with $Z_\bot[j]$ for every $j\in [0\dd x)$.
        For every $j \notin F_{\top} \cup F_\bot$, we have that $\mathcal{A}$ matches $Z_\top[j]=Z_\bot[j]$, which induces no cost.
        For every $j\in F_{\top} \setminus F_\bot$, the alignment $\mathcal{A}$ substitutes $Z_\top[j] = \top$ with $Z_\bot[j] \in \Sigma_X$, which costs $1$.
        Similarly, for every $j\in F_\bot \setminus F_{\top}$, the alignment $\mathcal{A}$ aligns a $\bot$ symbol with a symbol from $\Sigma_X$ for a cost of $1$.
        Finally, for every index $j\in F_{\top} \cap F_\bot$, the alignment $\mathcal{A}$ aligns a $\top$ symbol with a $\bot$ symbol, which costs $2$.
        The total cost of $\mathcal{A}$ is therefore $|F_\top \setminus F_\bot| +|F_\bot \setminus F_\top| + 2 |F_\bot \cap F_\top| = |F_\top| + |F_\bot| = c_\top+c_\bot$.
        
        \para{$\wwed{\hat w}{Z_\top}{Z_\bot} \ge c_\top + c_\bot$:}
        Let $\mathcal{A}$ be a $\hat w$-optimal alignment of $Z_\top$ onto $Z_\bot$.
        Define \begin{align*}M &=\setof{(j_1,j_2)\in F_\top \times F_\bot}{\mathcal{A} \text{ substitutes } Z_\top[j_1] \text{ to } Z_\bot[j_2]},\\ 
            C_\top &= \setof{j \in F_\top}{\mathcal{A} \text{ does not substitute }Z_\top[j] \text{ to a } \bot \text{ symbol}},\text{ and } \\C_\bot &= \setof{j \in F_\bot }{ \mathcal{A} \text{ does not substitute }Z_\bot[j] \text{ to a } \top \textit{ symbol }}.\end{align*}
        Notice that $c_\top = |M| + |C_\top|$ and $c_\bot = |M| + |C_\bot|$.
    
        Since $\top$ does not occur in $Z_\bot$, every $j\in C_\top$ is either substituted or deleted by $\mathcal{A}$, which induces a cost of $1$.
        Similarly, since there are no $\bot$ symbols in $Z_\top$, every $j\in C_\bot$ is either substituted or inserted, which induces a cost of $1$.
        This contributes a cost of $|C_\top|+|C_\bot|$ to $\mathcal{A}$.
        Notice that there is no double counting because the first summand accounts for substitutions (or deletions) of a $\top$ symbol with a symbol in $\Sigma_X\cup \{\emptystring\}$, and the second accounts for substitutions (or insertions) of a symbol in $\Sigma_X \cup \{ \emptystring \}$ with a $\bot$ symbol.
        Finally, by definition, each pair in $M$ induces a cost of $2$.
        Notice again that this cost is due to substitutions disjoint from the ones already considered.
    
        It follows from the above that the cost of $\mathcal{A}$ is at least $|C_\top| + |C_\bot| + 2|M| = c_\top + c_\bot$, which concludes the proof.     
    \end{claimproof}

    \begin{claim}\label{clm:topbot4h}
    For every $i\in [0\dd m]$, we have $\wwed{\hat w}{X_i^\top}{X_{i}^\bot}=\wwed{\hat w}{X_i^\bot}{X_{i}^\top}=4h$.
    \end{claim}
    \begin{claimproof}
    By definition, for every $j\in[0\dd x)$, we have either $X_i^\top[j] = \top$, $X_i^\bot[j] = \bot$, or $X_i^\bot[j] = X_i[j] = X_i^\top[j]$.
    Additionally, $X_i^\top \in \{ \Sigma_X \cup \{ \top \}\}^x$ contains exactly $2h$ $\top$ symbols and $X_i^\bot \in \{ \Sigma_X \cup \{ \bot \}\}^x$ contains exactly $2h$ $\bot$ symbols.
    The claim therefore follows from \cref{clm:topbotequalwed,clm:whatsymmetrictriangle}.
    \end{claimproof}

    \begin{claim}\label{clm:topprevbot4h}
    For every $i\in [1\dd m]$, we have $\wwed{\hat w}{X_{i}^{\top}}{X_{i-1}^\bot}=\wwed{\hat w}{X_{i-1}^\bot}{X_{i}^{\top}} = 4h$.
    \end{claim}
    \begin{claimproof}
    By definition, $X_i^\top \in \{ \Sigma_X \cup \{ \top \}\}^x$ contains exactly $2h$ $\top$ symbols and $X_{i-1}^\bot \in \{ \Sigma_X \cup \{ \bot \}\}^x$ contains exactly $2h$ $\bot$ symbols.
    Let $j$ be an index such that $X_i^\top[j] \neq \top$ and $X_{i-1}^{\bot}[j] \neq \bot$.
    By the construction of $X_i^{\top}$, we have that $j \notin F_{i-1}$ and, in particular, $j\notin H_{i-1}$.
    Therefore, $X_{i-1}^\bot[j] = X_{i-1}[j] = X_{i}[j] = X_i^\top[j]$.
    We have shown that, for every index $j \in [0\dd x)$, either $X_i^\top[j] = \top$, $X_{i-1}^\bot[j] = \bot$, or $X_i^\top[j] = X_{i-1}^\bot[j]$.
    It follows that \cref{clm:topbotequalwed} applies and yields $\wwed{\hat w }{X_i^\top}{X_{i-1}^\bot} = 4h$, as required. \cref{clm:whatsymmetrictriangle} implies the symmetric statement due to the symmetry of $\hat w$.
    \end{claimproof}

    \begin{claim}\label{clm:topbotdiamond2h}
        For every $i\in [1\dd m]$ and every $Z\in \{ X_{i}^\top, X_i^\bot , X_{i-1}^\top, X_{i-1}^\bot\}$, we have $\wwed{\hat w}{\diamond X_i \diamond}{Z}=\wwed{\hat w}{Z}{\diamond X_i \diamond}=2h+2$.
    \end{claim}
    \begin{claimproof}
    We prove $\wwed{\hat w}{\diamond X_i \diamond}{Z} = 2h+2$; the equality $\wwed{\hat w}{Z}{\diamond X_i \diamond} = 2h+2$ follows by \cref{clm:whatsymmetrictriangle}.

    First, we claim that, for every index $j\in [0\dd x)$, we either have $Z[j]=X_i[j]$ or $Z[j]\in \{ \top,\bot \}$.
    For $Z\in \{ X_i^\top, X_i^\bot\}$, this is immediate from the definitions of $X_i^\top$ and $X_i^\bot$.
    For $Z\in \{X_{i-1}^\bot,X_{i-1}^\top \}$, it holds that, for every $j\in [0\dd x)$, either $Z[j]\in \{ \top,\bot \}$, or $j\notin H_{i-1}$ and $Z[j]=X_{i-1}[j]$, which leads to $Z[j] = X_i[j]$.
    Therefore, the alignment that deletes both diamonds of $\diamond X_i \diamond$ and then aligns $X_i[j]$ to $Z[j]$ using matches and substitutions only has a cost of exactly $2h+2$ (where the cost of $2h$ arises from the occurrences of $\top$ or $\bot$ characters in $Z$).

    Let $\mathcal{A}$ be an alignment of $\diamond X_i \diamond $ onto $Z$.
    Notice that there are no occurrences of $\top$ or of $\bot$ in $\diamond X_i \diamond$ and no occurrences of $\diamond$ in $Z$.
    Therefore, each occurrence of $\top$, $\bot$, and $\diamond$ symbols be involved in an edit.
    It follows that there are $2h+2$ edits involving $\top$, $\bot$, and $\diamond$ symbols in $\mathcal{A}$, each with cost at least $1$.
    The only situation in which an operation is counted twice is when a $\diamond$ symbol is substituted for a $\top$ or a $\bot$ symbol, in which case the operation induces a cost of $2$ instead of $1$.
    It follows that $\cost(\mathcal{A}) \ge 2h+2$ holds as required. 
\end{claimproof}

\begin{claim}\label{clm:theoffsetchanger}
For every $i\in [1\dd m]$, it holds that $\wwed{\hat w}{\diamond X_i \diamond}{X_{i-1}^{\top} \cdot V  \cdot \diamond Y \diamond \cdot U \cdot X_{i-1}^\bot} \le 2r + 2x + k$ if and only if $\wwed{w}{X_i}{Y}\le k$.    
\end{claim}
\begin{claimproof}
If $\wwed{w}{X_i}{Y}\le k$, then the claim cost can be achieved as follows:
\begin{multline*}
\wwed{\hat w}{\diamond X_i \diamond}{X_{i-1}^{\top} \cdot V  \cdot \diamond Y \diamond \cdot U \cdot X_{i-1}^\bot} \le\\
\wwed{\hat w}{\varepsilon}{X_{i-1}^{\top} \cdot V } + \wwed{\hat w}{X_i}{Y} + \wwed{\hat k}{\varepsilon}{ U \cdot X_{i-1}^\bot} \le x+r+k + r + x = 2x + 2r + k.
\end{multline*}

For the other direction, assume $\wwed{\hat w}{\diamond X_i \diamond}{X_{i-1}^{\top} \cdot V  \cdot \diamond Y \diamond \cdot U \cdot X_{i-1}^\bot} \le 2r + 2x + k$ and consider a $\hat w$-optimal alignment $\mathcal{A}:\diamond X_i \diamond\onto X_{i-1}^{\top} \cdot V  \cdot \diamond Y \diamond \cdot U \cdot X_{i-1}^\bot$. 
Let $c$ be the number of characters of $Y$ inserted by $\mathcal{A}$.
Due to length difference, $\mathcal{A}$ must insert at least $2r+y+x$ characters in total.
The total cost of insertions is therefore $2c + (2r+y+x-c)$.
Moreover, the characters of $Y$ cannot be matched due to $\Sigma_X \cap \Sigma_Y = \emptyset$.
It follows that each of the $y-c$ characters of $Y$ that is not inserted is involved in a substitution, which contributes at least $1$ to $\cost(\mathcal{A})$.
In total, insertions and substitutions involving characters of $Y$ account for $2r+2y+x$ units within the cost of $\mathcal{A}$.
Since $2r+2y+x = 2r + 2x + 2y-x > 2r+2x+ k -1 \ge \cost(\mathcal{A})-1$, it must be the case that $\mathcal{A}$ does not make any deletions or any substitutions that do not involve characters of $Y$.
Consequently, the $\diamond$ symbols of $\diamond X_i \diamond$ must match the $\diamond$ symbols of $X_{i-1}^{\top} \cdot V  \cdot \diamond Y \diamond \cdot U \cdot X_{i-1}^\bot$.
Therefore, we have
\[
\cost(\mathcal{A}) =
\wwed{\hat w}{\varepsilon}{X_{i-1}^{\top} \cdot V } + \wwed{\hat w}{X_i}{Y} + \wwed{\hat k}{\varepsilon}{ U \cdot X_{i-1}^\bot} = 
2x + 2r + \wwed{w}{X_i}{Y}\]
Finally, from $\cost(\mathcal{A}) \le 2r+2x+k$, we have that $\wwed{w}{X_i}{Y}\le k$ as required.
\end{claimproof}

\begin{claim}\label{clm:wrongVUdeletion}
    For every $i\in [1\dd m]$, we have $\wwed{\hat w}{X^\top_i}{X^\bot_{i-1} \cdot V \cdot Y \cdot U \cdot \diamond X_i \diamond} > 2r+2x + k$ as well as $\wwed{\hat w}{X^\bot_i}{\diamond X_i \diamond \cdot V \cdot Y \cdot U \cdot  X^\top_i } > 2r+2x + k$.
\end{claim}
\begin{claimproof}
    We prove the first inequality; the second inequality can be proven similarly.
    Consider an alignment $\mathcal{A} : X^\top_i \onto \diamond X_i \diamond \cdot V \cdot Y \cdot U \cdot  X^\top_i$.
    Let $c$ be the number of characters in $\Sigma_Y$ that $\mathcal{A}$ inserts.
    These insertions induce a cost of $2c$.
    Due to length difference, $\mathcal{A}$ has to insert at least $2r+y+x+2$ characters in total.
    It follows that there are at least $2r+y+x+2 -c$ insertions that does not involve $\Sigma_Y$ in $\mathcal{A}$.
    Notice that characters in $\Sigma_Y$ cannot be matched, so each of the $y-c$ characters that is not inserted is involved in a substitution, which contributes at least $1$ to $\cost(\mathcal{A})$.
    In total, insertions and substitutions involving characters of $Y$ account for at least $2r+2y+x+2$ units within the cost of $\mathcal{A}$.
    The claim follows from $\cost(\mathcal{A}) \ge 2r+2x + 1  + (2y - x + 1) > 2r + 2x +k $. 
\end{claimproof}

\begin{claim}\label{clm:UVdeletion}
    For every $i\in [1\dd m]$, the following holds:
    \begin{enumerate}
        \item $\wwed{\hat w}{Y}{Y \cdot U \cdot X^\top_{i}\cdot V \cdot \diamond Y \diamond } > 2r + 2x +k$,
        \item $\wwed{\hat w}{Y}{\diamond Y \diamond \cdot U \cdot X^\bot_{i}\cdot V \cdot Y } > 2r + 2x +k$,
        \item $\wwed{\hat w}{\diamond Y \diamond }{ Y  \cdot U \cdot \diamond X_i \diamond \cdot V \cdot Y } > 2r + 2x +k-1$.
    \end{enumerate}
\end{claim}
\begin{claimproof}
    The first two inequalities follow directly from \cref{clm:justmatchpreforsuff}:
    Namely, $\wwed{\hat w}{Y}{Y U X^\top_i V \diamond Y \diamond} = \wwed{\hat w}{\emptystring}{U X^\top_i V \diamond Y \diamond}= 2r + 2y + x + 2 > 2r + 2x + k$.
    The second inequality is analogous.

    As for the third inequality, consider an alignment $\mathcal{A}:\diamond Y \diamond \onto Y U \diamond X_i \diamond V Y$.
    Let $c$ be the number of characters from $\Sigma_Y$ that the alignment $\mathcal{A}$ inserts.
    These insertions contribute $2c$ to the alignment cost.
    Due to length difference, $\mathcal{A}$ must inserts at least $y+x + 2r$ characters in total, leading to an additional cost of at least $(y+x+2r-c)$.
    Out of the $2y$ characters from $\Sigma_Y$ in $Y U \diamond X_i \diamond V Y$, at most $y$ can be matched.
    Therefore, the remaining $y-c$ characters that are neither inserted nor matched by $\mathcal{A}$ induce an additional substitution cost of at least $y-c$.
    In conclusion, insertions and substitutions involving characters of $\Sigma_Y$ account for at least $2c + (y+x+2r-c) + (y-c) = 2y + x + 2r > 2r + 2x + k -1$ units in the cost of $\cA$.
\end{claimproof}

We are now ready to complete prove \cref{lem:hatEquivalence}.    
If $\wwed{\hat w}{X_i}{Y} \le k$ holds for some $i\in [1\dd m]$, then we align $\hat X$ and $\hat Y$ as follows:
\begin{align*}
    \wwed{\hat w}{\hat X}{\hat Y} &\le \sum_{i=1}^{j-1} \Big[ \wwed{\hat w}{U \diamond X_i \diamond}{U X^\top_{i-1} } + \wwed{\hat w}{VY}{V\diamond Y \diamond} + 
    \wwed{\hat w}{U X^\top_i}{UX^\bot_{i-1}}\\
    &\quad\qquad + \wwed{\hat w}{V\diamond Y \diamond}{VY} + \wwed{\hat w}{U X^\bot_i}{U \diamond X_i \diamond }+ \wwed{\hat w}{VY}{VY}  \Big] \\
    &\quad+
    \wwed{\hat w}{U \diamond X_j \diamond }{U \cdot X^\top_{j-1} \cdot V \cdot \diamond Y \diamond \cdot U \cdot X^\bot_{j-1}} \\
    &\quad+
    \wwed{\hat{w}}{V Y}{V Y}+
    \wwed{\hat w}{U \cdot X^\top_j}{U \diamond X_j \diamond }+
    \wwed{\hat w}{V \diamond Y \diamond }{V \cdot Y}\\
    &\quad+
    \wwed{\hat w}{U \cdot X^\bot_j}{U \cdot X^\top_j}+
    \wwed{\hat w}{VY}{V \cdot \diamond Y \diamond}\\
    &\quad+
    \sum_{i=j+1}^{m} \Big[ 
    \wwed{\hat w}{U \diamond X_i \diamond}{U X^\bot_{i-1}}+
    \wwed{\hat w}{VY}{VY}+
    \wwed{\hat w}{U X^\top_i}{U \diamond X_i \diamond}\\
    &\qquad\qquad+
    \wwed{\hat w}{V \diamond Y \diamond}{VY}+
    \wwed{\hat w}{U \cdot X^\bot_i}{U \cdot  X^\top_i}+
    \wwed{\hat w}{VY}{V \diamond Y \diamond}
    \Big]\\
    &\leftstackrel{\text{\ref{clm:topbot4h}--\ref{clm:topbotdiamond2h}}}{=}
    (j-1)(8h+8)+ \wwed{\hat x}{\diamond X_j \diamond}{X^\top_{j-1} V \diamond Y \diamond U X^\bot_{j-1}}+0+2h+2 \\
    &\qquad+ 2 + 4h + 2+ (m-j)(8h+8) \\
    &\leftstackrel{\text{\ref{clm:theoffsetchanger}}}{\le}
    (m-1)(8h+8) + 2x + 2r + k + 6h+6\\ 
    &= \hat k
\end{align*}
as required. 

For the other direction, let us assume that there is an alignment $\mathcal{A} : \hX \onto \hY$ with cost at most $\hat k$.
We observe that $3r= 2r+ (m-1)(8h+8) + 2x + k + 6h + 7 > \hat k$.
By the pigeonhole principle, there is an index $q\in [1\dd r]$ such that at most two $v_q$ and $u_q$ symbols in $\hat Y$ are not matched by $\mathcal{A}$.
Observe that these characters form a subsequence of $\hat Y$ equal to $(u_q v_q)^{3m + 1}$ and a subsequence of $\hat X$ equal to $(u_q v_q)^{3m}$.
It follows that all $u_q$ and $v_q$ symbols of $\hat X$ must be matched perfectly.
The only way to achieve that is by leaving out two consecutive characters of $(u_q v_q)^{3m+1}$.
It follows from \cref{clm:justmatchpreforsuff} that, if a $u_q$ or a $v_q$ is matched, the entire enclosing $U$ and $V$ strings are matched as well.
We consider several cases depending on which two consecutive occurrences of $U$ and $V$ in $\hat Y$ are not matched by $\mathcal{A}$.
Let $j\in [1\dd m]$ be the maximal index such that the occurrence of $\diamond X_{j-1} \diamond V Y $ in $\hat Y$ is to the left of the occurrences of $U$ and $V$ that are not matched by $\mathcal{A}$ (If there is no such index, we set $j=0$).
The alignment $\mathcal{A}$ can be decomposed into three parts:
\begin{enumerate}
    \item An alignment \[\mathcal{A}_1 : \bigodot_{i=1}^{j-1}U \diamond X_i \diamond VY U X^\top_i V \diamond Y \diamond U X^\bot_i V Y \onto \bigodot_{i=1}^{j-1}U X^\top_{i-1} V \diamond Y \diamond U X^\bot_{i-1} V Y U \diamond X_i \diamond VY \]
    in which all $U$ and $V$ occurrences are matched.
    \item An alignment \[\mathcal{A}_2 : 
    U \diamond X_j \diamond V Y U X^\top_j V \diamond Y \diamond U X^\bot_j V Y \onto 
    U X^\top_{j-1} V \diamond Y \diamond U X^\bot_{j-1} V Y U \diamond X_j \diamond V Y U X^\top_j V \diamond Y \diamond\]
    in which all but two consecutive $U$ and $V$ occurrences are matched.
    \item An alignment \[\mathcal{A}_3:
    \bigodot_{i=j+1}^{m}U \diamond X_i \diamond VY U X^\top_i V \diamond Y \diamond U X^\bot_i V Y \onto 
    \bigodot_{i=j+1}^{m}U X^\bot_{i-1} V Y U \diamond X_i \diamond VY U X^\top_{i} V \diamond Y \diamond\]
    in which all $U$ and $V$ occurrences are matched.
\end{enumerate} 
The costs of $\mathcal{A}_1$ and $\mathcal{A}_3$ can be analyzed independently of $\cA_2$:
\begin{align*}
    \cost(\mathcal{A}_1)&= \sum_{i=1}^{j-1} \Big[ \wwed{\hat w}{U \diamond X_i \diamond}{U X^\top_{i-1} } + \wwed{\hat w}{VY}{V\diamond Y \diamond} + 
    \wwed{\hat w}{U X^\top_i}{UX^\bot_{i-1}}\\
    &\qquad\qquad + \wwed{\hat w}{V\diamond Y \diamond}{VY} + \wwed{\hat w}{U X^\bot_i}{U \diamond X_i \diamond }+ \wwed{\hat w}{VY}{VY}  \Big] \\
    &\leftstackrel{\text{\ref{clm:topbot4h},\ref{clm:topbotdiamond2h}}}{=}
    (j-1) \Big[ (2+2h) + 2 + 4h + 2 + (2h+2)+0  \Big]\\
    &=
    (j-1)(8h+8)
\end{align*}
and 
\begin{align*}
    \cost(\mathcal{A}_3) &= \sum_{i=j+1}^{m} \Big[ 
    \wwed{\hat w}{U \diamond X_i \diamond}{U X^\bot_{i-1}}+
    \wwed{\hat w}{VY}{VY}+
    \wwed{\hat w}{U X^\top_i}{U \diamond X_i \diamond}\\
    &\qquad\qquad+  \wwed{\hat w}{V \diamond Y \diamond}{VY}+
    \wwed{\hat w}{U \cdot X^\bot_i}{U \cdot  X^\top_i}+
    \wwed{\hat w}{VY}{V \diamond Y \diamond} \Big]\\
    &\leftstackrel{\text{\ref{clm:topbot4h},\ref{clm:topbotdiamond2h}}}{=}
    (m-j)\Big[ (2h + 2) + 0 + (2h+2) + 2 + 4h +2 \Big] \\
    &= (m-j)(8h+8)
\end{align*}

From $\hat k \ge \cost(\mathcal{A}) = \cost(\mathcal{A}_1)+\cost(\mathcal{A}_2)+\cost(\mathcal{A}_3)=\cost(\mathcal{A}_2)+(m-1)(8h+8)$, we conclude that $\cost(\mathcal{A}_2) \le 2r + 2x + k + 6h + 6$.
We are left with the task of analyzing the cost of $\mathcal{A}_2$.
We do so by considering several cases regarding the identity of the $U$ and $V$ occurrences not matched by $\mathcal{A}$.

\para{Case 1:} The non-matched occurrences are $V$ and $U$ such that $V$ is preceded by $X_{j-1}^\top$.

It follows that 
\begin{align*}
    \cost(\cA_2) &= 
    \wwed{\hat w}{U \diamond X_j \diamond }{U \cdot X^\top_{j-1} \cdot V \cdot \diamond Y \diamond \cdot U \cdot X^\bot_{j-1}} \\
    &\qquad+\wwed{\hat{w}}{V Y}{V Y}+
    \wwed{\hat w}{U \cdot X^\top_j}{U \diamond X_j \diamond }+
    \wwed{\hat w}{V \diamond Y \diamond }{V \cdot Y}\\
    &\qquad+\wwed{\hat w}{U \cdot X^\bot_j}{U \cdot X^\top_j}+
    \wwed{\hat w}{VY}{V \cdot \diamond Y \diamond}\\
    &\leftstackrel{\text{\ref{clm:topbot4h},\ref{clm:topbotdiamond2h}}}{=}\wwed{\hat w}{\diamond X_j \diamond}{X^\top_{j-1} V \diamond Y \diamond U X^\bot_{j-1}}+0+2h+2 + 2 + 4h + 2 \\
    &=\wwed{\hat w}{\diamond X_j \diamond}{X^\top_{j-1} V \diamond Y \diamond U X^\bot_{j-1}} + 6h+6.
\end{align*}

From $\cost(\mathcal{A}_2) \le 2x + 2r + k + 6h + 6$, we have that $\wwed{\hat w}{\diamond X_j \diamond}{X^\top_{j-1} V \diamond Y \diamond U X^\bot_{j-1}} \le 2x + 2r + k$.
According to \cref{clm:theoffsetchanger}, we have $\wwed{ w}{ X_j}{Y}\le k$ as required.

We proceed to show that every other choice of non-aligned occurrences of $U$ and $V$ leads to a contradiction with $\cost(\mathcal{A}_2)\le 2x+2r+k+6h+6$.

\para{Case 2:} The non-matched occurrences are $V$ and $U$ such that $V$ is preceded by $X^{\bot}_{j-1}$.

It follows that
\begin{align*}
    \cost(\cA_2) &= 
    \wwed{\hat w}{\diamond X_j \diamond}{X^\top_{j-1}}+
    \wwed{\hat w}{Y}{\diamond Y \diamond}+
    \wwed{\hat w}{ X^\top_{j} }{ X^\bot_{j-1} \cdot V \cdot Y \cdot U \cdot \diamond X_{j}\diamond } \\
    &\qquad+
    \wwed{\hat w}{\diamond Y \diamond }{ Y}+
    \wwed{\hat w}{X^\bot_j}{X^\top_j }+
    \wwed{\hat w}{ Y }{\diamond Y \diamond}\\
    &\leftstackrel{\text{\ref{clm:topbot4h},\ref{clm:topbotdiamond2h}}}{=} 2 + 2h + 2+ \wwed{\hat w}{ X^\top_{j} }{ X^\bot_{j-1} \cdot V \cdot Y \cdot U \cdot \diamond X_{j}\diamond }+2+4h+2 \\
    &\leftstackrel{\text{\ref{clm:wrongVUdeletion}}}{\ge} 
    2x + 2r + k + 6h+8
\end{align*}

A contradiction.

\para{Case 3:} The non-matched occurrences are $V$ and $U$ such that $V$ is preceded by $\diamond X_{j}\diamond$.
It follows that
\begin{align*}
    \cost(\cA_2)&=
    \wwed{\hat w}{\diamond X_j \diamond}{X^\top_{j-1}}+
    \wwed{\hat w}{Y}{\diamond Y \diamond}+
    \wwed{\hat w}{X^\top_j}{X^\bot_{j-1}}
    \\
    &\qquad+
    \wwed{\hat w}{\diamond Y \diamond}{Y}+
    \wwed{\hat w}{ X^\bot_{j} }{ \diamond X_{j}\diamond  \cdot V \cdot Y \cdot U \cdot  X^\top_{j}}+
    \wwed{\hat w}{ Y }{\diamond Y \diamond}\\
    &\leftstackrel{\text{\ref{clm:topprevbot4h},\ref{clm:topbotdiamond2h}}}{=}
    2 + 2h + 2+4h + 2+ \wwed{\hat w}{ X^\bot_{j} }{ \diamond X_{j}\diamond  \cdot V \cdot Y \cdot U \cdot  X^\top_{j}}+2\\
    &\leftstackrel{\text{\ref{clm:wrongVUdeletion}}}{\ge}
    2x + 2r + k + 6h+8 .
\end{align*} 

\para{Case 4:} The non-matched occurrences are $U$ and $V$, and the non-matched $U$ is followed by $X^\bot_{j-1}$.

In this case, we have that
\begin{align*}
    \cost(\cA_2)&= 
    \wwed{\hat w}{\diamond X_j \diamond}{X^\top_{j-1}}+
    \wwed{\hat w}{Y}{\diamond Y \diamond U  X^\bot_{j-1} V Y}+
    \wwed{\hat w}{X^\top_j}{\diamond X_{j} \diamond}
    \\
    &\qquad+
    \wwed{\hat w}{ \diamond Y \diamond  }{ Y}+
    \wwed{\hat w}{X^\bot_j}{ X^\top_j} + 
    \wwed{\hat w}{Y}{\diamond Y \diamond}\\
    &\leftstackrel{\text{\ref{clm:topbot4h},\ref{clm:topbotdiamond2h}}}{=}
    2h + 2 + \wwed{\hat w}{Y}{\diamond Y \diamond U  X^\bot_{j-1} V Y}+2h + 2 + 2 + 4h + 2\\
    &\leftstackrel{\text{\ref{clm:UVdeletion}}}{\ge}
    2x + 2r + k + 8h+8.
\end{align*}

\para{Case 5:} The non-matched occurrences are $U$ and $V$, and the non-matched $U$ is followed by $\diamond X_{j} \diamond$.

In this case, we have that
\begin{align*}
    \cost(\cA_2)&= 
    \wwed{\hat w}{\diamond X_j \diamond}{X^\top_{j-1}}+
    \wwed{\hat w}{Y}{\diamond Y \diamond}+
    \wwed{\hat w}{X^\top_j }{X^\bot_{j-1}}
    \\
    &\qquad+
    \wwed{\hat w}{\diamond Y \diamond}{Y U \diamond X_j \diamond V Y}+
    \wwed{\hat w}{X^\bot_j}{X^\top_j}+
    \wwed{\hat w}{ Y }{\diamond Y \diamond}\\
    &\leftstackrel{\text{\ref{clm:topbot4h}--\ref{clm:topbotdiamond2h}}}{=}
    (2h + 2) + 2 + 4h +  \wwed{\hat w}{\diamond Y \diamond}{Y U \diamond X_j \diamond V Y}+4h + 2\\
    &\leftstackrel{\text{\ref{clm:UVdeletion}}}{\ge}
    2x + 2r + k + 10h+5 .
\end{align*}
This also contradicts $\cost(\cA_2) \le 2x+2r+k+6h+6$ due to $h\ge 1$.

\para{Case 6:} The non-matched occurrences are $U$ and $V$, and the non-matched $U$ is followed by $X^\top_{j}$.

In this case, we have that
\begin{align*}
    \cost(\cA_2)&= 
    \wwed{\hat w}{\diamond X_j \diamond}{X^\top_{j-1}}+
    \wwed{\hat w}{Y}{\diamond Y \diamond}+
    \wwed{\hat w}{X^\top_j }{X^\bot_{j-1}}\\
    &\qquad+
    \wwed{\hat w}{\diamond Y \diamond}{Y}+
    \wwed{\hat w}{X^\bot_j}{\diamond X_j \diamond}+
    \wwed{\hat w}{ Y }{Y U X^\top_j V \diamond Y \diamond}\\
    &\leftstackrel{\text{\ref{clm:topprevbot4h},\ref{clm:topbotdiamond2h}}}{=}
    (2h + 2) + 2 + 4h + 2 + (2h+2)+ \wwed{\hat w}{ Y }{Y U X^\top_j V \diamond Y \diamond} \\
    &\leftstackrel{\text{\ref{clm:UVdeletion}}}{\ge}
    2x + 2r + k + 8h+8.\qedhere
\end{align*}

\end{proof}

Recall that, via \cref{clm:fromtildetohat}, we have that \cref{lem:tildeEqivalence} also holds.
We are now ready to prove
\thmstaticlbreduction*
\begin{proof}
    Based on the given instance of the Tripartite Negative Triangle problem, we first apply \cref{thm:BatchedWEDLBReduction} to obtain an instance of the Batched Weighted Edit Distance problem with inputs $\Sigma$, $w$, $X_1,X_2,\ldots, X_m$, $Y$, and $k$.
    We construct $\tilde{X}$, $\tilde{Y}$, $\tilde{k}$, and $\hat w$ as defined earlier in this section so that $\ed(\tilde{X},\tilde{Y})\le 4$.
    By \cref{lem:tildeEqivalence}, we have $\wwed{\hat w}{\tilde{X}}{\tilde{Y}} \le \tilde{k}$ if and only if $\min_{i=1}^{m}\wed(X_i,Y)\le k$.
    By \cref{thm:BatchedWEDLBReduction} the latter condition holds if and only if the answer to the Tripartite Negative Triangle instance is ``YES''.

    Let us analyze the size of the produced instance of the weighted edit distance problem. 
    We have $x \le y = \Oh(n)$, $m= \Oh(n^\beta)$, and $h = \Oh(n^{1-\beta})$.
    Consequently, $r = (m-1)(8h+8)+2x+k+6h+7 = \Oh(n)$.
    Moreover,
    \begin{align*}
        |\hat \Sigma|&= |\Sigma_X| + |\Sigma_Y| + 2r+4 \le x + (m-1)h + y + \Oh(n) = \Oh(n).
        \\
        \tilde k&= (m-1)(8h+8) + 8r +5x + 6y+ k + 6h + 8
        \\
        &< \Oh(n) + \Oh(n) + 8y+4x+6h+9= \Oh(n)
        \\
        |\tilde Y|= |\tilde X|&=4r+ 2x + 2 + 3(m-1)(2r+x+y+2)+2=\Oh(nm)=\Oh(n^{1+\beta}).
    \end{align*}

    Finally, note that the strings $\tilde{X}$, $\tilde{Y}$ and the weight function $\hat w$ can be constructed in $\Oh(|\tilde{X}|+|\tilde{Y}|+|\hat{\Sigma}|^2)=\Oh(n^{1+\beta}+n^2)=\Oh(n^2)$ time.
    The reduction of \cref{thm:BatchedWEDLBReduction} also takes $\Oh(n^2)$ time. \end{proof}

\thmstaticlb*
\begin{proof}
    Let us first pick $\kappa\in [0.5,1)$ and $\delta>0$.
    For a proof by contradiction, suppose that there exists an algorithm $A$ that, given strings $X,Y\in \Sigma^{\le n}$ such that $\ed(X,Y)\le 4$, a threshold $k\le n^{\kappa}$, and oracle access to a normalized weight function $w:\Esigma^2\to \Rz$, in $\Oh(n^{0.5+1.5\kappa - \delta})$ time decides whether $\wed(X,Y)\le k$.

    Let us pick $\beta=\frac{1}{\kappa}-1 \in (0,1]$ and an instance of the Tripartite Negative Triangle problem with parts of size at most $N$, $N$, and $N^{\beta/2}$.
    We can use \cref{lm:smallEDLBReduction} to derive, in $\Oh(N^2)$ time, an equivalent instance of the weighted edit distance problem with strings $X,Y$ of length $\Oh(N^{1+\beta})$ satisfying $\ed(X,Y)\le 4$ and a threshold $k=\Oh(N)$.
    For a sufficiently large $n = \Oh(N^{1+\beta})=\Oh(N^{1/\kappa})$, we have $X,Y\in \Sigma^{\le n}$ and $k \le n^{\kappa}$.
    Thus, the algorithm $A$ can decide $\wed(X,Y)\le k$ in time
    \[
    \Oh(n^{0.5+1.5\kappa-\delta}) = \Oh(N^{(0.5+1.5\kappa-\delta)/\kappa})= \Oh(N^{1.5+0.5\cdot (1+\beta)-\delta/\kappa})= \Oh(N^{2+\beta/2 - \delta/\kappa})
    \]
    The overall time spent on solving the given instance of the Tripartite Negative Triangle problem is $\Oh(N^2 + n^{0.5+1.5\kappa-\delta}) = \Oh(N^{\max(2, 2+\beta/2 - \delta/\kappa)})$. Due to $\beta,\delta>0$, we have $\max(2, 2+\beta/2 - \delta/\kappa)\le 2+\beta/2-\delta'$ for $\delta'\coloneqq \min(\beta/2,\delta/\kappa)>0$, and thus the derived algorithm violates the APSP Hypothesis in the light of \cref{fct:negtri}.

    It remains to consider the case of $\kappa=1$.
    For this, we observe that \cref{thm:static-lb} specialized to $(1,\delta)$ follows from the instance specialized to $(1-\tfrac13\delta,\tfrac12\delta)$ because $0.5+1.5(1-\tfrac13\delta)-\tfrac12\delta=0.5+1.5-\delta$, and we can assume without loss of generality that $\delta \in (0,1)$ so that $1-\tfrac13\delta\in (\tfrac23, 1)$.
\end{proof}

\section{Dynamic Lower Bounds}\label{app:dynamic-lower-bounds}

If we compare the algorithm of \cref{lm:simple-complete-algorithm-upgraded} with the algorithm of \cite[Theorem 8.11]{GK24} for dynamic bounded unweighted edit distance, we can notice a slight difference in the parameters used in the time complexities.
While in \cite[Theorem 8.11]{GK24} the time complexity depends on the current value of $\wed(X, Y)$, the time complexity in \cref{lm:simple-complete-algorithm-upgraded} uses a universal upper bound $k$ on $\wed(X, Y)$ that stays unchanged throughout the lifetime of the algorithm.
Clearly, dependence on the current value of $\wed(X, Y)$ is preferable (see \cite[Corollary 8.12]{GK24}) but as shown in the following lemma, it is not achievable in the weighted case.
The hardness stems from the fact that weighted edit distance of two strings (in contrast with unweighted edit distance) can change significantly after a single edit, and thus we can solve the static problem of \cref{thm:static-lb} using a constant number of ``expensive'' updates.
Therefore, if we want to have an algorithm with update time depending on the current value of $\wed(X, Y)$, there is no hope to surpass the static lower bound.

\begin{lemma} \label{lm:dynamic-lb-variable-k}
    Consider the following dynamic problem:
    Given an integer $n \ge 1$ and $\Oh(1)$-time oracle access to a normalized weight function $w:\Esigma^2\to \mathbb{R}_{\ge 0}$, maintain two initially empty strings $X,Y\in \Sigma^{\le n}$ that throughout the lifetime of the algorithm satisfy $\ed(X, Y) \le 4$ subject to updates (character edits) and output $\wed(X, Y)$ after every update.
    Suppose that there is an algorithm that solves this problem with $T_P(n)$ preprocessing and $T_U(n, \wed(X, Y))$ update time, where $T_U$ is non-decreasing.
    If $T_P(n) = \Oh(n^{0.5 + 1.5 \kappa - \delta})$, $T_U(n, 2) = \Oh(n^{1.5 \kappa - 0.5 - \delta})$, and $T_U(n, n^{\kappa}) = \Oh(n^{0.5 + 1.5 \kappa - \delta})$ hold for some real parameters $\kappa \in [0.5, 1]$ and $\delta > 0$, then the APSP Hypothesis fails.
\end{lemma}

\label{app:dynamic-lower-bounds:variable-k}

\begin{proof}
    For a proof by contradiction, suppose that there is a dynamic algorithm, $\kappa \in [0.5, 1]$, and $\delta > 0$, with $T_P(n) = \Oh(n^{0.5 + 1.5 \kappa - \delta})$, $T_U(n, 2) = \Oh(n^{1.5 \kappa - 0.5 - \delta})$, and $T_U(n, n^{\kappa}) = \Oh(n^{0.5 + 1.5 \kappa - \delta})$.
    We use this dynamic algorithm to solve the problem from \cref{thm:static-lb} in $\Oh(n^{0.5+1.5\kappa - \delta / 3})$ time, thus violating the APSP Hypothesis.
    Let $\tilde X, \tilde Y \in \Sigma^{\le n}$ be the input strings and $k \le n^{\kappa}$ be the threshold.
    As we want to decide whether $\wed(\tilde X, \tilde Y) \le k$, we may cap the weight function $w$ by $k+1$ from above and without loss of generality assume that we have $w \colon \Esigma^2 \to [0, k + 1]$.
    Due to \cref{obs:dolarlessZ}, we can assume that the strings $\tilde X$ and $\tilde Y$ are given alongside a string $Z$ with $\ed(Z,\tilde X), \ed(Z, \tilde Y) \le 2$ such that every symbol $\sigma$ that occurs in $Z$ satisfies $w(\sigma,\emptystring)\le 2$.

    We start the dynamic algorithm with two empty strings $X$ and $Y$ in $T_P(n)$ time.
    We then perform $|Z|$ updates to $X$, and after each update to $X$, we mimic the same update on $Y$.
    The $i$-th update in this sequence simply appends $Z[i]$ to the end of $X$.
    At every point in time we have $\wed(X, Y) \le 2$, and thus this process takes $\Oh(n) + 2 n \cdot T_U(n, 2)$ time.
    We then use \cref{lm:baseline-wed} to find the optimal unweighted alignments of $Z$ onto $\tilde X$ and $\tilde Y$ in time $\Oh(n \cdot (\ed(Z, \tilde X) + 1) + n \cdot (\ed(Z, \tilde Y) + 1)) = \Oh(n)$.
    Guided by the alignments, we apply at most $2$ edits to each of $X$ and $Y$ to transform $X$ into $\tilde X$ and $Y$ into $\tilde Y$.
    As $w$ is capped by $k+1$, we have $\wed(X, Y) \le 4 n^{\kappa} + 4$ during this process, and it takes $\Oh(n)+4 \cdot T_U(n, 4 n^{\kappa} + 4)$ time.
    When $X = \tilde X$ and $Y = \tilde Y$, we are given $\wed(\tilde X, \tilde Y)$, and we use it to answer the problem of \cref{thm:static-lb}.
    For all sufficiently large $n$, the algorithm takes $\Oh(n)+T_P(n) + 2n \cdot T_U(n, 2) + 4 T_U(n, 4 n^{\kappa} + 4) \le \Oh(n)+T_P(n) + 2n \cdot T_U(n, 2) + 4 T_U(n^{1 + \delta / (6 \kappa)}, n^{\kappa + \delta / 6}) \le \Oh(n^{0.5 + 1.5 \kappa - \delta / 3})$ time, thus violating the APSP Hypothesis.
    It is easy to verify that $\ed(X, Y) \le 4$ holds throughout the lifetime of the algorithm.
\end{proof}
\begin{remark}
    In \cref{lm:dynamic-lb-variable-k}, we assumed that the preprocessing algorithm can access the entire weight function despite the fact that the strings $X$ and $Y$ are initially empty.
    In an alternative model where this is not allowed (for example, the costs of edits involving a character $a\in \Sigma$ are revealed only when $a$ appears in $X$ or $Y$ for the first time), we can allow for $T_P(n)=n^{\Oh(1)}$ instead of $T_P(n) = \Oh(n^{0.5 + 1.5 \kappa - \delta})$.
    This is because our lower bounds are obtained via reductions from the self-reducible Tripartite Negative Triangle problem, i.e., up to subpolynomial factors, solving a batch of $m=n^{\Oh(1)}$ equal-size instances of Tripartite Negative Triangle is not easier than solving $m$ instances one by one.
    Tasked with solving many instances of this problem, we can run the preprocessing algorithm just once and then, for each instance, perform appropriate updates. 
    If we store the address and the previous value of each modified memory cell, then we can restore the original state of the data structure at no additional asymptotic cost and proceed with handling the next instance. 
\end{remark}

As a corollary of \cref{lm:dynamic-lb-variable-k} for $\kappa = 1$, assuming the APSP Hypothesis, there is no algorithm for dynamic weighted edit distance with $\poly(|X| + |Y|)$ preprocessing and $\Oh(\sqrt{(|X| + |Y| + 1) \cdot (k+1)^{3-\delta}})$ update time, where $k$ is the maximum of $\wed(X, Y)$ before and after the update.
On the other hand, $\tOh(\sqrt{(|X| + |Y| + 1) \cdot (k + 1)})$ update time is trivially achievable by maintaining dynamic string implementation supporting $\tOh(1)$-time equality tests between fragments of $X$ and $Y$ \cite{CKW20} and running the optimal static algorithm of \cite{CKW23} after every update.

\cref{lm:dynamic-lb-variable-k} motivates why \cref{lm:simple-complete-algorithm-upgraded} uses a universal bound $k$ on the value of $\wed(X, Y)$.
Another aspect of \cref{lm:simple-complete-algorithm-upgraded} one could hope to improve is the preprocessing time.
The following lemma shows why it is improbable.

\begin{lemma} \label{lm:dynamic-lb-fixed-k}
    Let $\kappa \in (0, 1)$, $\mu \in [\max(1 - 2 \kappa, 0), 1 - \kappa)$, and $\delta > 0$ be real parameters.
    Assuming the APSP Hypothesis, there is no algorithm that given an integer $n \ge 1$, two strings $X, Y \in \Sigma^{\le n}$, a threshold $k \le \Oh(n^{\kappa})$, and $\Oh(1)$-time oracle access to a normalized weight function $w:\Esigma^2\to \mathbb{R}_{\ge 0}$, maintains $Y$ dynamically subject to $\Oh(n^{\mu})$ updates (each represented as an edit) and computes $\wed_{\le k}(X, Y)$ after every update in $\Oh(n^{0.5 + 0.5 \mu + 1.5 \kappa - \delta})$ time in total.
    No such algorithm exists even if strings $X$ and $Y$ are guaranteed to satisfy $\ed(X, Y) \le 4$ throughout the lifetime of the algorithm.
\end{lemma}

\begin{proof}
    We assume that $\delta < (1 - \kappa - \mu) / 2$ holds.
    If it is not the case, we replace $\delta$ with $(1 - \kappa - \mu) / 4$, which only tightens the lemma statement.
    Let $\ell \coloneqq n^{1 - \mu}$, $\kappa' \coloneqq \kappa / (1 - \mu)$, $\varepsilon \coloneqq \delta / (0.5 + 0.5 \mu + 1.5 \kappa)$, and $m \coloneqq n^{\mu}$.
    Consider an arbitrary instance of Tripartite Negative Triangle problem with parts of sizes at most $\ell^{\kappa'}$, $\ell^{\kappa'}$, and $m \cdot \ell^{(1 - k') / 2}$.
    According to \cref{fct:negtri}, assuming the APSP Hypothesis there is no algorithm that solves all such instances in time $\Oh((\ell^{\kappa'} \cdot \ell^{\kappa'} \cdot (m \cdot \ell^{(1 - k') / 2}))^{1 - \varepsilon}) = \Oh(n^{0.5 + 0.5 \mu + 1.5 \kappa - \delta})$.
    By splitting the third part of the graph into $m$ equal-sized subparts, we can reduce this instance of Tripartite Negative Triangle to $m$ instances of Tripartite Negative Triangle each with part sizes at most $\ell^{\kappa'}$, $\ell^{\kappa'}$, and $\ell^{(1 - k') / 2}$.
    By \cref{lm:smallEDLBReduction} for $\beta = (1-\kappa')/\kappa'$, in $\Oh(m \cdot \ell^{2 \kappa'}) = \Oh(n^{2 \kappa + \mu}) \le \Oh(n^{0.5 + 0.5 \mu + 1.5 \kappa - \delta})$ time we can compute strings $X_0, \ldots, X_{m-1}, Y_0, \ldots, Y_{m-1} \in \Sigma^{\le \Oh(\ell)}$, satisfying $\max_{i=0}^{m-1} \ed(X_i, Y_i) \le 4$, oracle access to a normalized weight function $\wdef$, and thresholds $k_0, \ldots, k_{m-1} \le \Oh(n^{\kappa})$ such that the initial instance of Tripartite Negative Triangle is a ``YES'' instance if and only if $\wed(X_i, Y_i) \le k_i$ holds for at least one $i \in \fragmentco{0}{m}$.
    Assuming the existence of the dynamic algorithm from the lemma statement, we will check all inequalities $\wed(X_i, Y_i) \le k_i$ in $\Oh(n^{0.5 + 0.5 \mu + 1.5 \kappa - \delta})$ time, thus contradicting the APSP Hypothesis.

    Let $\hk \coloneqq \max_{i=0}^{m-1} k_i$.
    Extend the alphabet $\Sigma$ with a new symbol $\fancysymbol$ satisfying $w(\fancysymbol, c) = w(c, \fancysymbol) = \hk + 1$ for all $c \in \Esigma \setminus \set{\fancysymbol}$.
    Define $X \coloneqq \bigodot_{i=0}^{m-1} (X_i \cdot \fancysymbol)$.
    Note that $|X| = \Oh(n)$.
    Set $Y = X$ and initialize the algorithm from the lemma statement with $2|X|$, a pair of strings $(X, Y)$, a threshold $n^{\kappa}$, and $w$.
    For every $i \in \fragmentco{0}{m}$, we then do the following sequence of updates: we first make at most $4$ updates transforming the fragment of $Y$ corresponding to $X_i$ into $Y_i$, and then undo these updates to bring $Y$ back to $X$.
    This takes at most $8m = \Oh(n^{\mu})$ updates in total.
    We claim that after transforming the fragment of $Y$ corresponding to $X_i$ into $Y_i$, we have $\wed_{\le \hk}(X, Y) = \wed_{\le \hk}(X_i, Y_i)$, and thus we can decide whether $\wed(X_i, Y_i) \le k_i$.
    Indeed, misaligning any of the $\fancysymbol$ symbols between $X$ and $Y$ makes the alignment cost at least $\hk+1$, and if all $\fancysymbol$ symbols are aligned, the alignment of $X$ onto $Y$ corresponds to a sequence of alignments of $X_j$ onto $X_j$ for all $j \in \fragmentco{0}{m} \setminus \set{i}$ and an alignment of $X_i$ onto $Y_i$.
    It is easy to verify that $\ed(X, Y) \le 4$ holds throughout the lifetime of the algorithm.
\end{proof}

By picking an arbitrary $\kappa \in (0, 1)$ and applying \cref{lm:dynamic-lb-fixed-k} for an arbitrarily small $\delta > 0$ and $\mu = 1 - \kappa - \delta$, we can see that assuming the APSP Hypothesis, there is no algorithm as in \cref{lm:dynamic-lb-fixed-k} with $\Oh(n^{1 + \kappa - 1.5 \delta})$ preprocessing and $\Oh(n^{2 \kappa - \delta / 2})$ update time.
Hence, the algorithm of \cref{lm:simple-complete-algorithm-upgraded} is unlikely to be improved even in the highly restricted setting of \cref{lm:dynamic-lb-fixed-k}.

More generally, by picking some $\gamma \in [0.5, 1)$ and $\kappa \in (0, 1 / (3 - 2 \gamma)]$ and applying \cref{lm:dynamic-lb-fixed-k} for $\mu = 1 - (3 - 2 \gamma) \kappa$ and an arbitrarily small $\delta > 0$, we can see that assuming the APSP Hypothesis, there is no dynamic algorithm as in \cref{lm:dynamic-lb-fixed-k} with $\Oh(n^{1 + \gamma \kappa - \delta})$ preprocessing and $\Oh(n^{(3 - \gamma) \kappa - \delta})$ update time.
In \cref{app:full-algorithm} we give a family of algorithms matching this lower bound for all values of $\gamma$ in a setting that is even more general than the one of \cref{lm:simple-complete-algorithm-upgraded}.

We finish this section by deriving the following theorem from \cref{lm:dynamic-lb-fixed-k}.

\begin{restatable}{theorem}{thmlbmain} \label{thm:lb_main}
    Suppose that \cref{prob:fixed} admits a solution with preprocessing time $T_P(n,k)$ and update time $T_U(n,k)$ for non-decreasing functions $T_P$ and $T_U$.
    If $T_P(n,n^{\kappa}) = \Oh(n^{1+\kappa-\delta})$, $T_U(n,n^{\kappa}) = \Oh(n^{\min(0.5+1.5\kappa,2.5\kappa)-\delta})$, as well as $T_P(n,n^{\kappa})\cdot T_U(n,n^{\kappa})= \Oh(n^{1+3\kappa-\delta})$ hold for some real parameters $\kappa \in (0,1)$ and $\delta > 0$, then the APSP Hypothesis fails.
\end{restatable}
\label{app:dynamic-lower-bounds:lb_main}
\begin{proof}
    Let $\varepsilon \coloneqq \min(\kappa, 1 - \kappa, \delta) / 5$ and $M \coloneqq \setof{\max(1 - 2 \kappa, 0) + \varepsilon c}{c \in \ZZ_{\ge 0}, \max(1 - 2 \kappa, 0) + \varepsilon c < 1 - \kappa}$.
    Consider an arbitrary instance $G$ of Tripartite Negative Triangle with parts of sizes at most $n^{\kappa}$, $n^{\kappa}$, and $n^{1 - \kappa}$.
    Assuming that the algorithm from the theorem statement exists, we will solve this instance of Tripartite Negative Triangle in time $\Oh(n^{1 + \kappa - \varepsilon})$, thus violating the APSP Hypothesis due to \cref{fct:negtri}.
    To achieve it, we iterate over all $\mu \in M$ and run some algorithm for up to $C \cdot n^{1 + \kappa - \varepsilon}$ steps for a sufficiently large constant $C$.
    Each of the algorithms solves the instance $G$ of Tripartite Negative Triangle or returns FAIL if it was interrupted.
    We then prove that for at least one value of $\mu$, the algorithm finishes without being interrupted.

    We now describe the algorithm we run for a fixed value of $\mu \in M$.
    By splitting the third part of $G$ into $n^{(1 - \kappa - \mu) / 2}$ equal-sized subparts, we reduce the initial instance of Tripartite Negative Triangle to $n^{(1 - \kappa - \mu) / 2}$ instances with part sizes at most $n^{\kappa}$, $n^{\kappa}$, and $n^{(1 + \mu - \kappa) / 2}$ each.
    Using \cref{lm:dynamic-lb-fixed-k}, we can reduce these instances to $n^{(1 - \kappa - \mu) / 2}$ instances of the problem from \cref{lm:dynamic-lb-fixed-k} with parameters $n$, $\kappa$, and $\mu$.\footnote{More precisely, the instance produced by \cref{lm:dynamic-lb-fixed-k} has parameter $\Oh(n)$ rather than $n$, but we ignore this issue for the clarity of presentation.}
    We solve each of these instances using the algorithm from the theorem statement.

    Correctness of the algorithm for a fixed value of $\mu$ is trivial.
    It remains to prove that for at least one $\mu \in M$ the algorithm runs in $\Oh(n^{1 + \kappa - \varepsilon})$ time, and thus finishes before being interrupted.
    Let $\rho \coloneqq \min(\log_n T_P(n, n^{\kappa}), 1 + \kappa - \delta)$ and $\nu \coloneqq \min(0.5 + 1.5 \kappa, 2.5 \kappa, 1 + 3 \kappa - \rho) - \delta$.
    (Note that our algorithm does not know the value $T_P(n, n^{\kappa})$, and thus cannot calculate the values $\rho$ and $\nu$. We use them only for the analysis.)
    By the theorem statement, we have $T_P(n, n^{\kappa}) = \Oh(n^{\rho})$ and $T_U(n, n^{\kappa}) = \Oh(n^{\nu})$.
    Let $L \coloneqq \max(0, 1 - 2 \kappa, 2 \rho - 3 \kappa - 1)$ and $R \coloneqq \min(1 - \kappa, 1 + 3 \kappa - 2 \nu)$.
    It is easy to verify that by the lemma statement we have $L, R \in [\max(1 - 2 \kappa, 0), 1 - \kappa]$ and $L + 5\varepsilon \le R$.
    Therefore, by the definition of $M$, there is some $\mu \in M$ satisfying $\mu \in [L + 2\varepsilon, L + 3 \varepsilon) \subseteq [L + 2\varepsilon, R - 2\varepsilon)$.
    For this choice of $\mu$, the algorithm takes $\Oh(n^{(1 - \kappa - \mu) / 2} \cdot (n^{\rho} + n^{\mu} \cdot n^{\nu})) = \Oh(n^{1 + \kappa - \varepsilon})$ time, thus proving the theorem.
\end{proof}

\end{document}